\documentclass{article}
\usepackage[utf8]{inputenc}
\usepackage[T1]{fontenc}
\usepackage[margin=1in]{geometry}
\usepackage{amsmath, amssymb, amsthm, mathtools}
\usepackage{hyperref}
\usepackage{color}
\usepackage[capitalise,nameinlink]{cleveref}

\newtheorem{definition}{Definition}
\newtheorem{lemma}{Lemma}
\newtheorem{theorem}{Theorem}
\newtheorem{corollary}{Corollary}
\newtheorem{remark}{Remark}
\newtheorem{proposition}{Proposition}
\newtheorem{claim}{Claim}

\providecommand{\E}{\mathbb{E}}
\providecommand{\Samp}{\operatorname{Samp}}

\providecommand{\INW}{\operatorname{INW}}

\providecommand{\poly}{\operatorname{poly}}
\providecommand{\polylog}{\operatorname{polylog}}

\providecommand{\norm}[1]{\left\lVert #1\right\rVert}

\providecommand{\asvapprox}{\mathrel{\approx}^{\mathrm{sv}}}

\providecommand{\rk}{\operatorname{rk}}
\providecommand{\Bcl}{\mathcal{B}}
\providecommand{\Rcl}{\mathcal{R}}
\providecommand{\Pcl}{\mathcal{P}}

\providecommand{\one}{\mathbf{1}}
\providecommand{\Edge}{\mathcal{E}}
\providecommand{\LevFB}{\Lambda_{\mathrm{FB}}}
\newcommand{\R}{\mathbb{R}}

\newcommand{\Rot}{\mathsf{Rot}}

\providecommand{\E}{\mathbb{E}}
\providecommand{\SD}{\operatorname{SD}}

\providecommand{\SC}{\mathsf{SC}}
\providecommand{\DTISP}{\mathsf{DTISP}}

\providecommand{\INW}{\mathrm{INW}}
\providecommand{\Samp}{\operatorname{Samp}}

\providecommand{\cut}{\operatorname{cut}}

\providecommand{\rk}{\operatorname{rk}}

\newcommand{\knote}[1]{{\color{cyan} (Kuan: #1)}}

\definecolor{revision}{RGB}{45,90,150}

\newif\ifhighlightnewresults
\highlightnewresultstrue
\newenvironment{newresulttext}
  {\par\begingroup\ifhighlightnewresults\color{revision}\fi}
  {\par\endgroup}

\title{SC Derandomization for Regular ROBPs and Models Beyond BPL}
\author{}
\date{}

\begin{document}

\author{ 
Kuan Cheng \footnote{CFCS, School of Computer Science, Peking University. ckkcdh@pku.edu.cn.}
\and 
Ruiyang Wu \footnote{CFCS, School of Computer Science, Peking University. wuruiyang@stu.pku.edu.cn.}
}

\date{}
\maketitle

\begin{abstract}
We study SC derandomizations for regular read-once branching programs (ROBPs) and computation models beyond BPL.

For regular ROBPs with length $n$,  width $w$, and multiple accept nodes, we attain the following results. 
\begin{itemize}
    \item When $n \le w$,  we show an SC derandomization with space $O(\log^2 n+\log w)$ and two-sided error $1/\poly(nw)$.
    \item When $n \ge w$, we show an SC derandomization with space $O(\log n \log w)$ and two-sided error $1/\poly(w)$. In addition, when $w=O(\log n)$, we attain an optimal $O(\log n)$ space derandomization with two-sided error $1/\poly(w)$. 
    \item When $w \le 2^{O(\sqrt{\log n})}$, we show that reachability of regular ROBPs (i.e. derandmization of one-sided unbounded small error ROBPs) can be computed in SC.
\end{itemize} 
The 1st result can be viewed as extending the optimal derandomization of medium length ($n = 2^{O(\sqrt{\log w})}$) regular ROBPs of Cheng and Wu (SODA'26) to the case $2^{O(\sqrt{\log w})}  << n \le w $.
For $n \ge 2^{O(\sqrt{\log w})} $, to our knowledge, the only previous SC derandomization is still the classic SC derandomization of Nisan (STOC'92) which has space $O(\log n \log (nw))$.
The space complexities of our 1st and 2nd results are significantly below this bound when $n<< w$ and $w << n$ respectively.
For the context of the 3rd result, we are unaware of any previous results showing reachability belongs to SC unless with strong promises, such as the work by Buntrock, Jenner, Lange, Rossmanith (FCT'91), and the work by Lange (STACS'97). We indicate that when $w = n^{1/c}$ for any constant $c$, the reachability for regular ROBPs is NL-complete. So pushing forward this result to larger $w$ leads directly to NL $\subseteq$ SC. We also show that when regular ROBPs are powering, optimal deterministic logspace can be attained for computing reachability and for bounded two-sided error derandomization, when $w = n^{1/c}$ for some constant $c$. 
Cohen, Doron, Sberlo, and Ta-Shma (STOC'23) achieves optimal logspace computation for general matrix powering when $w = 2^{O(\sqrt{\log n})}$ previously.

We further show that two super sets of BPL can be computed in SC.
\begin{itemize}
    \item For  probabilistic logspace TMs with a two-way access random tape, we show that it can be approximated in SC if each entry of the random tape is accessed for at most a constant number of times.
    \item For probabilistic logspace TMs with a polynomial size stack, i.e. probabilistic logspace Auxiliary Push-down Machines (AuxPDMs), we show that it can be approximated in SC if the timings of push/pop/idle stack operations do not depend on the randomness.
\end{itemize}
The first model is the read-multiplicity model considered by Impagliazzo, Nisan, Wigderson (STOC'94), in which they show that their INW generator can fool such computations. For the second model, we indicate that it contains candidate languages separating BQL from BPL considered by Apers and Edenhofer (CCC'25). 

\end{abstract}

\newpage
\tableofcontents
\newpage

\section{Introduction}

The complexity class SC$^{i}$ is the set of languages that can be decided by Turing Machines running in time $\poly(n)$ and simultaneously $O(\log^i n)$ space. The union of such classes is denoted as SC (Steve's Class), i.e.  $\mathrm{SC}= \bigcup_{i}\mathrm{SC}^i$, i.e. languages in this class can be decided by TMs in polynomial time and simutaneously poly-logarithmic space. 
SC is first defined and formally studied in the seminal work of Cook\cite{cook1979deterministic}, which shows that logspace TMs with an auxiliary stack (also denoted as stack machines or AuxPDMs), with polynomial running time, can be simulated in SC. Pippenger \cite{pippenger1979simultaneous} first formally used the name Steve's Class, where ``Steve'' is short for the first name of Stephan Cook. 
One immediate question after \cite{cook1979deterministic} is whether SC can compute some more powerful languages.

BPL (resp. RL), short for Bounded-error Probabilistic Logspace (resp. Randomized Logspace), is the set of languages which can be decided by probabilistic TMs running in polynomial time and $O(\log n)$ space, with two-sided errors (resp. one-sided errors), where the randomness is read-once.
The definitions of BPL and its related classes are originally given in \cite{aleliunas1979random}.
For derandomizing space bounded computation, the seminal work \cite{nisan1992rl}  of Nisan shows that RL and BPL are in SC.
\cite{nisan1992rl}  crucially uses the classic Nisan's PRG \cite{nisanPseudorandomGeneratorsSpacebounded1992}. 
\begin{definition}[PRG]
    Let $\mathcal{F}$ be a class of space bounded Computations $B:(\{0,1\}^s)^n\to \{0,1\}$.
    An $\varepsilon$-PRG for $\mathcal{F}$ is a function  $G:\{0,1\}^d\to (\{0,1\}^s)^n$
    such that for every $B\in \mathcal{F}$, we have
    \begin{equation*}
        \left|\E_{x\in (\{0,1\}^s)^n}B(x)-\E_{r\in \{0,1\}^d}B(G(r))\right|\leq \varepsilon.
    \end{equation*}
    The input length $d$ is called the seed length of the PRG. We say that $G$ is explicit if it can be computed in space $O(d)$ and time $\poly(d, n)$.
\end{definition}
When fixing the input, the computation of a space bounded PTM can be viewed as a Read-once Branching Program (ROBP).
\begin{definition}[Read-once branching programs (ROBPs)]
    An ROBP $B$ of length $n$, width $w$ and alphabet size $|\Sigma|=2^s$ is a directed acyclic graph with $n+1$ layers $V_0,\ldots,V_n$. For any layer $V_i$ except $V_n$, each node $v\in V_i$ has $2^s$ outgoing edges to nodes in $V_{ i+1}$. These edges are labeled by distinct symbols in $\Sigma$. There exists a unique start node $v_{start} \in V_0$ and a set of accept nodes $V_{accept}\subset V_n$. Given an input $x\in \Sigma^n$, the computation of $B(x)$ is defined as: 
    $B(x) = 1$, if there exists a unique path $v_{start},v_1,\ldots,v_n$ such that the edge between $v_i$ and $v_{i+1}$ is labeled by $x_i$, and $v_n\in V_{accept}$;  $B(x) = 0$ otherwise.
\end{definition}
A long history line of work, including but not limited to \cite{ajtai1987deterministic, nisanPseudorandomGeneratorsSpacebounded1992, impagliazzoPseudorandomnessNetworkAlgorithms1994, nisanRandomnessLinearSpace1996, armoniDerandomizationSpaceBoundedComputations1998,  razRecyclingRandomnessStates1999, reingold2006pseudorandom, bravermanPseudorandomGeneratorsRegular2014, ganor2014space, hoza2021pseudorandom, ChenTaShma2025Armoni, cohenDoronGoldgraberForwardBackward2026, HozaShalunov2026BRRY}, studies PRGs against ROBPs. 
Recently two important subclasses of ROBPs i.e. Regular and Permutation ROBPs have attracted widespread attention.
\begin{definition}[Regular ROBP]
    A regular ROBP is a standard-order ROBP $B$, where for every $i\in [n]$, the bipartite graph induced by the nodes in layers $V_{i-1}$ and $V_i$ is a regular graph i.e. every node has the same degree.
\end{definition}
\begin{definition}[Permutation ROBP]
    A (standard-order) permutation ROBP is a standard-order ROBP $B$, where for every $i\in [n]$ and $x\in \{0,1\}^s$, the transition matrix from $V_i$ to $V_{i+1}$ through edges labeled by $x$, is a permutation matrix in $\mathbb{R}^{w\times w}$.
\end{definition}
We mention that Regular ROBPs correspond to the case that transition matrices are doubly-stochastic.
Permutation ROBPs are special regular ROBPs in the sense that there are no collisions in labels of edges from previous layers.
A recent line of work \cite{koucky2011pseudorandom, steinke2012pseudorandomness, de2011pseudorandomness, bravermanPseudorandomGeneratorsRegular2014, braverman2019pseudorandom, chengHittingSetsGive2020, chattopadhyayOptimalErrorPseudodistributions2020, hoza2020simple, cohenErrorReductionWeighted2021, hozaBetterPseudodistributionsDerandomization2021, pyne2021pseudodistributions, pyneHittingSetsRegular2021, hoza2021pseudorandom, chengWuWeightedPseudorandomGenerators2026, chen25better, chenImprovedErrorReduction2026} specifically studies PRGs for regular or permutation ROBPs, Weighted Pseudorandom Generators (WPRGs), and Hitting Sets Generators (HSGs) for ROBPs.
\begin{definition}[WPRGs, first defined and constructed by \cite{braverman2019pseudorandom}]
    Let $\mathcal{F}$ be a class of ROBPs $B:(\{0,1\}^s)^n\to \{0,1\}$. A $W$-bounded $\varepsilon$-WPRG for $\mathcal{F}$ is a function $(G,\sigma):\{0,1\}^d\to(\{0,1\}^s)^n\times \mathbb{R}$ such that for every $B\in \mathcal{F}$, we have
    \begin{align*}
        &\left|\E_{x\in(\{0,1\}^s)^n}B(x)-\sum_{r\in \{0,1\}^d}\left[\frac{1}{2^d}\sigma(r)\cdot  B(G(r))\right]\right|\leq \varepsilon,\\
        &\forall r, |\sigma(r)|\leq W.
    \end{align*}
    The input length $d$ is called the seed length of the WPRG. We say that $(G,w)$ is  explicit  if it can be computed in space $O(d)$ and time $\poly(d, n)$.
\end{definition}
\begin{definition}[HSG]
     Let $\mathcal{F}$ be a class of ROBPs $B:(\{0,1\}^s)^n\to \{0,1\}$. A HSG for $\mathcal{F}$ is a function $(G):\{0,1\}^d\to(\{0,1\}^s)^n$ such that for every $B\in \mathcal{F}$, we have if $\Pr[B(x) = 1]\ge \varepsilon,$ then there is $r\in \{0, 1\}^d$ such that $B(G(r)) = 1$.  
    The input length $d$ is called the seed length of the WPRG. We say that $(G,w)$ is explicit if it can be computed in space $O(d)$ and time $\poly(d, n)$.

\end{definition}
Based on these recent progress, improvements for several types of derandomizaions of space bounded computation have already been achieved. For example the classic Saks and Zhou derandomizations \cite{saksBPHSPACEDSPACES31999} have been improved in several aspects. Ahmadinejad, Kelner, Murtagh, Peebles, Sidford, and Vadhan \cite{ahmadinejadHighprecisionEstimationRandom2022} showed that Regular ROBPs can be derandomized in near optimal small space $O\left(\log (nw) \log \log (nw)\right)$.  Hoza \cite{hozaBetterPseudodistributionsDerandomization2021} showed that standard ROBPs can be derandomized in space $\mathrm{BPL}\subseteq \mathrm{DSPACE}\left(\frac{\log^{3/2}n}{\sqrt{\log\log n}}\right)$. 
Cohen, Doron, Sberlo, and Ta-Shma \cite{cohenApproximatingIteratedMultiplication2023}, also Pyne and Putterman \cite{puttermanOptimalDerandomizationMediumWidth2022},  showed that standard ROBPs can be derandomized in $O\left( \left(\log n + \sqrt{\log n} \log w \right) \log \log n\right)$ space. 
 \cite{cohenApproximatingIteratedMultiplication2023} also shows that stochastic matrix powering can be done in space $O(\log n + \log^2 w)$ which is essentially optimal when $w = 2^{O(\sqrt{\log n})}$.
For the classic Nisan-Zuckerman derandomization \cite{nisanRandomnessLinearSpace1996} which requires optimal space, Cheng and Wu \cite{chengWuWeightedPseudorandomGenerators2026} shows that  regular ROBPs with medium length $n = 2^{O(\sqrt{\log w})}$ can be derandomized optimally in logspace.
Also for derandomizing BPL to the circuit hierarchy,
Cheng and Wang \cite{cheng2024bpl} showed that $\mathrm{BPL}$  is contained in logspace uniform $ \mathrm{AC}^1$.
However, it is still an interesting question, whether one can give improved SC derandomizaions for certain kinds of ROBPs.

On the other hand, whether SC can simulate models of computations beyond BPL is widely open.
Perhaps one such candidate is the model considered by the seminal work by Inpagliazzo, Nisan, Wigderson \cite{impagliazzoPseudorandomnessNetworkAlgorithms1994}.  \cite{impagliazzoPseudorandomnessNetworkAlgorithms1994} gives a generator (the INW generator) which can fool a kind of multiparty computations. A very interesting sub-case is the read-$k$ model, where the adversary can read every entry of the random string for at most $k$ times. 
Probabilistic logspace AuxPDM is another interesting model beyond BPL.
\begin{definition}[Probabilistic AuxPDMs]
   A probabilistic logspace AuxPDM is a logspace TM with an auxiliary unbounded stack, and an extra read-once random tape. Each step of the machine can push/pop one symbol to/from the top of the stack, or do nothing to the stack (we call this an idle operation on stack, or just idle for simplicity). 
\end{definition}
This model was first considered by Macarie and Ogihara \cite{MacarieOgihara1998}.
Vencateswaran \cite{venkateswaran2006derandomization} claims that Probabilistic logspace Auxiliary Push-down Machines with bounded two-sided error can be computed in SC.
However, as indicated in a later work by Allender and Lange\cite{allender2010symmetry}, the proof of \cite{venkateswaran2006derandomization} is incomplete.
So whether this model can be computed in SC is still an open question.

Another major open question is whether NL is constained in SC. This problem can also be viewed as a derandomization problem since one can view this as derandomizing one-sided error probabilistic logspace Turing Machines with unbounded small error. 
Previous results on this direction can only handle instances with ``few path'' promises. 
Specifically, the work Buntrock, Jenner, Lange, Rossmanith \cite{buntrock1991unambiguity} and the work by Lange \cite{lange1997unambiguous} show that if all pair of nodes have polynomial number of paths inbetween them, or the number of paths starting from the start and the number of paths reaching the acceptance are both polynomial, then the connectivity problem is in SC.
We note that for regular or standard ROBPs, especially when $n$ is large, the number of paths can be exponentially large.
We are unaware of any previous results which can get rid of these promises.

\subsection{Our Results}
We give improved SC derandomizations for regular ROBPs and further show that several super classes of BPL is contained in SC.
\paragraph{SC derandomization for regular ROBPs with bounded error} 

Our first result is an SC derandomization for short-wide regular ROBPs.
\begin{theorem}
[SC Derandomization for short-wide regular ROBPs]
\label{thm:regular-derand-low-error}
Let $n,w\in \mathbb{N}$. Let $\Sigma$ be a finite alphabet. For any $\varepsilon = 1/\poly(nw)$. There exists a deterministic algorithm that, given a length-$n$ width-$w$ regular ROBP $B$ over alphabet $\Sigma$, computes an approximate transition matrix $\widehat{W}$ such that $\|W  - \widehat{W} \|_1\leq \varepsilon$, where $W$ is the transition matrix of $B$. The algorithm runs in time $\poly(n,w,|\Sigma| )$ and uses space
\[
    O(\log|\Sigma|+\log^2 n+ \log w).
\]
\end{theorem}
Notice that this improves the SC space from $O(\log n \log (nw))$ of \cite{nisan1992rl} to $O(\log^2 n + \log w)$, getting rid of the main term $O(\log n \log w)$ when $w>>n$.

The proof of \cref{thm:regular-derand-low-error} relies on our next result for SC derandomization of Regular ROBPs with single accept nodes.
\begin{theorem}
\label{thm:regular-derand}
Let $n,w\in \mathbb{N}$. Let $\Sigma$ be a finite alphabet. Let $\varepsilon>0$. There exists a deterministic algorithm that, given a length-$n$ width-$w$ regular ROBP $B$ over alphabet $\Sigma$, computes an approximate transition matrix $\widehat{W}$ such that $|[W_{0..n}]_{u,v}-[\widehat{W}]_{u,v}|\leq \varepsilon$ for every $u,v\in [w]$,where $W$ is the transition matrix of $B$. The algorithm runs in time $\poly(n,w,|\Sigma|,1/\varepsilon)$ and uses space
\[
    O(\log|\Sigma|+\log n\cdot (\log\log n+\log(1/\varepsilon))+\log w).
\]
\end{theorem}

Notice that one can take $\varepsilon = O(1/w)$ in \cref{thm:regular-derand} to attain an SC derandomization for multiple accept nodes. In space this is even better than \cref{thm:regular-derand} when $w < n$. In fact, for long narrow regular ROBPs, we can also attain the following SC derandomization. 
\begin{theorem} 
\label{thm:perm-regular-main}
Let \(B\) be a length-\(n\), width-\(w\) regular ROBP over \(\Sigma\).
For every \(0<\varepsilon<1/2\), there is a deterministic algorithm that
computes, entry by entry, a matrix \(\widehat W_{0..n}\) satisfying
\(
   \left\|[\widehat W_{0..n}-W_{0..n}\right\|_1
   \le\varepsilon
   \text{ for every }u,v\in[w].
\)
The algorithm runs in time \(\poly(n,w,|\Sigma|,1/\varepsilon)\) and space
\[
   O\!\left(\log n\log\left({w|\Sigma|}/{\varepsilon}\right)
   \right).
\]
\end{theorem}
Notice that essentially the above theorems can improve the SC space from $O(\log n \log (nw))$ of \cite{nisan1992rl} to $O(\log n \log w)$ i.e. getting rid of the main term $O(\log^2 n)$ when $w << n$.

When the width is logarithmic, i.e. $w=O(\log n)$, we attain optimal
$O(\log n)$ space derandomization in polynomial time.
\begin{theorem}
\label{thm:narrow-evolving-main}
Fix a constant $C>0$. Let $B$ be a length-$n$, width-$w$ binary
regular ROBP. For every $w^{-C}\le\varepsilon<1/2$, there is
a deterministic algorithm that, given the transition matrices of $B$,
outputs a matrix $\widehat W_{0..n}$ satisfying
\[
   \|\widehat W_{0..n}-W_{0..n}\|_1\le\varepsilon
\]
in space $O(\log n+w)$ and time
$2^{O(w)}\poly(n,w)$. 
\end{theorem}

\paragraph{Reachability in regular ROBPs.}
For one-sided error regular ROBPs, we can attain a much stronger result. Our next result shows that for medium width regular BOBPs, even with unbounded small error, i.e. the reachability problem, can be decided in SC. 

\begin{theorem}[Reachability in regular ROBPs]
\label{thm:regular-positivity-main}
Fix a constant $C>0$. There is a deterministic algorithm that, given a binary regular ROBP of length $n$ and width $w$ with $\log^2 w\le C\log n$, decides whether it accepts any input in time $\poly_C(n,w)$ and space
\[
   O_C\!\left(\log n+
      \frac{\log n\log w}
           {\max\!\left\{1,\log\!\left(\frac{\log n}{\log^2 w}\right)\right\}}\right).
\]
\end{theorem}

Notice that this attains optimal logarithmic space for $w=(\log n)^{O(1)}$. For $w=2^{O(\sqrt{\log n})}$, it attains an SC algorithm with $O((\log n)^{3/2})$ space. 

For every fixed integer $c\ge1$, reachability in binary regular ROBPs is $\mathrm{NL}$-complete for $w\ge n^{1/c}$ (\Cref{rem:regular-reachability-nl}). An SC algorithm for such width bounds would therefore imply $\mathrm{NL}\subseteq\mathrm{SC}$.

\paragraph{Powering of regular transition matrices.}
For binary regular ROBPs with identical layers (a.k.a regular matrix powering), we obtain a logarithmic-space reachability algorithm for a much larger range of $w$ (\Cref{app:regular-powering}).

\begin{theorem}[Reachability for regular matrix powers]
\label{thm:regular-powering-positivity-main}
Let $A$ be a binary regular transition matrix on $w$ states, with start state $s$ and accepting set $F$. Given $A,s,F$, and a $b$-bit integer $n\ge w^2-1$, a deterministic algorithm decides whether $\mathbf1_F^\top A^ne_s>0$ in space $O(\log(w+b))$ and time $\poly(w,b)$.
\end{theorem}
This gives a logarithmic-space algorithm for $w\le\sqrt n$. For $n=O(w)$, we show that the reachability problem for this class is $\mathrm{NL}$-complete (\Cref{rem:regular-powering-nl}).

For $w \approx n^{1/3}$, our next result shows that approximating the acceptance probability of regular matrix powering can be done with optimal logarithmic space.
\begin{proposition}[Approximation of regular matrix powers]
\label{thm:regular-powering-main}
Given $A,s,F$ as above and a $b$-bit integer $n\ge(w-1)w^2$, a deterministic algorithm computes the numerator and denominator of
\[
   q_n=\frac{|F\cap R_n|}{|R_n|},
   \qquad
   R_n=\{v\in[w]:(A^n)_{v,s}>0\},
\]
in space $O(\log(w+b))$ and time $\poly(w,b)$. For every integer $k\ge1$ such that $n\ge k(w-1)w^2$,
\[
   \left|\mathbf1_F^\top A^ne_s-q_n\right|\le(5/48)^k.
\]
\end{proposition}
We note that \cite{cohenApproximatingIteratedMultiplication2023} attains a Cayley-Hamilton theorem based algorithm for powering of standard stochastic matrices which attains optimal logspace for exact computing of acceptance path counting when $w = 2^{O(\sqrt{\log n})}$.
Our results indicate that optimal logspace derandomization can be attained for a much larger range of $w$, for both positvity and the bounded two-sided error case.

\paragraph{SC derandomization for models beyond BPL.}
For the case $w = n$ (perhaps the most difficult case), we show that in SC$^2$, one can actually simulate two larger classes of computations.


\begin{theorem}[SC derandomization of read-$p$ machines]
\label{thm:read-p-sc-main}
Fix a constant $p\ge1$. Let $A$ be a language decided with
bounded error by a polynomial-time probabilistic TM
using $O(\log n)$ workspace and a two-way read-only random tape.
Suppose that each random-tape cell is visited at most $p$ times
in every computation, where consecutive steps at the same cell
count as a single visit. Then $A\in\mathrm{SC}^2$.

Specifically the acceptance
probability of the probabilistic TM can be approximately computed to within additive error $\varepsilon$ in time
$\poly(n,1/\varepsilon)$ and space
$O\!\left(\log^2 n+\log n\log(1/\varepsilon)\right).$
\end{theorem}
Notice that this generalizes the SC derandomization of \cite{nisan1992rl} for standard BPL to an SC derandomization for the read-multiplicity model of \cite{impagliazzoPseudorandomnessNetworkAlgorithms1994}. 

\begin{theorem}[SC derandomization of machines with oblivious stacks]
\label{thm:oblivious-stack-sc-main}
Let $A$ be a language decided with bounded error by a
polynomial-time probabilistic logspace AuxPDMs
with the restriction that for each input, the timings of its stack operations (push/pop/idle) do not depend on the randomness.
Then $A\in\mathrm{SC}^2$.

Specifically the acceptance
probability of the probabilistic AuxPDMs can be approximately computed to within additive error $\varepsilon$ in time
$\poly(n,1/\varepsilon)$ and space
$O\!\left(\log^2 n+\log n\log(1/\varepsilon)\right).$
\end{theorem}
Notice that if without the restriction, i.e. timings of stack operations can depend on the randomness, then this is exactly the model considered in \cite{venkateswaran2006derandomization}. As the proof of \cite{venkateswaran2006derandomization} is incomplete \cite{allenderLange2014symmetry}, our result can be viewed as positively answering the open question on whether Probabilistic logspace AuxPDMs can be simulated in SC, in a slightly restricted form. 

We further remark that though with such a restriction, the model can compute certain problems that are regarded as candidates to separate some more powerful classes from BPL.
For example the sparse s-t connectivity problem
STCON$_{sf} = \left\{ \langle G, s, t, 1^k \rangle \mid \forall i, j\in V(G), N(i, j)\le k, N(s,t)\ge 1   \right\}$ is a problem considered in \cite{ApersEdenhofer2025}, where $N(i, j)$ is a function counting the number of paths between $i$ and $j$. It is regarded by \cite{ApersEdenhofer2025} as a candidate separating Bounded-error Quantum logspace (BQL) from BPL. We show that STCON$_{sf}$ can be computed by deterministic logspace AuxPDMs (see \cref{app:stcon-sf} for proofs) and hence also computable by our model. (As logspace AuxPDMs is already in SC \cite{cook1979deterministic}. This also positively answers one open question of \cite{ApersEdenhofer2025} on whether STCON$_{sf}$ can be computed by SC.)

\subsection{Technical Overview}
\label{sec:technical-overview}

We show our SC derandomizations for long-narrow regular ROBPs first. We have two theorems \cref{thm:regular-derand} and \cref{thm:perm-regular-main} for this case. They share a similar algorithm structure. We show in \cref{sec:reg-to-perm-phase} that regular ROBPs can be transformed to permutation ROBPs implicitly in logspace. For these two constructions, we therefore focus on permutation ROBPs.

\paragraph{An $m$-ary Structure}

The general algorithm structure for the long-narrow case is as the following.
Consider a permutation ROBP of length $n$, width $w$ and alpbabet $\Sigma$. Our algorithm has $K = \lceil\log n/\log m\rceil$ stages.
For each stage, for every consecutive $m$ transitions, consider applying an INW generator. Then this can be viewed as a shorter program where each transition consumes a seed of an INW generator. 
The program length is reduced by dividing a factor of $m$ while the alphabet size is increased to be the support size of the INW seed. 
To control this growth of alphabet, we use an averaging sampler to reduce the alphabet, i.e. we deterministically search a common offline seed for all the INW generators such that the transition probabilities can be approximated ``good-enough''. We will explain later what is called ``good-enough'' and how to check it, for two different approaches corresponding to \cref{thm:regular-derand} and \cref{thm:perm-regular-main} respectively. 
Here we only focus on the general structure.
As we start with permutation ROBPs, after using the sampler, the program remains a permutation ROBP. 
We store the offline seed for this stage. 
The shortened program is left to the next stage until the length becomes a constant.

The space bound depends on how many stages are needed which essentially depends on $m$.
We choose $m$ as large as possible, while also ensuring that the offline seed must be enumerable within polynomial time. 
Let $d$ be the length of such a seed.
The space used to store the offline seeds sums up to
\[
   O\!\left(\log|\Sigma|
      +d K\right) =   O\!\left(\log|\Sigma|
      +d\left\lceil\frac{\log n}{\log m}\right\rceil\right).
\]
We also need to ensure the inner seed length of these samplers to be logarithmic so that
given the offline seeds of the samplers, enumerating all
values of a inner seed consumes polynomial time. 
We need this property to check whether the approximations are ``good-enough''.

\paragraph{Converting \cite{hoza2021pseudorandom} to a SC derandomization, separately tracking SV and Euclidean errors.}
We use the singular value analysis of
\cite{hoza2021pseudorandom,
chenWeightedPseudorandomGenerators2023},
together with the SV approximation framework of
\cite{ahmadinejad2023singular}.
The corresponding guarantee for large alphabets is stated in
\cite[Lemma~4.6]{chengWuWeightedPseudorandomGenerators2026}.
We extend this analysis to accommodate the error
introduced by sampling.
For doubly stochastic matrices $W$ and $\widehat W$, a
$\varepsilon$-SV approximation $\widehat W$ of $W$ implies
$$
   |v^\top(\widehat W-W)u|
   \le \frac{\varepsilon}{4}
   \left(
      u^\top(I-W^\top W)u
      +v^\top(I-WW^\top)v
   \right)
   \qquad\text{for all }u,v\in\mathbb R^w.
$$
In contrast, the sampler's entrywise guarantee gives $\|\widehat W-W\|_2\le\delta$, and hence
$
   |v^\top(\widehat W-W)u|
   \le \delta\|u\|_2\|v\|_2.
$
Even for unit vectors $u,v$, the sum of the two quadratic forms
in the SV bound can be arbitrarily small.
Thus we don't know how to absorb the Euclidean error into the required
SV guarantee, regardless of how small $\delta$ is compared to $\varepsilon$. Instead, we track the sampling error separately. For a doubly stochastic matrix $W$, we say $\widehat W$ is an $(\varepsilon,\delta)$-SV approximation of $W$ if
\[
\begin{aligned}
 |v^\top(\widehat W-W)u|
 &\le \frac{\varepsilon}{2}
   \bigl(\|u\|_{R(W)}^2+\|v\|_{L(W)}^2\bigr)\\
 &\quad+\frac{\delta}{2}
   \bigl(\|u\|_2^2+\|v\|_2^2\bigr)
 \qquad\text{for all }u,v\in\mathbb R^w,
\end{aligned}
\]
where $R(W)=I-W^\top W$, $L(W)=I-WW^\top$, and
$\|u\|_A^2=u^\top Au$.

We augment the estimates of matrix product and transitivity rule of SV-approximations
\cite{ahmadinejad2023singular} to accommodate this two-parameter SV-approximation. 
The singular-value and Euclidean errors accumulate essentially independently, up to small mixed terms. 
This allows us to choose their target accuracies separately: the INW generator accuracy is set as in \cite{hoza2021pseudorandom}, while the Euclidean error is made much smaller.

For standard $\varepsilon/\polylog(n)$-SV error, by the analysis of \cite{hoza2021pseudorandom}, the INW generator for a block of
$m$-transitions uses
$$
O\!\left(\log m\bigl(\log\log n+\log(1/\varepsilon)\bigr)\right)
$$
seed bits. 
We can show that after using a $( \delta/\poly(nw),  1/\poly(nw))$-averaging sampler, with offline seed length
$d=\Theta(\log(nw/(\varepsilon\delta)))$, we can attain an $ (\varepsilon, \delta) $ two parameter SV approximation. 
In addition, we can take $\log m=\Theta\left(\frac{\log(nw/(\varepsilon\delta))}{\log\log n+\log(1/\varepsilon)}\right)$ under this parameter setting.
Hence one can check the overall space for all stages is
$$
   dK = O\!\left(
      \log|\Sigma|
      +\log n\bigl(\log\log n+\log(1/\varepsilon)\bigr)
      +\log\frac{nw}{\delta}
   \right).
$$

\paragraph{Converting \cite{cohenDoronGoldgraberForwardBackward2026} into an SC derandomization by incorporating sampling into the forward--backward analysis.}
For \cref{thm:perm-regular-main}, we use the forward--backward
analysis of Cohen, Doron, and Goldgraber
\cite{cohenDoronGoldgraberForwardBackward2026}.
To make this construction compatible with our SC simulation, we
insert averaging samplers between applications of INW generators.

We extend the forward--backward analysis to arbitrary finite
alphabets and show that, for the sampling steps in our construction,
entrywise error $\alpha$ implies forward--backward error at most
$|\Sigma|^2w^4\alpha$. Here $\Sigma$ is the original alphabet,
and the loss is independent of the interval length. 
We also extend the INW merge analysis to allow the child matrices
to contain sampling error. If their forward--backward errors are
at most $\eta_L$ and $\eta_R$, and the expander has second singular
value at most $\lambda_{\exp}$, then
\[
   \eta_{\mathrm{new}}
   \le
   \max\!\left\{
      \eta_L,\eta_R,\,
      \lambda_{\exp}|\Sigma|^2+16w^3\eta_L\eta_R
   \right\}.
\]
For a sufficiently small error bound $\eta=O(\varepsilon/w^3)$,
we can choose $\lambda_{\exp}$ so that
$\lambda_{\exp}|\Sigma|^2+16w^3\eta^2\le\eta$.
Consecutive INW merges then preserve this bound.

The required expander accuracy depends only on $w$, $|\Sigma|$,
and $\varepsilon$, so each INW level adds
$O(\log(w|\Sigma|/\varepsilon))$ seed bits.
With seed budget $d=\Theta(\log(nw|\Sigma|/\varepsilon))$, we
choose $m\le n$ as large as possible within this budget, giving
\[
   \log m=\Theta\!\left(
      \min\!\left\{
         \log n,\,
         \frac{\log(nw|\Sigma|/\varepsilon)}
              {\log(w|\Sigma|/\varepsilon)}
      \right\}
   \right).
\]
The sampler searches and online seed enumeration take polynomial
time. Storing the offline seeds over
$K=\lceil\log n/\log m\rceil$ stages therefore takes space
\[
   O\!\left(\log|\Sigma|+dK\right)
   =
   O\!\left(\log n\log\frac{w|\Sigma|}{\varepsilon}\right).
\]
The conversion to regular ROBPs preserves this bound, giving
\Cref{thm:perm-regular-main}.

\paragraph{The case $w=O(\log n)$.}
Our proof for \Cref{thm:narrow-evolving-main} is inspired by the evolving-set method
of Morris and Peres \cite{morrisPeresEvolvingSets2005}.
Starting from the final layer, we process layer by layer until we reach the first layer.
At the beginning we take the set $S_n$ to be the accepting set of the final layer.
At each layer $i$, let $A$ be consisted of layer $i-1$ states whose two
successors belong to the set $S_i$. 
Let $B$ be those with at least one successor in the set. 
If $A=B$, we set $S_{i-1}=A$;
otherwise, we use a fair random bit to choose $S_{i-1}=A$ or $S_{i-1}=B$.
The probability that the start state belongs to the $S_0$ we maintained, is an
unbiased estimation of the acceptance probability of the ROBP.

Regularity makes the $|S_n|, |S_{n-1}|, \ldots, |S_0| $ a martingale. If at layer $i$, set $S_i$ has size $x$ and $d=(|B|-|A|)/2$, for potential function  $H(x)=x(w-x)$, we have $\E[H(|S_{i-1}|)\mid S_i]\leq H(|S_i|)-d^2$. Notice that $\forall x, H(x)\leq w^2/4$. 
Every time we use a random bit to choose between $A$ and $B$, notice that it has to be $d\ge1$.
So the expected number of the random choices is at most $w^2/4$. This bound allows us to truncate the simulation to use at most $O(w^2\log(1/\varepsilon))$ random bits.
The whole process only pays space for storing a subset of $[w]$ and a layer index.
So the total space is $S=O(\log n+w)$. 
For $1/\poly (w)$ accuracy, it is sufficient to use $\poly(S)$ random bits.
We show that Nisan--Zuckerman generator
\cite{nisanRandomnessLinearSpace1996} can be used here.
So we only need an $O(S)$-bit seed to simulate this randomized computation. 
Enumerating those seeds takes space $O(S)$ and time $2^{O(w)}\poly(n,w)$. 

\paragraph{SC derandomization for  short-wide  regular programs} We show \Cref{thm:regular-derand-low-error} using recursive error reductions.
We first construct
base approximations with inverse-polynomial accuracy in $n$, using
\Cref{thm:regular-derand}, and
then apply the recursive error reduction of Chattopadhyay and Liao
\cite{chattopadhyayRecursiveErrorReduction2023}.

For each dyadic interval $J$, let $W_J$ denote the average
transition matrix of the original program on $J$. We obtain a base approximation
$\widehat W_J^{(0)}$ for every nonunit interval, while unit intervals
are represented exactly. We recursively construct order-$h$ approximation matrices given approximation matrices of lower orders. Order $0$ is the
base approximation, and the order-$h$ approximation $W_{J}^{(h)}$ has an error of roughly
$n^{-(h+1)}$. Therefore, taking
$$
   k=\left\lceil\frac{\log(1/\varepsilon)}{\log n}\right\rceil
$$
gives error at most $\varepsilon$.

Our base approximations have both SV and Euclidean error.
To start the error reduction, the initial SV error is required to be smaller than $1/\poly (\log n)$, and the Euclidean error smaller than $1/\poly(n)$. We set the parameters as such.

The approximation $W_{J}^{(h)}$ appears as a weighted sum of products of matrices.
Expanding it gives, by
\cite[Lemma~25]{chattopadhyayRecursiveErrorReduction2023}, a signed
sum of $\poly(n,1/\varepsilon)$ products, each containing
$$
   O(\log n+\log(1/\varepsilon))
$$
order-$0$ matrices.

It remains to evaluate the corrected approximation in polynomial
time. We view each
such product as the average transition matrix of a short permutation
ROBP and approximate it using the permutation WPRG of
\cite{chengWuWeightedPseudorandomGenerators2026}. The expansion
terms and the corresponding weighted seeds can then be enumerated
in polynomial time, giving space
$$
   O\!\left(
      \log|\Sigma|+\log^2 n+\log(1/\varepsilon)+\log w
   \right).
$$

\paragraph{Reachability in regular ROBPs.}
We managed to show that there is an efficient relabeling algorithm in logspace, which, for each layer $i$, can transform the transition matrix $W_i$ to be $$ \frac{I + R_i}{2}.$$
That is, after relabeling, the new program will choose to wait ($I$ is the identity matrix) at its current state, with probability 1/2. It will choose to transit according to a permutation matrix $R_i$ for the other 1/2 probability.
Then notice that a path segment that leaves a state and later returns to it can then be replaced by waiting. Thus every reachable state has a corresponding certificate path which has at most $w-1$ state changes.

We represent the timing of state changes by earliest-arrival functions. Let $(E^m)_{v,u}(t)$ denote the earliest arrival time at $v$ from $u$ at time $t$ using at most $m$ state changes, with value $\infty$ if no such path exists. The case $m=1$ is computed by scanning the layers. For positive integers $r,m$, we have
\[
   (E^{r+m})_{v,u}(t)=\min_{z\in[w]}(E^r)_{v,z}\bigl((E^m)_{z,u}(t)\bigr).
\]
Here the path first uses at most $m$ changes to reach $z$, then at most $r$ changes to reach $v$. Keeping only the earliest arrival at $z$ suffices because the path can wait there for any later continuation. Since every reachable state has a path with at most $w-1$ changes, evaluating these functions for $m\ge w-1$ decides reachability.

We recursively group $a$ compositions at each level, enumerating the $w^{a-1}$ tuples of intermediate states and evaluating the $a$ child compositions sequentially. Each level stores the tuple and a constant number of time values, using $O(\log n+a\log w)$ bits. The depth is $d=\lceil\log_a(w-1)\rceil$, and each call makes at most $aw^{a-1}$ recursive calls. Choosing $a=\Theta(\max\{2,\log n/\log^2 w\})$ bounds the running time by $2^{O(a\log^2 w)}\poly(n,w)$, which is polynomial when $\log^2 w=O(\log n)$. For polylogarithmic width, the depth is constant and the space is $O(\log n)$.

\paragraph{Powers of a fixed regular transition matrix.}
We first describe our method for approximation for the two-sided error case and then describe our method for positivity.

We approximate $A^ne_s$ by the uniform distribution on its support $R_n$. In the component containing $s$, the cyclic classes have a common size $h$, and every edge advances to the next class. Our key lemma shows that the set reachable from a nonempty proper subset of one class grows in size within $w-1$ steps. Hence $R_n$ is the whole class reached at time $n$ once $n\ge(w-1)(h-1)$.

The same lemma gives a mixing bound through the forward evolving-set process. Starting from $S_0=\{s\}$, each update chooses equally between the vertices with both incoming edges from $S_t$ and those with at least one. If $|S_t|=j$, the two successor sizes are $j-a$ and $j+a$. The potential $\Phi(S)=|S|(h-|S|)$ satisfies
\[
   \mathbb E[\Phi(S_{t+1})\mid S_t]=\Phi(S_t)-a^2,
   \qquad
   \left\|A^te_s-\operatorname{Unif}(C_t)\right\|_{\mathrm{TV}}
   \le\frac{\mathbb E\Phi(S_t)}h,
\]
where $C_t$ is the cyclic class reached at time $t$. Until the two choices differ, the process follows the reachable set deterministically. The growth lemma forces them to differ within $w-1$ steps whenever the set is nonempty and proper, decreasing the expected potential by at least one. Since $\Phi\le h^2/4$, this gives an error bound that decreases geometrically over blocks of $O(w^3)$ steps.

To compute $q_n$, we determine membership in $R_n$ and count $|R_n|$ and $|F\cap R_n|$. For each $1\le m\le w$, attach a clock modulo $m$ to the support graph, replacing $u\to v$ by $(u,z)\to(v,z+1\bmod m)$. Every vertex in the resulting graph has equal indegree and outdegree, so each weak component is strongly connected. For $n\ge w^2-1$, we show that $v\in R_n$ exactly when $(s,0)$ and $(v,n\bmod m)$ are connected in the underlying undirected graph for every $m\le w$. A length-$n$ walk satisfies all these tests. Conversely, choose a simple directed cycle through $s$ of length $c\le w$. The test for $m=c$ gives a directed simple path in the clock graph of length $\ell\le wc-1\le n$ with $\ell\equiv n\pmod c$. Prefixing its projection with $(n-\ell)/c$ traversals of the cycle gives a walk of length exactly $n$. Each clock graph has at most $w^2$ vertices, so undirected connectivity and the two counters use space $O(\log(w+b))$ and time $\poly(w,b)$, where $b$ is the bit length of $n$.

\paragraph{SC derandomization for models beyond BPL.}
Next we show our SC derandomizations for models beyond BPL. We start with the communication-network model of INW
\cite{impagliazzoPseudorandomnessNetworkAlgorithms1994}. Then we apply the result to read-$p$ machines and stack machines.

In the communication-network model, a computation
is done by local processors that use private randomness and
exchange messages according to a fixed communication graph $H = (V, E)$. 
Each node in the graph corresponds to a processor. 
The computation proceeds in alternating rounds of local computation and communication, beginning with a local computation round. 
In a communication round, each processor can choose some of its neighbors in $H$, and then send messages to them. 
In a local computation round, each processor do a computation based on the messages received, and their locally stored information.
On the other hand, these processors can be organized into a balanced binary partition tree in the following way. 
Each leaf represents one processor, and each
internal node represents the processors in its subtree. 
We specifically consider the following notion (previously defined in \cite{impagliazzoPseudorandomnessNetworkAlgorithms1994}).
\begin{definition}[Partition tree and width \cite{impagliazzoPseudorandomnessNetworkAlgorithms1994}]
A \emph{partition tree} \(\mathcal T\) for \(H\) is a rooted binary tree
whose leaves are in bijection with \(V\), and in  \(\mathcal T\) every internal
node has two children. For a node \(\nu\), let \(V_\nu\) be the set of
processors corresponding to its descendant leaves. For an internal
node with children \(\nu_{\mathrm L},\nu_{\mathrm R}\), write
\[
   A_\nu:=V_{\nu_{\mathrm L}},
   \qquad
   B_\nu:=V_{\nu_{\mathrm R}},
   \qquad
   C_\nu:=V\setminus V_\nu.
\]
Let \(\cut(\nu)\) consist of the edges with endpoints in different sets
among \(A_\nu,B_\nu,C_\nu\). The width \(\omega(\nu)\) is the minimum
size of a vertex cover that can cover all edges of \(\cut(\nu)\), and the width of the tree is
\[
   \omega(\mathcal T):=\max_{\nu\text{ internal}}\omega(\nu),
\]
with value zero for a one-leaf tree. The tree is \emph{balanced} if its
depth is \(O(\log(N+1))\). The \emph{width} $\omega$ of \(H\) is the minimum
width of a balanced partition tree for \(H\).
\end{definition}

For our algorithm, we need to compute a balanced parition a tree with minimum width, implicitly in SC. This means that there is an SC algorithm which can specify each node of the tree from leaves to the root, from left to right. The algorithm can provide the processors that belong to the substree of a given node. Later when we consider our two specific models we will explain how to realize this algorithm. Now we just assume it. 

For every layer from leaves to the root,
we use averaging samplers to construct, at every node of the layer, a generator whose seed
determines the private randomness of all the processors in this subtree.
At a leaf, we use an averaging sampler to generate the processor's
private random string from a short inner seed, after fixing the
sampler's offline seed. Then consider an internal node whose two
child generators have already been constructed. We apply the
two-party INW generator to produce a pair of seeds, one for each child generator, using one fresh seed together with the randomness for one step of random walk. The seed required by this INW step is longer than either child
seed. We therefore apply another averaging sampler to sample this longer seed. 
We will explain how to find a good (offline) seed of the sampler, but with the sampler's offline seed fixed, a short inner seed
determines the INW seed, which determines both child seeds and
hence all randomness in the subtree. This gives the parent
a generator with a short seed, allowing us to repeat the same
construction at the next level.
We continue until reaching the root, whose generator supplies
randomness for the entire network. Once all offline seeds have
been fixed, we enumerate the root generator's short seed and
average the resulting outputs of the computation.

For each node, to find an appropriate offline seed of the samplers, we use \emph{boundary tests} as follows (We think this concept is also implied in \cite{impagliazzoPseudorandomnessNetworkAlgorithms1994}).
Each test specifies a sequence of messages between the processors which belong to the subtree of this node and all the other processors. It also specifies all local outputs for this subtree. It checks their consistency with the computation inside the subtree.
We require the offline seed to provide a good enough approximation to the probabilities of passing every such test. 
We also managed to show that any distribution fooling these boundary tests is enough to attain a good approximation to the acceptance probability of the true randomness case.
For each offline seed candidate, we run every boundary-test, and compute the acceptance probabilities corresponding to the pseudo-distribution and the true randomness case. We check if the two probabilities are close enough. 

Let $\omega$ be the partition-tree width and $N$ the number of processors.
Suppose for each processor, its total communication in bit length, is upper bounded by a universal constant $c$.
Also assume that each boundary test records $k$ output bits in the corresponding subtree.
We can show that for each node, the (offline) seed length of the sampler is
\[
   O\!\left(\omega c+k+\log(N/\varepsilon)\right).
\]
When this seed-length is logarithmic, searching for offline seeds and the final seed
enumeration remain in polynomial-time.

For polynomial-time logspace read-$p$ machines with constant $p$,
we consider the communication graph $H$ to be a simple path
\cite[Section~4.1]{impagliazzoPseudorandomnessNetworkAlgorithms1994},
where each processor possesses one random-tape cell.
We managed to show that the partition tree with minimum width can be computed in SC since the perfect binary tree is already the desired tree. The width of it can be verified to be  $\omega=O(1)$.
Each processor communicates at most $O(p)$ configurations to its neighbors (Each configuration is a configuration of the processor at a certain time point, but excluding the original input $x$ since it is fixed). 
So its
communication (we slightly modify the protocol to include a timestamp for each message) has length
$c = O(p\log n)=O(\log n)$.
Since $N=\poly(n)$,
$k=1$, the sampler seed length is $O(\log(n/\varepsilon))$.
Each boundary test only specifies $O(\log n)$ bits of information. So enumerating all tests and hence also the total time to find a good seed of a sampler is a polynomial.

For polynomial-time logspace machines with an oblivious stack,
we let each processor correspond to a computation step. 
Each processor will send its configuration information to the processor corresponding to the next computation step. It also sends a symbol ($O(\log n)$ bits) if there is a push operation pushed this symbol at this step, to the corresponding processor which will pop out this pushed symbol. 
Since the work space is logarithmic, 
each processor exchanges at most $c=O(\log n)$ bits of information.
We managed to show that, the restriction that all timings of stack operations are independent of the randomness, yields a balanced partition
tree of width $\omega=O(1)$, which can be computed in SC. 
Then with $N=\poly(n)$ and $k=1$,
we obtain the same seed-length bound for samplers.
Hence the overall space is as stated.
We managed to show that each boundary test can be computed by a deterministic auxiliary pushdown
machine in polynomial-time. So this test can be computed in SC$^2$ by
\cite{Sudborough1978,cook1979deterministic}.
The total number of tests is bounded by $2^c$. So the overall running time is a polynomial.

\section{Preliminaries}

All logarithms are to base two, and $[n]:=\{1,\ldots,n\}$.
Finite alphabets and sampler universes are nonempty indexed sets; their
elements are represented by binary indices. Unless otherwise stated,
expectations over finite sets are uniform.
For a matrix $A\in\mathbb R^{w\times w}$, we use the induced norms
\[
   \|A\|_1:=\max_v\sum_u|A_{u,v}|,
   \qquad
   \|A\|_\infty:=\max_u\sum_v|A_{u,v}|,
\]
and write $\|A\|_2$ for its spectral norm and
$\|A\|_{\max}:=\max_{u,v}|A_{u,v}|$ for its entrywise maximum norm.

\begin{definition}[Read-once branching programs (ROBP)]
A read-once branching program $B$ of length $n$, width $w$, and alphabet
$\Sigma$ is a directed acyclic multigraph with $n+1$ layers
$V_0,\ldots,V_n$, each identified with $[w]$.
For every $i\in[n]$, every node in $V_{i-1}$ has exactly one outgoing
edge labeled by each symbol $\sigma\in\Sigma$, leading to a node in $V_i$.
The program has a designated start node $v_{\mathrm{start}}\in V_0$ and
an accepting set $V_{\mathrm{accept}}\subseteq V_n$.
On input $x=(x_1,\ldots,x_n)\in\Sigma^n$, the computation follows the
unique path from $v_{\mathrm{start}}$ whose $i$-th edge is labeled by $x_i$.
The output $B(x)$ is $1$ if the path ends in $V_{\mathrm{accept}}$, and $0$
otherwise. All ROBPs in this paper read their input symbols in this order.
\end{definition}

We use $B_i(\sigma)$ to denote the transition matrix from $V_{i-1}$ to
$V_i$, with $[B_i(\sigma)]_{u,v}=1$ if the $\sigma$-labeled edge from
$v\in V_{i-1}$ leads to $u\in V_i$, and $0$ otherwise.
Thus probability distributions are column vectors.
For $0\le i\le j\le n$, define
\[
   B_{i..j}(z_{i+1},\ldots,z_j)
   :=B_j(z_j)\cdots B_{i+1}(z_{i+1}),
   \qquad B_{i..i}:=I.
\]
We also write
\[
   W_i:=\E_{\sigma\in\Sigma}[B_i(\sigma)],
   \qquad
   W_{i..j}:=\E_{z\in\Sigma^{j-i}}[B_{i..j}(z)]
   =W_j\cdots W_{i+1},
   \qquad W_{i..i}:=I.
\]
For a generator $G:\{0,1\}^d\to\Sigma^n$, let
\[
   M_{i..j}[G]
   :=\E_{r\in\{0,1\}^d}
      [B_{i..j}(G(r)|_{(i,j]})],
\]
where $G(r)|_{(i,j]}$ consists of coordinates $i+1,\ldots,j$.
For a named interval $J\subseteq[n]$ of input positions, we also use
$B_J$, $W_J$, and $M_J[G]$ for the corresponding transitions.
In particular, if $J=[i,j]$, then $W_J=W_{i-1..j}$.
Superscripts distinguish programs in a family, as in
$B^\iota_{i..j}$, $W^\iota_{i..j}$, and $M^\iota_{i..j}[G]$.

\begin{definition}[Regular ROBP]
A regular ROBP is an ROBP $B$ such that, for every $i\in[n]$, each node
in $V_i$ has in-degree $|\Sigma|$, counting parallel edges with multiplicity.
Equivalently, every layer average $W_i$ is doubly stochastic.
\end{definition}

\begin{definition}[Permutation ROBP]
A permutation ROBP is an ROBP $B$ such that, for every $i\in[n]$ and
$\sigma\in\Sigma$, the matrix $B_i(\sigma)$ is a permutation matrix.
\end{definition}

We write $\Bcl_{(n,w,\Sigma)}$, $\Rcl_{(n,w,\Sigma)}$, and
$\Pcl_{(n,w,\Sigma)}$ for the classes of general, regular, and permutation
ROBPs, respectively. Every permutation ROBP is regular.

\begin{definition}[PRG]
Let $\mathcal F$ be a class of ROBPs $B:\Sigma^n\to\{0,1\}$.
An $\varepsilon$-PRG for $\mathcal F$ is a function
$G:\{0,1\}^d\to\Sigma^n$ such that, for every $B\in\mathcal F$,
\[
   \left|\E_{x\in\Sigma^n}B(x)
   -\E_{r\in\{0,1\}^d}B(G(r))\right|\le\varepsilon.
\]
The input length $d$ is the seed length.
We say that $G$ is \emph{explicit} if it can be computed in space $O(d)$
and time $\poly(d,n,\log|\Sigma|)$, using a write-only output tape.
\end{definition}

\begin{definition}[WPRG]
Let $\mathcal F$ be a class of ROBPs $B:\Sigma^n\to\{0,1\}$.
A $W$-bounded $\varepsilon$-WPRG for $\mathcal F$ is a pair of functions
$(G,\sigma):\{0,1\}^d\to\Sigma^n\times\mathbb R$ such that, for every
$B\in\mathcal F$,
\[
   \left|\E_{x\in\Sigma^n}B(x)
   -\E_{r\in\{0,1\}^d}[\sigma(r)B(G(r))]\right|\le\varepsilon,
   \qquad |\sigma(r)|\le W\quad\text{for every }r.
\]
The input length $d$ is the seed length.
We say that $(G,\sigma)$ is \emph{explicit} if both functions can be
computed in space $O(d)$ and time $\poly(d,n,\log|\Sigma|)$, with the
weights output as binary-encoded rational numbers.
\end{definition}

\begin{definition}[Averaging sampler]
\label{def:Sampler}
Let $\Omega$ be a finite set. An $(\alpha,\gamma)$-averaging sampler for
$\Omega$ is a function
\[
   \Samp:\{0,1\}^r\times\{0,1\}^p\to\Omega
\]
such that, for every function $f:\Omega\to[-1,1]$,
\[
   \Pr_{x\in\{0,1\}^r}
   \left[
      \left|
         \E_{y\in\{0,1\}^p}f(\Samp(x,y))
         -\E_{z\in\Omega}f(z)
      \right|\ge\alpha
   \right]\le\gamma.
\]
We call $x$ the \emph{offline seed} and $y$ the \emph{online seed}.
\end{definition}

We use the following finite-set form of explicit averaging samplers, based on
\cite[Lemma~19 and Appendix~B]{chattopadhyayOptimalErrorPseudodistributions2020};
see also \cite{goldreich2011sample}.

\begin{lemma}[Explicit averaging samplers]
\label{lem:good sampler}
For every finite indexed set $\Omega$ and every $0<\alpha,\gamma<1/2$,
there exists an explicit $(\alpha,\gamma)$-averaging sampler
\[
   \Samp:\{0,1\}^r\times\{0,1\}^p\to\Omega
\]
with
\[
   r=\lceil\log|\Omega|\rceil
      +O(\log(1/\alpha)+\log(1/\gamma)),
   \qquad
   p=O(\log(1/\alpha)+\log\log(1/\gamma)).
\]
Given $x$ and $y$, the index of $\Samp(x,y)$ is computable in time
polynomial in $\log|\Omega|+\log(1/\alpha)+\log(1/\gamma)$ and workspace
\[
   O(\log|\Omega|+\log(1/\alpha)+\log(1/\gamma)).
\]
\end{lemma}

\begin{proof}
For a binary-string universe, the result follows from the cited sampler
construction, rescaling its $[0,1]$-valued tests to $[-1,1]$.
For a general universe, identify $\Omega$ with $\{0,\ldots,N-1\}$, where
$N=|\Omega|$. Set
$m:=\lceil\log(8N/\alpha)\rceil$ and map an $m$-bit integer $u$ to
the element of $\Omega$ indexed by $u\bmod N$.
The image of the uniform distribution differs from uniform on $\Omega$
by less than $N/2^m$ in total variation, so every $[-1,1]$-valued
expectation changes by at most $2N/2^m\le\alpha/4$.
Compose this map with a binary-string sampler having accuracy $\alpha/2$
and failure probability $\gamma$.
The resulting error is less than $\alpha$ except with probability at most
$\gamma$, and the extra $O(\log(1/\alpha))$ bits are absorbed by the
stated bounds.
\end{proof}

For a nonempty set $S\subseteq[w]$, let $\operatorname{Unif}(S)$ denote the uniform probability vector on $S$.

\begin{definition}[Total variation distance]
\label{def:total-variation}
For probability vectors $\mu,\nu\in\mathbb R^w$, their total variation distance, also called statistical distance, is
\[
   \SD(\mu,\nu)
   =\|\mu-\nu\|_{\mathrm{TV}}
   :=\frac12\|\mu-\nu\|_1
   =\max_{F\subseteq[w]}
      \left|\mathbf1_F^\top(\mu-\nu)\right|.
\]
\end{definition}

\paragraph{Reducing alphabet size with a sampler.}
Sampling each input symbol with a fixed offline seed reduces the alphabet
while approximately preserving all interval transition matrices.

\begin{lemma}[Reducing alphabet size with averaging samplers]
\label{lem:reducing-alphabet-size-with-sampler}
Let $n,w\in\mathbb N$, let $\Sigma$ be a finite alphabet, and let
$0<\varepsilon<1/2$.
Let $\{B^\iota\}_{\iota\in I}$ be a nonempty finite indexed family of
length-$n$ width-$w$ ROBPs over $\Sigma$.
There are an explicit sampler
$\Samp:\{0,1\}^r\times\{0,1\}^p\to\Sigma$ and an advice string
$x\in\{0,1\}^r$ such that the ROBPs defined by
\[
   B_i^{\iota,\Samp,x}(y):=B_i^\iota(\Samp(x,y)),
   \qquad i\in[n],\quad y\in\{0,1\}^p,
\]
satisfy, for every $\iota\in I$ and $0\le i\le j\le n$,
\[
   \|W^{\iota,\Samp,x}_{i..j}-W^\iota_{i..j}\|_1\le\varepsilon.
\]
If the original programs are permutation ROBPs, the sampled programs are
also permutation ROBPs, and the same bound holds in the induced
$\infty$-norm. The seed lengths satisfy
\[
   p=O(\log(nw/\varepsilon)+\log\log(|I|+2)),
   \qquad
   r=\lceil\log|\Sigma|\rceil+O(\log(nw|I|/\varepsilon)).
\]
Given query access to the transition functions, the advice can be found in
time $\poly(n,w,|\Sigma|,|I|,1/\varepsilon)$ and additional workspace
$O(\log(nw|\Sigma||I|/\varepsilon))$.
These bounds count each transition query as one operation; its evaluation
time is charged for every call, and the workspace needed by one call is
added when the input programs are given implicitly.
\end{lemma}

\begin{proof}
Take an $(\varepsilon/(nw),1/(4nw^2|I|))$-averaging sampler from
\Cref{lem:good sampler}.
Apply its guarantee to the $nw^2|I|$ functions
$\sigma\mapsto[B_i^\iota(\sigma)]_{u,v}$.
By a union bound, with probability at least $3/4$ over $x$, every such test
has error at most $\varepsilon/(nw)$.
Fix such an $x$ and write $\widetilde W_i^\iota:=W_i^{\iota,\Samp,x}$.
Then $\|\widetilde W_i^\iota-W_i^\iota\|_1\le\varepsilon/n$.
Both matrices are column-stochastic, so their induced $1$-norms equal $1$.
The identity
\[
   \widetilde W^\iota_{i..j}-W^\iota_{i..j}
   =\sum_{t=i+1}^j
      \widetilde W^\iota_{t..j}
      (\widetilde W_t^\iota-W_t^\iota)
      W^\iota_{i..t-1}
\]
therefore gives an induced $1$-norm error at most
$(j-i)\varepsilon/n\le\varepsilon$.
For permutation ROBPs, every sampled transition is a permutation matrix.
Both layer averages are then doubly stochastic, and the same argument
applies in the induced $\infty$-norm.

To find $x$, enumerate offline seeds and check the displayed scalar tests
one at a time. Each test is evaluated by counting transitions over
$\Sigma$ and over $\{0,1\}^p$.
The sampler bounds give $2^r,2^p\le\poly(n,w,|\Sigma|,|I|,1/\varepsilon)$.
The current seed, test indices, and counters require
$O(\log(nw|\Sigma||I|/\varepsilon))$ bits.
Error tolerances may be rounded down to powers of two, so the comparisons
use rational arithmetic within the same bounds.
\end{proof}

\paragraph{SV-approximation with two error parameters.}
For a doubly stochastic matrix $W$, define
\[
   R(W):=I-W^\top W,
   \qquad L(W):=I-WW^\top,
\]
and the associated squared seminorms
\[
   \|x\|_{R(W)}^2:=\|x\|_2^2-\|Wx\|_2^2,
   \qquad
   \|y\|_{L(W)}^2:=\|y\|_2^2-\|W^\top y\|_2^2.
\]

\begin{definition}[$(\varepsilon,\delta)$-SV approximation]
\label{def: two para SV}
Let $W,\widehat W\in\mathbb R^{w\times w}$, where $W$ is doubly
stochastic, and let $\varepsilon,\delta\ge0$.
We write $\widehat W\asvapprox_{\varepsilon,\delta}W$ if, for every
$x,y\in\mathbb R^w$,
\[
   |y^\top(\widehat W-W)x|
   \le
   \frac{\varepsilon}{2}
      (\|x\|_{R(W)}^2+\|y\|_{L(W)}^2)
   +\frac{\delta}{2}(\|x\|_2^2+\|y\|_2^2).
\]
\end{definition}

\begin{lemma}
\label{fact:sv-to-pointwise}
If $\widehat W\asvapprox_{\varepsilon,\delta}W$, then, for every
$u,v\in[w]$,
\[
   |[\widehat W]_{u,v}-[W]_{u,v}|\le\varepsilon+\delta.
\]
\end{lemma}

\begin{proof}
Apply the definition with $x=e_v$ and $y=e_u$, where $e_u,e_v$ are
standard basis vectors. Since both have Euclidean norm $1$,
\[
   |e_u^\top(\widehat W-W)e_v|
   \le\frac{\varepsilon}{2}
      (\|e_v\|_{R(W)}^2+\|e_u\|_{L(W)}^2)+\delta
   \le\varepsilon+\delta.\qedhere
\]
\end{proof}

We say that a generator $G:\{0,1\}^d\to\Sigma^n$
$(\varepsilon,\delta)$-SV-approximates a length-$n$ width-$w$ regular
ROBP $B$ if
$M_{i..j}[G]\asvapprox_{\varepsilon,\delta}W_{i..j}$ for every
$0\le i\le j\le n$.
It fools $B$ with SV error $\varepsilon$ if it
$(\varepsilon,0)$-SV-approximates $B$.

\paragraph{The INW generator for permutation ROBPs.}
We use the following SV guarantee for the INW generator
\cite{impagliazzoPseudorandomnessNetworkAlgorithms1994}, established in
\cite{chenWeightedPseudorandomGenerators2023} as a strengthening of the
analysis in \cite{hoza2021pseudorandom}.

\begin{lemma}[INW generator with SV error]
\label{lem:inw-sv-generator}
For every $n\ge2$, $s\ge1$, and $0<\varepsilon<1/2$, there is an
explicit generator $G:\{0,1\}^d\to(\{0,1\}^s)^n$ such that, for
every length-$n$ permutation ROBP $B$ over $\Sigma=\{0,1\}^s$, of
arbitrary width, and every $0\le i\le j\le n$,
\[
   M_{i..j}[G]\asvapprox_{\varepsilon,0}W_{i..j}.
\]
Its seed length is
\[
   d=s+O\bigl(\log n(\log\log n+\log(1/\varepsilon))\bigr).
\]
The generator is the INW recursion recalled below, with suitable expander
parameters. Given a seed and an index $a\in[n]$, its $a$-th output symbol
is computable in time $\poly(d,\log n)$ and workspace $O(d)$.
\end{lemma}

The guarantee for every interval follows by replacing the transitions
outside that interval by identity matrices in the length-$n$ program.
For $n=1$, the identity generator is exact.
By \Cref{fact:sv-to-pointwise}, the SV guarantee gives entrywise error
at most $\varepsilon$; summing over an accepting set $S$ gives acceptance
error at most $|S|\varepsilon$.

\section{Two-Parameter SV Approximation}\label{sec:approx-sv}

Throughout, all matrices are real $w\times w$ matrices. We use
$(\varepsilon,\delta)$-SV approximation as defined in
\Cref{def: two para SV}, with its associated right and left squared
seminorms. When $\delta=0$, this is the usual SV approximation, up to
constant normalization.

For a contraction $W$, we also write
\[
   D(W,x):=\|x\|_2^2-\|Wx\|_2^2=x^\top(I-W^\top W)x.
\]
Thus $D(W,x)=\|x\|_{R(W)}^2$ when $W$ is doubly stochastic.

\subsection{A One-Sided Consequence}

\begin{lemma} 
\label{lem:one-sided}
Let \(W,\widetilde W\in \R^{w\times w}\), where \(W\) is doubly stochastic.  If
\[
  \widetilde W \asvapprox_{\varepsilon,\delta} W,
\]
then for every \(x,y\in\R^w\),
\[
  \norm{(\widetilde W-W)x}_2
  \le
  \sqrt{(\varepsilon+\delta)
  \left(\varepsilon\norm{x}_{R(W)}^2+\delta\norm{x}_2^2\right)},
\]
and
\[
  \norm{(\widetilde W-W)^\top y}_2
  \le
  \sqrt{(\varepsilon+\delta)
  \left(\varepsilon\norm{y}_{L(W)}^2+\delta\norm{y}_2^2\right)}.
\]
In particular, when \(\delta=0\),
\[
  \norm{(\widetilde W-W)x}_2\le \varepsilon\norm{x}_{R(W)},
  \qquad
  \norm{(\widetilde W-W)^\top y}_2\le \varepsilon\norm{y}_{L(W)}.
\]
\end{lemma}

\begin{proof}
We prove the first inequality; the second follows by transposing.  Let
\(\Delta=\widetilde W-W\).  If \(\Delta x=0\), there is nothing to prove.  If
\[
  \varepsilon\norm{x}_{R(W)}^2+\delta\norm{x}_2^2=0,
\]
then the desired bound follows by applying the argument below to \(x+tz\) and
letting \(t\to 0\).  Thus assume
\[
  \varepsilon\norm{x}_{R(W)}^2+\delta\norm{x}_2^2>0.
\]

For any \(\eta>0\), let
\[
  u=\sqrt{\eta}\,x,
  \qquad
  v=\frac{\Delta x}{\sqrt{\eta}}.
\]
Since \(v^\top\Delta u=\norm{\Delta x}_2^2\), by using using \cref{def: two para SV}, we get  

\begin{align*}
  \norm{\Delta x}_2^2 =  |v^\top\Delta u| 
  &\le
  \frac{\varepsilon}{2}
  \left(\eta\norm{x}_{R(W)}^2
  +\frac{1}{\eta}\norm{\Delta x}_{L(W)}^2\right)
  +\frac{\delta}{2}
  \left(\eta\norm{x}_2^2
  +\frac{1}{\eta}\norm{\Delta x}_2^2\right)                         \\
  &\le
  \frac{1}{2}
  \left(
  \eta\left(\varepsilon\norm{x}_{R(W)}^2+\delta\norm{x}_2^2\right)
  +\frac{\norm{\Delta x}_2^2}{\eta}(\varepsilon+\delta)
  \right),
\end{align*}

where we use \(\norm{z}_{L(W)}\le \norm{z}_2\).  Choose
\[
  \eta=\norm{\Delta x}_2
  \sqrt{\frac{\varepsilon+\delta}
  {\varepsilon\norm{x}_{R(W)}^2+\delta\norm{x}_2^2}}.
\]
Then the two terms in the RHS of the last inequality are equal. 
Hence
\[
  \norm{\Delta x}_2^2
  \le
  \norm{\Delta x}_2
  \sqrt{(\varepsilon+\delta)
  \left(\varepsilon\norm{x}_{R(W)}^2+\delta\norm{x}_2^2\right)}.
\]
Dividing by \(\norm{\Delta x}_2\) proves
\[
  \norm{\Delta x}_2
  \le
  \sqrt{(\varepsilon+\delta)
  \left(\varepsilon\norm{x}_{R(W)}^2+\delta\norm{x}_2^2\right)}.
\]
The other inequality can be proved in the same way, by using the fact that \(R(W^\top)=L(W)\).  
\end{proof}
\subsection{Products}

\begin{lemma}[Product of approximate SV approximations]
\label{lem:product}
Let \(W_1,W_2,\widetilde W_1,\widetilde W_2\in \R^{w\times w}\), where \((W_1,W_2\) are doubly stochastic.
Assume
\[
  \widetilde W_1 \asvapprox_{\varepsilon,\delta} W_1
  \qquad\text{and}\qquad
  \widetilde W_2 \asvapprox_{\varepsilon,\delta} W_2.
\]
Then
\[
  \widetilde W_2\widetilde W_1 \asvapprox_{\varepsilon+2\varepsilon(\varepsilon+\delta),2\delta+2\delta(\varepsilon+\delta)} W_2W_1.
\]
\end{lemma}

\begin{proof}
Write
\[
  \Delta_i=W_i-\widetilde W_i,
  \qquad
  W_{21}=W_2W_1.
\]
For arbitrary \(x,y\in \R^w\),
\[
\begin{aligned}
  y^\top(W_{21}-\widetilde W_2\widetilde W_1)x
  &=
  y^\top\Delta_2W_1x
  +y^\top W_2\Delta_1x
  -y^\top\Delta_2\Delta_1x .
\end{aligned}
\]
We bound the three terms separately.  From the assumption for
\(\widetilde W_2\),
\[
\begin{aligned}
  |y^\top\Delta_2W_1x|
  &\le
  \frac{\varepsilon}{2}
  \left(\norm{y}_{L(W_2)}^2+\norm{W_1x}_{R(W_2)}^2\right)          \\
  &\quad+
  \frac{\delta}{2}
  \left(\norm{y}_2^2+\norm{W_1x}_2^2\right)                       \\
  &\le
  \frac{\varepsilon}{2}
  \left(\norm{y}_{L(W_2)}^2+\norm{W_1x}_{R(W_2)}^2\right)
  +\frac{\delta}{2}
  \left(\norm{y}_2^2+\norm{x}_2^2\right),
\end{aligned}
\] 
where the last inequality uses the property of doubly stochastic matrix.
Similarly,
\[
\begin{aligned}
  |y^\top W_2\Delta_1x|
  &\le
  \frac{\varepsilon}{2}
  \left(\norm{W_2^\top y}_{L(W_1)}^2+\norm{x}_{R(W_1)}^2\right)
  +\frac{\delta}{2}
  \left(\norm{y}_2^2+\norm{x}_2^2\right).
\end{aligned}
\]
For the cross term, we have that
\[
\begin{aligned}
  |y^\top\Delta_2\Delta_1x|
  &\le
  \norm{\Delta_2^\top y}_2\norm{\Delta_1x}_2                                      \\
  &\le
  (\varepsilon+\delta)
  \sqrt{\left(\varepsilon\norm{y}_{L(W_2)}^2+\delta\norm{y}_2^2\right)
  \left(\varepsilon\norm{x}_{R(W_1)}^2+\delta\norm{x}_2^2\right)}                  \\
  &\le
  (\varepsilon+\delta)
  \Bigl(
    \varepsilon\norm{y}_{L(W_2)}\norm{x}_{R(W_1)}
    +\sqrt{\varepsilon\delta}\norm{y}_{L(W_2)}\norm{x}_2                           \\
  &\hspace{3.9cm}
    +\sqrt{\varepsilon\delta}\norm{y}_2\norm{x}_{R(W_1)}
    +\delta\norm{y}_2\norm{x}_2
  \Bigr),
\end{aligned}
\]
where the second inequality is using  Lemma~\ref{lem:one-sided} for both terms.

We now apply AM-GM with weights chosen to put the mixed terms mostly into the
SV-error term:
\[
  \varepsilon\norm{y}_{L(W_2)}\norm{x}_{R(W_1)}
  \le
  \frac{\varepsilon}{2}
  \left(\norm{y}_{L(W_2)}^2+\norm{x}_{R(W_1)}^2\right),
\]
\[
  \sqrt{\varepsilon\delta}\norm{y}_{L(W_2)}\norm{x}_2
  \le
  \frac{\varepsilon}{2}\norm{y}_{L(W_2)}^2+
  \frac{\delta}{2}\norm{x}_2^2,
\]
\[
  \sqrt{\varepsilon\delta}\norm{y}_2\norm{x}_{R(W_1)}
  \le
  \frac{\varepsilon}{2}\norm{x}_{R(W_1)}^2+
  \frac{\delta}{2}\norm{y}_2^2,
\]
and
\[
  \delta\norm{y}_2\norm{x}_2
  \le
  \frac{\delta}{2}\left(\norm{x}_2^2+\norm{y}_2^2\right).
\]
Therefore
\[
\begin{aligned}
  |y^\top\Delta_2\Delta_1x|
  &\le
  \varepsilon(\varepsilon+\delta)
  \left(\norm{y}_{L(W_2)}^2+\norm{x}_{R(W_1)}^2\right)              \\
  &\quad+
  \delta(\varepsilon+\delta)
  \left(\norm{x}_2^2+\norm{y}_2^2\right).
\end{aligned}
\]
Finally, use the identities
\[
  \norm{x}_{R(W_1)}^2+\norm{W_1x}_{R(W_2)}^2
  =
  \norm{x}_{R(W_2W_1)}^2,
\]
\[
  \norm{y}_{L(W_2)}^2+\norm{W_2^\top y}_{L(W_1)}^2
  =
  \norm{y}_{L(W_2W_1)}^2.
\]
Adding the three estimates yields
\[
\begin{aligned}
  |y^\top(W_{21}-\widetilde W_2\widetilde W_1)x|
  &\le
  \left(\frac{\varepsilon}{2}+\varepsilon(\varepsilon+\delta)\right)
  \left(\norm{x}_{R(W_{21})}^2+\norm{y}_{L(W_{21})}^2\right)        \\
  &\quad+
  \left(\delta+\delta(\varepsilon+\delta)\right)
  \left(\norm{x}_2^2+\norm{y}_2^2\right).
\end{aligned}
\]
This is exactly
\[
  \widetilde W_2\widetilde W_1 \asvapprox_{\varepsilon+2\varepsilon(\varepsilon+\delta),2\delta+2\delta(\varepsilon+\delta)} W_2W_1.
\]
\end{proof}

\begin{corollary}
\label{cor:products}
Let $n\ge2$, and let $W_1,\ldots,W_n$ be doubly stochastic matrices and suppose that for each $i\in[n]$ there is a matrix $\widetilde W_{i}$ such that
\[
   \widetilde W_{i} \asvapprox_{\varepsilon,\delta} W_{i},
\]
with \(\varepsilon\le 1/(100\log n)\) and \(\delta\le 1/(100n\log n)\).  Then,
\[
    \widetilde W_{n}\cdots\widetilde W_{1} \asvapprox_{\left(1+6\lceil\log n\rceil(\varepsilon+n\delta)\right)\varepsilon,3n\delta} W_{n}\cdots W_{1}.
\]

\end{corollary}

\begin{proof}
Set $N:=2^{\lceil\log n\rceil}$ and pad the sequence to length $N$ with
identity matrices, retaining $n$ in the error parameters below.
Use the interval convention
$\widetilde W_{a..b}:=\widetilde W_b\cdots\widetilde W_{a+1}$ and
$W_{a..b}:=W_b\cdots W_{a+1}$.
Define $\varepsilon^{(h)}=(1+6h(\varepsilon+n\delta))\varepsilon$ and $\delta^{(h)}=2^h(1+6h(\varepsilon+n\delta))\delta$. 
 We show by induction on \(h\) that for any \(0\le h\le\lceil\log n\rceil\) and any dyadic interval \((a,a+2^h]\subseteq(0,N]\),

\begin{equation}
\label{eq:balanced-logm-induction}
   \widetilde W_{a..a+2^h} \asvapprox_{\varepsilon^{(h)},\delta^{(h)}} W_{a..a+2^h}.
   \end{equation}

The case \(h=0\) is the hypothesis.  For \(h>0\), We write \(\widetilde W_{a..a+2^h}=\widetilde W_{a+2^{h-1}..a+2^h}\widetilde W_{a..a+2^{h-1}}\) and \(W_{a..a+2^h}=W_{a+2^{h-1}..a+2^h}W_{a..a+2^{h-1}}\).
By the induction hypothesis, both \(\widetilde W_{a+2^{h-1}..a+2^h}\) and \(\widetilde W_{a..a+2^{h-1}}\) are \((\varepsilon^{(h-1)},\delta^{(h-1)})\)-SV approximations of their corresponding exact products.  By Lemma~\ref{lem:product}, we have
\begin{equation}
\label{eq:balanced-logm-induction-step}
   \widetilde W_{a..a+2^h} \asvapprox_{\varepsilon^{(h-1)}+2\varepsilon^{(h-1)}(\varepsilon^{(h-1)}+\delta^{(h-1)}),2\delta^{(h-1)}+2\delta^{(h-1)}(\varepsilon^{(h-1)}+\delta^{(h-1)})} W_{a..a+2^h}.
\end{equation}

Since $(1+6h(\varepsilon+n\delta))\leq 3/2$ for all \(h\le\lceil\log n\rceil\), and $2^{h-1}\le n$ for $h\ge1$, we have
\(2(\varepsilon^{(h-1)}+\delta^{(h-1)})\le 3(\varepsilon+n\delta)\), and we have
\begin{align}
\label{eq:balanced-logm-induction-step-E}
   \varepsilon^{(h-1)}+2\varepsilon^{(h-1)}(\varepsilon^{(h-1)}+\delta^{(h-1)})
   &\le
   \varepsilon^{(h-1)}+3\varepsilon\left(1+6h(\varepsilon+n\delta)\right)(\varepsilon+n\delta)
   \le\varepsilon^{(h)},\\
\label{eq:balanced-logm-induction-step-D}
   2\delta^{(h-1)}+2\delta^{(h-1)}(\varepsilon^{(h-1)}+\delta^{(h-1)})
    &\le
    2\delta^{(h-1)}+3\cdot2^{h-1}\delta\left(1+6h(\varepsilon+n\delta)\right)(\varepsilon+n\delta)
    \le\delta^{(h)}.
\end{align}

At the root, $\delta^{(\lceil\log n\rceil)}
\le(3/2)N\delta\le3n\delta$, and identity padding preserves the
original product. This proves the stated bound.
\end{proof}

\subsection{Transitivity}

\begin{lemma}[Transitivity with additive $\ell_2$ term]
\label{lem:transitivity}
Let \(A,B,C\) be doubly stochastic matrices.  Suppose
\[
  A \asvapprox_{\varepsilon_1,\delta_1} B
  \qquad\text{and}\qquad
  B \asvapprox_{\varepsilon_2,\delta_2} C.
\]
Then
\[
  A \asvapprox_{\varepsilon_1+\varepsilon_2+2\varepsilon_1\varepsilon_2,\delta_1+\delta_2+2\varepsilon_1\delta_2} C.
\]
\end{lemma}

\begin{proof}
We claim that
for every \(x,y\in\R^w\),
\[
  \norm{x}_{R(B)}^2
  \le
  (1+2\varepsilon_2)\norm{x}_{R(C)}^2+2\delta_2\norm{x}_2^2,
\]
and
\[
  \norm{y}_{L(B)}^2
  \le
  (1+2\varepsilon_2)\norm{y}_{L(C)}^2+2\delta_2\norm{y}_2^2.
\]
We prove the first inequality; the second follows by transposing.  Since
\[
\begin{aligned}
  \norm{Cx}_2^2-\norm{Bx}_2^2
  &=x^\top(C^\top C-B^\top B)x                                      \\
  &=x^\top\bigl((C^\top-B^\top)C+C^\top(C-B)-(C^\top-B^\top)(C-B)\bigr)x,
\end{aligned}
\]
and the last term is nonpositive, we have
\[
  \norm{Cx}_2^2-\norm{Bx}_2^2
  \le
  2 |(Cx)^\top(C-B)x|.
\]
Using \(B \asvapprox_{\varepsilon_2,\delta_2} C\) with test vectors
\(x\) and \(Cx\),
\[
\begin{aligned}
  |(Cx)^\top(C-B)x|
  &\le
  \frac{\varepsilon_2}{2}
  \left(\norm{x}_{R(C)}^2+\norm{Cx}_{L(C)}^2\right)
  +\frac{\delta_2}{2}
  \left(\norm{x}_2^2+\norm{Cx}_2^2\right)                         \\
  &\le
  \varepsilon_2\norm{x}_{R(C)}^2+
  \delta_2\norm{x}_2^2.
\end{aligned}
\]
Here we used \(\norm{Cx}_{L(C)}\le \norm{x}_{R(C)}\) (Because $ \norm{x}_{R(C)}^2- \norm{Cx}_{L(C)}^2 = \norm {x-C^TCx}_2^2>0$.) and
\(\norm{Cx}_2\le \norm{x}_2\).  Therefore, combining the above two inequalities,
\[
  \norm{Cx}_2^2-\norm{Bx}_2^2
  \le
  2\varepsilon_2\norm{x}_{R(C)}^2+2\delta_2\norm{x}_2^2.
\]
Since
\[
  \norm{x}_{R(B)}^2
  =\norm{x}_{R(C)}^2+\norm{Cx}_2^2-\norm{Bx}_2^2,
\]
this proves the claimed inequality for \(x\).  

Now fix arbitrary \(x,y\in\R^w\).  By the triangle inequality,
\[
  |y^\top(A-C)x|
  \le
  |y^\top(A-B)x|+|y^\top(B-C)x|.
\]
The second term is already measured relative to \(C\):
\[
  |y^\top(B-C)x|
  \le
  \frac{\varepsilon_2}{2}
  \left(\norm{x}_{R(C)}^2+\norm{y}_{L(C)}^2\right)
  +\frac{\delta_2}{2}
  \left(\norm{x}_2^2+\norm{y}_2^2\right).
\]
For the first term,
\[
\begin{aligned}
  |y^\top(A-B)x|
  &\le
  \frac{\varepsilon_1}{2}
  \left(\norm{x}_{R(B)}^2+\norm{y}_{L(B)}^2\right)
  +\frac{\delta_1}{2}
  \left(\norm{x}_2^2+\norm{y}_2^2\right)                          \\
  &\le
  \frac{\varepsilon_1(1+2\varepsilon_2)}{2}
  \left(\norm{x}_{R(C)}^2+\norm{y}_{L(C)}^2\right)                 \\
  &\quad+
  \frac{\delta_1+2\varepsilon_1\delta_2}{2}
  \left(\norm{x}_2^2+\norm{y}_2^2\right).
\end{aligned}
\]
Adding the two estimates gives
\[
\begin{aligned}
  |y^\top(A-C)x|
  &\le
  \frac{\varepsilon_1+\varepsilon_2+2\varepsilon_1\varepsilon_2}{2}
  \left(\norm{x}_{R(C)}^2+\norm{y}_{L(C)}^2\right)                 \\
  &\quad+
  \frac{\delta_1+\delta_2+2\varepsilon_1\delta_2}{2}
  \left(\norm{x}_2^2+\norm{y}_2^2\right),
\end{aligned}
\]
which is the claimed \((\varepsilon,\delta)\)-SV approximation.
\end{proof}

\section{Converting the PRG of \cite{hoza2021pseudorandom} to a derandomization in SC}
\label{sec:INW_in_SC}

We prove \cref{thm:regular-derand} in this section.
Length-one programs can be handled by averaging their transitions directly;
we therefore assume $n\ge2$ below.
The following lemma describes our SC derandomization for permutation ROBPs. We describe the derandomization as constructing a target generator with an offline seed and an online seed. We show that the offline seed can be found in SC and with the offline seed, the target generator can be computed in logspace. Also notice that we consider fooling a collection of ROBPs simultaneously  which is to be used for the next section, while for this section one can just think of fooling only one ROBP.
\begin{lemma}[Sampler-fixed INW generator, SV form]
\label{lem:det-compressed-inw-sv}

Let $n\ge2$ and $w\in \mathbb{N}$. Let $\Sigma$ be a finite alphabet. Let $0<\varepsilon,\delta<1/2$. There exists a deterministic algorithm that, given a nonempty finite indexed collection $\{B^\iota\}_{\iota\in\mathcal J}$ of length-$n$ width-$w$ permutation ROBPs over alphabet $\Sigma$, computes a binary string $x$ with length
\[
   |x|=O\!\left(
      \log|\Sigma|
      +\log n(\log\log n+\log(1/\varepsilon))
      +\log(nw|\mathcal J|/\delta)
   \right),
\]
and also computes a targeted generator $\mathcal G_x:\{0,1\}^{d}\to \Sigma^n$ (the generator is described by the offline seed $x$) such that for every $\iota\in \mathcal J$, we have
\[
    M^\iota_{0..n}[\mathcal G_x]\asvapprox_{\varepsilon,\delta} W^\iota_{0..n},
\]
Here $W^\iota_{0..n}$ is the transition matrix of $B^\iota$ and $M^\iota_{0..n}[\mathcal G_x]$ is the transition matrix of $B^\iota$ induced by $\mathcal G_x$. 
The seed length $d$ of the generator is $O(\log (nw|\mathcal J|/\varepsilon\delta))$. 

The computing of  \(x\) is in time
\(\poly(n,w,|\Sigma|,|\mathcal J|,1/\varepsilon,1/\delta)\) and uses workspace
\(|x|+O(\log(nw|\Sigma||\mathcal J|/\varepsilon\delta))\).
Given read-only access to \(x\), the generator is strongly explicit: on input
a seed \(y\in\{0,1\}^d\) and an index \(q\in[n]\), the symbol
\(\mathcal G_x(y)_q\) can be computed in time
\(\poly(|x|,d,\log n,\log|\Sigma|)\) and workspace
\(O(d+\log n+\log|\Sigma|)\).
\end{lemma}

\begin{proof}

   Pad all input programs with identity transitions to the next power of two.
   This increases their common length by less than a factor of two;
   henceforth $n$ denotes the padded length, and the final generator is
   restricted to the original coordinates.
   We recursively define the generator $\mathcal G_x$. Let $n_0=n,\Sigma_0=\Sigma$ and 
   \[
   B^{\iota,0}:=B^{\iota}, \quad \mathcal G^0:\Sigma_0^n\to \Sigma_0^n
   \]
   be the input ROBPs and the identity generator. We will define a sequence of ROBPs $B^{\iota,t}$ and generators $\mathcal G^t$ such that $B^{\iota,t}\equiv B^\iota\circ \mathcal G^t$ always holds.
   Here \(\circ\) denotes composition and we regard $B^{\iota,t}, B^\iota$ as matrix-valued functions.
   We regard $x$ as a concatenation of $x_0, \ldots, x_{K-1}$, where $K$ and the lengths of $x_0, \ldots, x_{K-1}$ will be defined later.
   For every $t =  1,\ldots, K-1$, our $B^{\iota,t}$ and $\mathcal G^t$ both rely on $x_0, \ldots, x_{t-1}$. 
   
   Assume that we have defined the ROBPs $B^{\iota,t }$ of length $n_t$ and alphabet $\Sigma_t$, and the generator $\mathcal G^{t }:\Sigma_t^{n_t}\to \Sigma^n$ for $0\leq t\leq K-1$.
   We now define the ROBPs $B^{\iota,t+1}$ and the generator $\mathcal G^{t+1}$. 

The procedure is divided into two parts: Constructing $\widetilde{\mathcal G}^{t}$ from $\mathcal G^{t}$ and constructing $\mathcal G^{t+1}$ from $\widetilde{\mathcal G}^{t}$.

   \paragraph{From $\mathcal G^{t}$ to $\widetilde{\mathcal G}^{t}$:}
   We first apply \cref{lem:reducing-alphabet-size-with-sampler} on  ${B^{\iota,t}}_{\iota\in \mathcal J}$ with error $\delta_{\mathrm{samp}}/\sqrt{w}$. 
   Let $\Samp:\{0,1\}^{r_t}\times \{0,1\}^{p_t}\to \Sigma_t$ be the sampler used in \cref{lem:reducing-alphabet-size-with-sampler}. 
   By \cref{lem:reducing-alphabet-size-with-sampler} we can find an advice $x_t\in \{0,1\}^{r_t}$, such that for every $\iota\in \mathcal J$ and $0\leq i\leq j\leq n_t$, we have
   \[
         \left\|W^{\iota,t}_{i..j}-\E_{y_{i+1},\ldots,y_j\in \{0,1\}^{p_t}}\left[B^{\iota,t}_{i..j}\left(\Samp(x_t,y_{i+1}),\ldots,\Samp(x_t,y_j)\right)\right]\right\|_\infty\leq \delta_{\mathrm{samp}}/\sqrt{w},
   \]
   where
   \[
   \begin{aligned}
      \delta_{\mathrm{samp}}&=\frac{\varepsilon\delta}{100n^4\lceil\log^2 n\rceil},\\
      |x_t|&=\lceil\log|\Sigma_t|\rceil+O(\log(nw|\mathcal J|/(\varepsilon\delta))),\\
      p_t&=O(\log(nw/(\varepsilon\delta))+\log\log(|\mathcal J|+2)).
   \end{aligned}
   \]
   We then define the new ROBPs $\widetilde{B}^{\iota,t}$ by replacing each symbol of $B^{\iota,t}$ with $\Samp(x_t,\cdot)$, i.e.,
   \[
     \widetilde{B}^{\iota,t}_i(y):= B^{\iota,t}_i(\Samp(x_t,y)), \quad \forall i\in [n_t], y\in \{0,1\}^{p_t}.
   \]
   And we define the new generator $\widetilde{\mathcal G}^{t}:\{0,1\}^{n_t\cdot p_t}\to \Sigma^{n}$ by
   \[
   \widetilde{\mathcal G}^{t}(y_1,\ldots,y_{n_t}):=
   \mathcal G^{t}(\Samp(x_t,y_1),\ldots,\Samp(x_t,y_{n_t})), \quad \forall y_1,\ldots,y_{n_t}\in \{0,1\}^{p_t}.
   \]
   
   \paragraph{From $\widetilde{\mathcal G}^{t}$ to $\mathcal G^{t+1}$:}
   Set $\varepsilon_{\mathrm{INW}}=\varepsilon/(100\lceil\log^2 n\rceil)$.
   Choose a power of two $m$ by
   \[
      \log m=\min\left\{\log n,
         \left\lceil\frac{\log(nw|\mathcal J|/(\delta\varepsilon))}
         {\log(\log n/\varepsilon_{\mathrm{INW}})}\right\rceil\right\},
      \qquad m_t:=\min\{m,n_t\}.
   \]
   Since $n_t$ and $m_t$ are powers of two, $m_t$ divides $n_t$.
   Use \cref{lem:inw-sv-generator} to obtain
   $\INW:\{0,1\}^{d_t}\to(\{0,1\}^{p_t})^{m_t}$ with
   $\varepsilon_{\mathrm{INW}}$-SV error. Since $m_t\le m\le n$,
   \[
      d_t=p_t+O\!\left(\log m_t
         (\log\log m_t+\log(1/\varepsilon_{\mathrm{INW}}))\right)
         =O(\log(nw|\mathcal J|/(\delta\varepsilon))).
   \]
   Replace each block of $m_t$ transitions by the transition induced by
   $\INW$, defining
   \[
      B^{\iota,t+1}_i(y)
      :=\widetilde B^{\iota,t}_{(i-1)m_t..im_t}(\INW(y)),
      \quad i\in[n_{t+1}],\quad y\in\{0,1\}^{d_t},
   \]
   where $n_{t+1}=n_t/m_t$ and $\Sigma_{t+1}=\{0,1\}^{d_t}$.
   And we define the new generator $\mathcal G^{t+1}:\{0,1\}^{n_{t+1}\cdot d_t}\to \Sigma^{n}$ by
   \[
   \mathcal G^{t+1}(y_1,\ldots,y_{n_{t+1}}):=
   \widetilde{\mathcal G}^{t}(\INW(y_1),\ldots,\INW(y_{n_{t+1}})), \quad \forall y_1,\ldots,y_{n_{t+1}}\in \{0,1\}^{d_t}.
   \]

   \paragraph{Final generator:}
   We repeat the above procedure for $K$ times, where $K=\left\lceil\frac{\log n}{\log m}\right\rceil$. The final generator is $\mathcal G_x:=\mathcal G^{K}: \{0,1\}^{d_{K-1}}\to \Sigma^n$, where $x=(x_0,\ldots,x_{K-1})$ is the advice. 

   \paragraph{Calculating the error:} We denote the $i$-th average transition matrix of $B^{\iota,t}$ (resp. $\widetilde{B}^{\iota,t}$) by $W^{\iota,t}_i$ (resp. $\widetilde{W}^{\iota,t}_i$). And we use $W^{\iota,t}_{i..j}$ (resp. $\widetilde{W}^{\iota,t}_{i..j}$) to denote the transition matrix of $B^{\iota,t}$ (resp. $\widetilde{B}^{\iota,t}$) from step $i$ to step $j$.

   By the permutation case of \cref{lem:reducing-alphabet-size-with-sampler}, we have $\|W^{\iota,t}_{i..j}-\widetilde{W}^{\iota,t}_{i..j}\|_\infty\leq \delta_{\mathrm{samp}}/\sqrt{w}$ for every $0\leq i\leq j\leq n_t$. Since $\|A\|_2\leq\sqrt{w}\|A\|_\infty$ for every $w\times w$ matrix $A$, the spectral-norm error is at most $\delta_{\mathrm{samp}}$. Applying Cauchy--Schwarz and $2\|x\|_2\|y\|_2\leq\|x\|_2^2+\|y\|_2^2$ gives $\widetilde{W}^{\iota,t}_{i..j}\asvapprox_{0,\delta_{\mathrm{samp}}} W^{\iota,t}_{i..j}$ for every such interval.
   By \cref{lem:inw-sv-generator}, we have $W^{\iota,t+1}_{i}\asvapprox_{\varepsilon_{\mathrm{INW}},0} \widetilde{W}^{\iota,t}_{(i-1)m_t..im_t}$ for every $i\in [n_{t+1}]$.
   By \cref{lem:transitivity}, these two approximations give
   \[
      W^{\iota,t+1}_{i}
      \asvapprox_{\varepsilon_{\mathrm{INW}},
         (1+2\varepsilon_{\mathrm{INW}})\delta_{\mathrm{samp}}}
      W^{\iota,t}_{(i-1)m_t..im_t}
      \quad (i\in[n_{t+1}]).
   \]
   The additive parameter is at most $2\delta_{\mathrm{samp}}$.
   If $n_{t+1}\ge2$, the chosen parameters satisfy the hypotheses of
   \cref{cor:products}, which yields
   \[
      W^{\iota,t+1}_{0..n_{t+1}}
      \asvapprox_{
         (1+6\log n_{t+1}(\varepsilon_{\mathrm{INW}}
            +2n_{t+1}\delta_{\mathrm{samp}}))\varepsilon_{\mathrm{INW}},
         6n_{t+1}\delta_{\mathrm{samp}}}
      W^{\iota,t}_{0..n_t}.
   \]
   Here $6\log n_{t+1}(\varepsilon_{\mathrm{INW}}
   +2n_{t+1}\delta_{\mathrm{samp}})\le1$.
   If $n_{t+1}=1$, the preceding one-block guarantee applies directly.
   Thus every stage has error at most
   \[
      \varepsilon_{\mathrm{step}}:=2\varepsilon_{\mathrm{INW}},
      \qquad
      \delta_{\mathrm{step}}:=6n\delta_{\mathrm{samp}}.
   \]
   Initialize the cumulative errors by
   $\varepsilon^{(0)}=\delta^{(0)}=0$ and, for $0\le t<K$, set
   \[
   \begin{aligned}
      \varepsilon^{(t+1)}
         &=\varepsilon_{\mathrm{step}}
           +(1+2\varepsilon_{\mathrm{step}})\varepsilon^{(t)},\\
      \delta^{(t+1)}
         &=\delta_{\mathrm{step}}
           +(1+2\varepsilon_{\mathrm{step}})\delta^{(t)}.
   \end{aligned}
   \]
   Transitivity gives
   $W^{\iota,t}_{0..n_t}\asvapprox_{\varepsilon^{(t)},\delta^{(t)}}
   W^{\iota,0}_{0..n_0}$.
   Since $K\le\log n$, the parameters give
   $2K\varepsilon_{\mathrm{step}}\le\varepsilon/25$ and
   $2K\delta_{\mathrm{step}}\le\delta$.
   In particular,
   $(1+2\varepsilon_{\mathrm{step}})^K
      \le\exp(2K\varepsilon_{\mathrm{step}})<2$.
   Summing the two geometric recurrences therefore gives
   \[
      \varepsilon^{(K)}\le2K\varepsilon_{\mathrm{step}}\le\varepsilon,
      \qquad
      \delta^{(K)}\le2K\delta_{\mathrm{step}}\le\delta.
   \]
   Hence $W^{\iota,K}_{0..n_K}\asvapprox_{\varepsilon,\delta}
   W^\iota_{0..n}$.

   \paragraph{Complexity of computing the advice and the generator:} The total advice used by the generator is $|x|=|x_0|+\cdots+|x_{K-1}|.$ The first advice $x_0$ has length $|x_0|=O(\log|\Sigma|+\log(nw|\mathcal J|/(\varepsilon\delta)))$. For $t\geq 1$, we have $|x_t|=\log|\Sigma_t|+O(\log(nw|\mathcal J|/(\varepsilon\delta)))$, where $\log|\Sigma_t|=d_{t-1}=p_{t-1}+O(\log m_{t-1}\cdot(\log\log m_{t-1}+\log(1/\varepsilon_{\mathrm{INW}})))=O(\log(nw|\mathcal J|/(\varepsilon\delta)))$. Since $K=\left\lceil\frac{\log n}{\log m}\right\rceil=O(1+\log n\cdot \frac{\log(\log n/\varepsilon_{\mathrm{INW}})}{\log(nw|\mathcal J|/(\delta\varepsilon))})$, we have $|x|=O(\log|\Sigma|+\left\lceil\log n\cdot \frac{\log(\log n/\varepsilon_{\mathrm{INW}})}{\log(nw|\mathcal J|/(\delta\varepsilon))}\right\rceil\cdot \log(nw|\mathcal J|/(\delta\varepsilon)))=O(\log|\Sigma|+\log n(\log\log n+\log(1/\varepsilon))+\log(nw|\mathcal J|/\delta))$.

   The final online seed length of the generator is $d=d_{K-1}=p_{K-1}+O(\log m_{K-1}\cdot(\log\log m_{K-1}+\log(1/\varepsilon_{\mathrm{INW}})))=O(\log(nw|\mathcal J|/\varepsilon\delta))$.
   
   At every level, by \cref{lem:reducing-alphabet-size-with-sampler}, the computing of $x_t$ takes time
   \(\poly(n,w,|\Sigma|,|\mathcal J|,1/\varepsilon,1/\delta)\).  Giving access to $x_0, \ldots , x_{t-1}$,  the algorithm uses
   \(O(\log(nw|\Sigma||\mathcal J|/\varepsilon\delta))\) workspace to compute $x_t$.
\paragraph{Strong explicitness.}
Finally, we show how to compute a single output symbol without expanding any
intermediate string.  Write \(\INW_t\) and \(\Samp_t\) for the INW generator
and the fixed sampler used at level \(t\).  For \(q\in[n]\), define
\[
   q_t=\left\lceil\frac{q}{m^t}\right\rceil
   \quad\text{and}\quad
   j_t=q_t-(q_{t+1}-1)m_t\in[m_t].
\]
Each coordinate of \(\INW_t\) can be evaluated
in polynomial time in its seed length and \(\log m\), using workspace
linear in its seed length.
Starting with the final seed \(s_K=y\), we trace the requested coordinate
\(q\) through the recursion from level \(K\) down to level \(0\).  Suppose
that \(s_{t+1}\) is the symbol at position \(q_{t+1}\) in the
level-\((t+1)\) input.  The position \(q_t\) lies in its \(j_t\)-th child
block.  Hence the corresponding level-\(t\) symbol is obtained by taking
the \(j_t\)-th output of the INW generator and then applying the sampler
fixed by \(x_t\):
\[
   z_t=\bigl(\INW_t(s_{t+1})\bigr)_{j_t},
   \qquad
   s_t=\Samp_t(x_t,z_t).
\]
Repeating this computation for \(t=K-1,\ldots,0\) yields
\[
   s_0=\mathcal G_x(y)_q.
\]
Thus, each level requires only one coordinate evaluation of the INW
generator and one evaluation of the sampler.  Since there are
\(K=O(\log n)\) levels, the computation uses
\(O(d+\log n+\log|\Sigma|)\) workspace and runs in time
\(\poly(|x|,d,\log n,\log|\Sigma|)\).  Hence \(\mathcal G_x\) is strongly
explicit with the claimed bounds.

\end{proof}

\subsection{Proof of \cref{thm:regular-derand}}
\label{sec:perm-derand-proof}

\begin{proof}
   We first apply \Cref{lem:reg2perm-phase} to lift the regular ROBP \(B\) to a permutation ROBP \(B^p\) of length \(n\), width \(w|\Sigma|\), and alphabet \(\Sigma\). The states of \(B^p\) are pairs \((u,\sigma)\in[w]\times\Sigma\), where the auxiliary phase \(\sigma\) indexes the \(|\Sigma|\) lifted copies of each state \(u\). Every symbol induces a permutation on the lifted state space.

    Let \(W_{0..n}\) and \(W^p_{0..n}\) denote the full transition matrices of \(B\) and \(B^p\), respectively. The transition matrix of \(B\) is recovered from \(W^p_{0..n}\) through the following sum identity: for every \(u,v\in[w]\) and any fixed initial phase \(\sigma\in\Sigma\),
    \[
    \sum_{\sigma'\in\Sigma}\bigl[W^p_{0..n}\bigr]_{(u,\sigma'),(v,\sigma)}=\bigl[W_{0..n}\bigr]_{u,v}.
    \]

   We use \cref{lem:det-compressed-inw-sv} to compute the offline advice $x$ and compute the targeted generator $\mathcal G_x:\{0,1\}^{d}\to \Sigma^n$ that $(\varepsilon/2,\varepsilon/2)$-SV fools the permutation ROBP $B^p$, where $d=O(\log(nw|\Sigma|/\varepsilon))$ is the seed length of the generator. 
   Enumerating the online seed gives the approximate transition matrix
   \(\widehat{W}^p_{0..n}\asvapprox_{\varepsilon/2,\varepsilon/2}W^p_{0..n}\).
   We then output the approximate transition matrix $\widehat{W}_{0..n}$ defined by $[\widehat{W}_{0..n}]_{u,v}=\frac{1}{|\Sigma|}\sum_{\sigma,\sigma'\in\Sigma} [\widehat{W}^p_{0..n}]_{(u,\sigma),(v,\sigma')}$ for every $u,v\in [w]$.

   \paragraph{Error analysis:} For every $u,v\in [w]$, let $y=\frac{1}{\sqrt{|\Sigma|}}\sum_{\sigma\in\Sigma} e_{(u,\sigma)}$ and $x=\frac{1}{\sqrt{|\Sigma|}}\sum_{\sigma'\in\Sigma} e_{(v,\sigma')}$, where $e_{(u,\sigma)}$ is the unit vector which has $1$ at the $(u,\sigma)$-th entry and $0$ at other entries. By the SV approximation guarantee, we have
   \begin{align*}
      &\left|y^\top (\widehat{W}^p_{0..n}-W^p_{0..n})x\right|\leq\frac{\varepsilon}{4}(\|x\|_2^2-\|W^p_{0..n}x\|_2^2+\|y\|_2^2-\|(W^p_{0..n})^\top y\|_2^2)+\frac{\varepsilon}{4}(\|x\|_2^2+\|y\|_2^2)\\
      \implies&\left|\frac{1}{|\Sigma|}\sum_{\sigma,\sigma'\in\Sigma} [\widehat{W}^p_{0..n}]_{(u,\sigma),(v,\sigma')}-\frac{1}{|\Sigma|}\sum_{\sigma,\sigma'\in\Sigma} [W^p_{0..n}]_{(u,\sigma),(v,\sigma')}\right|\leq \frac{\varepsilon+\varepsilon}{4}(\|x\|_2^2+\|y\|_2^2)\\
      \implies&\left|[\widehat{W}_{0..n}]_{u,v}-[W_{0..n}]_{u,v}\right|\leq \varepsilon.
   \end{align*}

   \paragraph{Space and time analysis:} By \cref{lem:reg2perm-phase}, transforming the regular ROBP to a permutation ROBP can be done in space $O(\log (nw|\Sigma|))$ and time $\poly(n,w,|\Sigma|)$. 
   By  \cref{lem:det-compressed-inw-sv}, constructing $x$ can be done in space $O(\log n(\log\log n+\log(1/\varepsilon))+\log(nw|\Sigma|/\varepsilon))$ and time $\poly(n,w,|\Sigma|,1/\varepsilon)$. The generator $\mathcal G_x$ has seed length $d=O(\log(nw|\Sigma|/\varepsilon))$ and
   can be computed in space $O(d + \log n + \log |\Sigma|)$ and time \(\poly(|x|,d,\log n,\log|\Sigma|)\), given access to $x$.
   Enumerating the online seed of the generator and computing the approximate transition matrix $\widehat{W}_{0..n}$ can be done in space $O(d+|x|+\log(nw|\Sigma|))=O(\log n(\log\log n+\log(1/\varepsilon))+\log(nw|\Sigma|/\varepsilon))$ and time $\poly(2^d,n,w,|\Sigma|)=\poly(n,w,|\Sigma|,1/\varepsilon)$.
   The total space used by the algorithm is $O(\log n(\log\log n+\log(1/\varepsilon))+\log(nw|\Sigma|/\varepsilon))$ and the total time used by the algorithm is $\poly(n,w,|\Sigma|,1/\varepsilon)$.
\end{proof}

\section{Derandomization of short-wide regular \texorpdfstring{ROBP}{ROBP}s in SC}
\label{sec:perm-regular-derand}

We prove \cref{thm:regular-derand-low-error} in this section.  Again we first prove a derandomization for permutation ROBPs and then extend it to regular ROBPs.

The main technical lemma of our SC derandomization for permutation ROBPs in the short-wide regime is stated as the following.
\begin{lemma}[Low-error derandomization for permutation ROBPs]
\label{lem:perm-derand-low-error}
Let \(n,w\in\mathbb N\), let \(\Sigma\) be a finite alphabet, and let
\(\varepsilon>0\).  There exists a deterministic algorithm that, given a
length-\(n\), width-\(w\) permutation ROBP \(B\) over alphabet \(\Sigma\),
computes a matrix \(\widehat W\in\mathbb R^{w\times w}\) such that
\[
   \left|
      [W_{0..n}]_{u,v}-[\widehat W]_{u,v}
   \right|
   \le \varepsilon
   \qquad\text{for every }u,v\in[w],
\]
where \(W_{0..n}\) is the average transition matrix of \(B\) from its first
layer to its last layer.  The algorithm runs in time
\(\poly(n,w,|\Sigma|,1/\varepsilon)\) and uses space
\[
   O\bigl(\log|\Sigma|+\log^2 n+\log(1/\varepsilon)+\log w\bigr).
\]
\end{lemma}

The proof of the lemma relies on the following two-parameter variant of the recursive error reduction of \cite{chattopadhyayRecursiveErrorReduction2023}.
It tracks the two error parameters of our SV approximation 
separately.

Starting from \((0,n]\), iteratively split every interval into its left and
right halves until all intervals have length one.  Let dyadic tree \(\mathcal D\) denote
the resulting collection of intervals for all layers. Notice that the root is \((0,n]\), the leaves are the unit intervals \((j,j+1], j\in \{0, \ldots,  n-1\}\).
\begin{lemma}[Recursive error reduction for two-parameter SV approximation]
\label{lem:recursive-sv-separated-slack}
Let \(n=2^L\) with \(L\ge1\).  Suppose that each \(J\in\mathcal D\) is associated
with a doubly stochastic matrix \(M_J\), and that every nonleaf interval
\(J\), with left and right children \(J_{\mathrm L}\) and
\(J_{\mathrm R}\), satisfies
\[
   M_J=M_{J_{\mathrm L}}M_{J_{\mathrm R}}.
\]
Let \(0<\gamma<1/2\), and define
\begin{equation}
\label{eq:eps-delta-sequences}
\begin{aligned}
   \varepsilon^{(i)}
      &:=\frac{\gamma^{i+1}}{100L(i+1)^2},
   &\qquad
   \delta^{(i)}
      &:=\frac{\varepsilon^{(i)}}{n^2},\\
   C_t
      &:=\left(1+\frac1{2L}\right)^t,
   &
   S_t
      &:=2^tC_t.
\end{aligned}
\end{equation}

For every leaf \(J\), set \(\widehat M_J^{(h)}:=M_J\) for all \(h\ge0\).
For every nonleaf \(J\), suppose that a base approximation
\(\widehat M_J^{(0)}\) is given, and define, for \(h\ge1\),
\begin{equation}
\label{eq:wprr-matrix-recursion}
   \widehat M_J^{(h)}
   :=
   \sum_{p+q=h}
      \widehat M_{J_{\mathrm L}}^{(p)}
      \widehat M_{J_{\mathrm R}}^{(q)}
   -
   \sum_{p+q=h-1}
      \widehat M_{J_{\mathrm L}}^{(p)}
      \widehat M_{J_{\mathrm R}}^{(q)}.
\end{equation}
Assume that every nonleaf \(J\in\mathcal D\) satisfies
\begin{equation}
\label{eq:base-eps-delta-assumption}
   \widehat M_J^{(0)}
   \asvapprox_{\varepsilon^{(0)},\delta^{(0)}}
   M_J.
\end{equation}
Then, for every \(k\ge0\) and every interval \(J\in\mathcal D\) of length
\(2^t\),
\begin{equation}
\label{eq:scaled-slack-induction-conclusion}
   \widehat M_J^{(k)}
   \asvapprox_{C_t\varepsilon^{(k)},S_t\delta^{(k)}}
   M_J.
\end{equation}
\end{lemma}

Under our column-vector convention, the transition matrices satisfy
\(W_J=W_{J_{\mathrm R}}W_{J_{\mathrm L}}\), so we apply the lemma to
\(M_J=W_J^\top\).  Two-parameter SV approximation is invariant under
transposition, since \(L(W^\top)=R(W)\) and \(R(W^\top)=L(W)\).
Now \cref{lem:perm-derand-low-error} can be shown using \cref{lem:recursive-sv-separated-slack}.
\begin{proof}[Proof of \cref{lem:perm-derand-low-error}]

We may assume that \(0<\varepsilon\le 1/4\), since for larger
\(\varepsilon\) it suffices to prove the statement with error \(1/4\).
By padding \(B\) with identity layers, we may further assume that \(n\) is a power of two.  This increases the length by less than a factor
of two and preserves the permutation property.  Henceforth, \(n\) denotes the
padded length, and we write \(L=\log_2 n\).

\paragraph{The recursive approximation.}
Starting from \((0,n]\), recursively split every interval into its left and
right halves until all intervals have length one.  Let \(\mathcal D\) denote
the resulting collection of intervals.  For each \(J=(a,b]\in\mathcal D\), let
\[
   W_J:=W_{a..b}
\]
denote the average transition matrix of the subprogram of \(B\) on \(J\).
  Set
\[
   \gamma:=\frac1n,
   \qquad
   k:=\left\lceil\frac{\log(1/\varepsilon)}{\log n}\right\rceil.
\]
With the choice \(\gamma=1/n\),  we further setup the following parameters
\[
   \varepsilon^{(0)}
      =\frac{\gamma}{100L}
      =\frac1{100Ln},
   \qquad
   \delta^{(0)}
      =\frac{\gamma}{100Ln^2}
      =\frac1{100Ln^3}.
\]
Applying \cref{lem:det-compressed-inw-sv} to the subprograms on all intervals
\(J\in\mathcal D\) with \(|J|>1\), we obtain a single offline seed \(x\) and
a targeted generator
\[
   \mathcal G_x:\{0,1\}^d\longrightarrow\Sigma^n
\]
such that, for every such interval \(J=(a,b]\),
\[
   \E_{y\sim U_d}
   \left[B_{a..b}\bigl(\mathcal G_x(y)|_J\bigr)\right]
   \asvapprox_{\varepsilon^{(0)},\delta^{(0)}} W_J.
\]
Its online seed length is
\[
   d=O(\log(nw/\varepsilon)),
\]
and its advice length is
\[
   |x|
   =O\bigl(\log|\Sigma|+\log^2n+\log(nw/\varepsilon)\bigr).
\]

For every interval \(J=(a,b]\in\mathcal D\) with \(|J|>1\), define
\[
   \widehat W_J^{(0)}
   :=
   \E_{y\sim U_d}
   \left[
      B_{a..b}\bigl(\mathcal G_x(y)|_J\bigr)
   \right].
\]
For every \(J\in\mathcal D\), set \(M_J:=W_J^\top\).
On every multi-step interval, set
\(\widehat M_J^{(0)}:=(\widehat W_J^{(0)})^\top\).
For every \(J=(a,a+1]\), set
\(\widehat M_J^{(h)}:=M_J\) for all \(h\ge0\).
Define the higher-order matrices \(\widehat M_J^{(h)}\) by
\eqref{eq:wprr-matrix-recursion}, and set
\(\widehat W_J^{(h)}:=(\widehat M_J^{(h)})^\top\).
Thus the one-step transition matrices are
exact, whereas every longer dyadic subprogram is approximated using the generator $G_x$.

We expand \(\widehat M_{0..n}^{(k)}\) via
\eqref{eq:wprr-matrix-recursion} until every factor has order zero,
and then transpose to obtain an expansion of \(\widehat W_{0..n}^{(k)}\).
Let \(\mathcal S\) denote a multiset: each element in \(\mathcal S\)
is specified by a sign \(\lambda\in\{-1,+1\}\) and a sequence
\[
   \mathbf i=(i_0,\ldots,i_r),
   \qquad
   0=i_0<i_1<\cdots<i_r=n,
\]
where the consecutive intervals
\[
   (i_0,i_1],(i_1,i_2],\ldots,(i_{r-1},i_r]
\]
partition \((0,n]\).  With this notation, the expansion can be denoted as
\begin{equation}
\label{eq:signed-short-expansion}
   \widehat W_{0..n}^{(k)}
   =
   \sum_{(\mathbf i,\lambda)\in\mathcal S}
      \lambda
      \widehat W_{i_{r-1}..i_r}^{(0)}
      \cdots
      \widehat W_{i_1..i_2}^{(0)}
      \widehat W_{i_0..i_1}^{(0)},
\end{equation}
where $\mathcal S$ is the multiset of signed terms produced by
the recursive expansion, with multiplicities retained.
Applying \cite[Lemma~25]{chattopadhyayRecursiveErrorReduction2023}
to \eqref{eq:wprr-matrix-recursion} gives
\[
   |\mathcal S|\le n^{2k}
      \le n^2/\varepsilon^2
      =\poly(n,1/\varepsilon),
\]
and every term has $r$ factors with
\[
   r\le k\log_2 n+1
      =O(\log n+\log(1/\varepsilon)),
\]
where we used
$k=\lceil\log(1/\varepsilon)/\log n\rceil$.

\paragraph{Approximating the expansion terms as short programs.}

We next realize each product in \eqref{eq:signed-short-expansion} as the
average transition matrix of a short permutation ROBP.  Fix
\((\mathbf i,\lambda)\in\mathcal S\), and define a length-\(r\), width-\(w\)
permutation ROBP \(B^{\mathbf i}\) over the alphabet
\[
   \Sigma':=\{0,1\}^d\times\Sigma.
\]
For \(q\in[r]\), the \(q\)-th transition on
\((y,\sigma)\in\Sigma'\) is
\[
   B_q^{\mathbf i}(y,\sigma)
   :=
   \begin{cases}
      B_{i_{q-1}..i_q}
      \bigl(\mathcal G_x(y)|_{(i_{q-1},i_q]}\bigr),
         & i_q-i_{q-1}>1,\\[2mm]
      B_{i_{q-1}..i_q}(\sigma),
         & i_q-i_{q-1}=1.
   \end{cases}
\]
In the first case, this transition is the composition of the permutation
transitions of \(B\) along the interval \((i_{q-1},i_q]\); in the second
case, it is a single transition of \(B\).  Thus, \(B^{\mathbf i}\) is a
permutation ROBP.
When
\(i_q-i_{q-1}>1\), averaging over \((y,\sigma)\in\Sigma'\) 
gives
\(\widehat W_{i_{q-1}..i_q}^{(0)}\).  When \(i_q-i_{q-1}=1\), and averaging over \(\sigma\) gives the exact one-step
matrix
\(\widehat W_{i_{q-1}..i_q}^{(0)}=W_{i_{q-1}..i_q}\).
Therefore the average transition matrix of \(B^{\mathbf i}\) is
\begin{equation}
\label{eq:short-program-realization}
   W[B^{\mathbf i}]
   =
   \widehat W_{i_{r-1}..i_r}^{(0)}
   \cdots
   \widehat W_{i_1..i_2}^{(0)}
   \widehat W_{i_0..i_1}^{(0)}.
\end{equation}

We now approximate each short program \(B^{\mathbf i}\) using the weighted
generator of \cite{chengWuWeightedPseudorandomGenerators2026} for permutation ROBPs, with error 
\[
   \eta:=\frac{\varepsilon}{2|\mathcal S|}.
\]
Note that applying the weighted generator to
\(B^{\mathbf i}\) with entrywise error \(\eta\) requires seed length
\begin{align*}
   d_{\mathrm{CW}}
   &=O\left(
      \log|\Sigma'|
      +\log r\log\log r
      +\log r\sqrt{\log(|\mathcal S|/\varepsilon)}
      +\log(|\mathcal S|/\varepsilon)
   \right)\\
   &=O\bigl(\log|\Sigma|+\log(nw/\varepsilon)\bigr),
\end{align*}
where the second bound uses the bounds on \(r\) and \(|\mathcal S|\) above,
together with \(\log|\Sigma'|=d+\log|\Sigma|\).

Enumerating the weighted seeds produces a matrix
\(\widetilde W^{\mathbf i}\) such that, for every \(u,v\in[w]\),
\[
   \left|
      [\widetilde W^{\mathbf i}]_{u,v}
      -[W[B^{\mathbf i}]]_{u,v}
   \right|
   \le \eta.
\]
We then combine these approximations with the signs from
\eqref{eq:signed-short-expansion} and define
\[
   \widetilde W
   :=
   \sum_{(\mathbf i,\lambda)\in\mathcal S}
      \lambda\widetilde W^{\mathbf i}.
\]

\paragraph{Calculating the error}
We now bound the error of the final matrix \(\widetilde W\).  There are two
approximation steps: recursive error reduction replaces \(W_{0..n}\) by
\(\widehat W_{0..n}^{(k)}\), and the weighted generator then approximates the
signed expansion of \(\widehat W_{0..n}^{(k)}\).

We first bound the error from recursive error reduction. 
The choice of \(x\) and transposition invariance give the base assumption
of \cref{lem:recursive-sv-separated-slack} for the matrices
\(M_J=W_J^\top\).  Applying the lemma and transposing back gives
\[
   \widehat W_{0..n}^{(k)}
   \asvapprox_{C_L\varepsilon^{(k)},S_L\delta^{(k)}}
   W_{0..n}.
\]
We now bound the two error parameters.  By the choice of \(k\),
\[
   n^{-k}\le \varepsilon.
\]
Moreover, since \(n\ge4\) after padding, we have \(L\ge2\), and hence
\[
   C_L=\left(1+\frac1{2L}\right)^L<2
   \le L(k+1)^2.
\]
Therefore,
\[
\begin{aligned}
   C_L\varepsilon^{(k)}
   &=
   \frac{C_L}{100L(k+1)^2}\cdot\frac1{n^{k+1}}\\
   &\le
   \frac1{100n}\cdot\frac1{n^k}
   \le
   \frac{\varepsilon}{100n}.
\end{aligned}
\]
The additive error is smaller by another factor of \(n\): since
\(S_L=nC_L\) and \(\delta^{(k)}=\varepsilon^{(k)}/n^2\),
\[
   S_L\delta^{(k)}
   =
   \frac{C_L\varepsilon^{(k)}}{n}
   \le
   \frac{\varepsilon}{100n^2}.
\]

 Fix
\(u,v\in[w]\), and let \(e_u,e_v\in\mathbb R^w\) denote the vectors with a
single \(1\) in coordinates \(u\) and \(v\), respectively, and \(0\) elsewhere.
for every \(w\times w\) matrix \(A\), the two-parameter SV approximation at
the root gives
\[
\begin{aligned}
   \left|
      [\widehat W_{0..n}^{(k)}]_{u,v}
      -[W_{0..n}]_{u,v}
   \right|
   &\le
   \frac{C_L\varepsilon^{(k)}}2
   \left(
      D(W_{0..n}^\top,e_u)
      +D(W_{0..n},e_v)
   \right)\\
   &\quad+
   \frac{S_L\delta^{(k)}}2
   \left(
      \|e_u\|_2^2+\|e_v\|_2^2
   \right).
\end{aligned}
\]
The vectors \(e_u\) and \(e_v\) have Euclidean norm one.  Moreover,
\(W_{0..n}\) and \(W_{0..n}^\top\) are doubly stochastic contractions, so
\[
   D(W_{0..n}^\top,e_u)\le1,
   \qquad
   D(W_{0..n},e_v)\le1.
\]
Using the bounds obtained above,
\[
   C_L\varepsilon^{(k)}
      \le\frac{\varepsilon}{100n},
   \qquad
   S_L\delta^{(k)}
      \le\frac{\varepsilon}{100n^2},
\]
we obtain
\[
   \left|
      [\widehat W_{0..n}^{(k)}]_{u,v}
      -[W_{0..n}]_{u,v}
   \right|
   \le
   \frac{\varepsilon}{100n}
   +\frac{\varepsilon}{100n^2}
   \le\frac{\varepsilon}{2}.
\]

We next bound the error introduced by the weighted generator.  By
\eqref{eq:signed-short-expansion} and
\eqref{eq:short-program-realization},
\[
   \widehat W_{0..n}^{(k)}
   =
   \sum_{(\mathbf i,\lambda)\in\mathcal S}
      \lambda W[B^{\mathbf i}],
\]
whereas
\[
   \widetilde W
   =
   \sum_{(\mathbf i,\lambda)\in\mathcal S}
      \lambda\widetilde W^{\mathbf i}.
\]
Therefore, for the same \(u,v\),
\begin{align*}
   \left|
      [\widetilde W]_{u,v}
      -[\widehat W_{0..n}^{(k)}]_{u,v}
   \right|
   &\le
   \sum_{(\mathbf i,\lambda)\in\mathcal S}
      \left|
         [\widetilde W^{\mathbf i}]_{u,v}
         -[W[B^{\mathbf i}]]_{u,v}
      \right|\\
   &\le
   |\mathcal S|\eta
   =\frac{\varepsilon}{2}.
\end{align*}
Here the signs \(\lambda\) disappear after taking absolute values, and the
last inequality follows from the choice
\(\eta=\varepsilon/(2|\mathcal S|)\).

Combining the two approximation bounds by the triangle inequality gives
\[
\begin{aligned}
   \left|
      [\widetilde W]_{u,v}
      -[W_{0..n}]_{u,v}
   \right|
   &\le
   \left|
      [\widetilde W]_{u,v}
      -[\widehat W_{0..n}^{(k)}]_{u,v}
   \right|\\
   &\quad+
   \left|
      [\widehat W_{0..n}^{(k)}]_{u,v}
      -[W_{0..n}]_{u,v}
   \right|\\
   &\le\varepsilon.
\end{aligned}
\]
Thus \(\widetilde W\) approximates \(W_{0..n}\) entrywise within
\(\varepsilon\).

\paragraph{Space and time analysis}
We finish by verifying the time and space bounds.  The offline seed \(x\)
provided by \cref{lem:det-compressed-inw-sv} can be computed and stored in
time \(\poly(n,w,|\Sigma|,1/\varepsilon)\) and space
\[
   O\bigl(\log|\Sigma|+\log^2 n+\log(1/\varepsilon)+\log w\bigr).
\]
The signed expansion is then enumerated one term at a time.  Given the index
of a term and \(q\in[r]\), its sign and the interval
\((i_{q-1},i_q]\) can be recomputed from the recursion, so the entire sequence
\(\mathbf i\) need not be stored.

For each expansion term, we enumerate the weighted seeds and simulate the
corresponding short program \(B^{\mathbf i}\) sequentially.  To evaluate its
\(q\)-th transition on \((y,\sigma)\in\Sigma'\), we first recompute the
interval \((i_{q-1},i_q]\).  If this interval has length one, we evaluate the
corresponding transition of \(B\) directly.  Otherwise, we scan the interval
and compute each required symbol of
\(\mathcal G_x(y)|_{(i_{q-1},i_q]}\) on demand.  By the strong explicitness
of \(\mathcal G_x\), each such symbol can be computed from \(x\), \(y\), and
its position using
\[
   O(d+\log n+\log|\Sigma|)
\]
workspace, without storing the full output \(\mathcal G_x(y)\).
At any point in the simulation, apart from the stored offline seed \(x\), we
only need to keep the current state of \(B^{\mathbf i}\), the current weighted
seed, and the workspace for evaluating \(\mathcal G_x\).  This requires
\[
   O\bigl(d_{\mathrm{CW}}+d+\log n+\log|\Sigma|+\log w\bigr)
   =
   O\bigl(\log|\Sigma|+\log(nw/\varepsilon)\bigr)
\]
additional space.
Thus the storage of \(x\) dominates, and the total space is
\[
   O\bigl(\log|\Sigma|+\log^2 n+\log(1/\varepsilon)+\log w\bigr).
\]

Finally, the signed expansion has polynomially many terms, the weighted
generator has polynomially many weighted seeds, and each transition of
\(B^{\mathbf i}\) is evaluated using at most \(n\)
computations of \(\mathcal G_x\).  Each of these computations takes
\(\poly(|x|,d,\log n,\log|\Sigma|)\) time.  Hence the total running time is
\(\poly(n,w,|\Sigma|,1/\varepsilon)\).
\end{proof}

Now we prove \cref{thm:regular-derand-low-error}, by using our main lemma \cref{lem:perm-derand-low-error} and the transformation from regular ROBPs to permutation ROBPs.

\begin{proof}[Proof of \cref{thm:regular-derand-low-error}]
Apply \cref{lem:reg2perm-phase} to convert \(B\) into a permutation ROBP
\(B^p\) of length \(n\) and width \(w|\Sigma|\) over the same alphabet.
In \(B^p\), each state \(v\) of \(B\) is replaced by \(|\Sigma|\) copies
\((v,\tau)\), indexed by \(\tau\in\Sigma\).
Let \(W_{0..n}\) and \(W^p_{0..n}\) denote the average transition matrices of
\(B\) and \(B^p\), respectively.  By \cref{lem:reg2perm-phase}, for every
\(u,v\in[w]\) and every \(\tau\in\Sigma\),
\begin{equation}
\label{eq:reg2perm-copy-sum}
   \sum_{\tau'\in\Sigma}
      [W^p_{0..n}]_{(u,\tau'),(v,\tau)}
   =
   [W_{0..n}]_{u,v}.
\end{equation}
Thus, for any fixed copy \((v,\tau)\) of \(v\), summing the transition
probabilities to all copies \((u,\tau')\) of \(u\) gives the transition
probability from \(v\) to \(u\) in \(B\).

Apply \cref{lem:perm-derand-low-error} to \(B^p\) with entrywise error
\(\varepsilon/|\Sigma|\), obtaining \(\widehat W^p_{0..n}\).  Fix any
\(\tau\in\Sigma\), and define
\[
   [\widehat W_{0..n}]_{u,v}
   :=
   \sum_{\tau'\in\Sigma}
      [\widehat W^p_{0..n}]_{(u,\tau'),(v,\tau)}.
\]
By the definition of \(\widehat W_{0..n}\) and
\eqref{eq:reg2perm-copy-sum},
\begin{align*}
   \left|
      [\widehat W_{0..n}]_{u,v}
      -[W_{0..n}]_{u,v}
   \right|
   &\le
   \sum_{\tau'\in\Sigma}
   \left|
      [\widehat W^p_{0..n}]_{(u,\tau'),(v,\tau)}
      -[W^p_{0..n}]_{(u,\tau'),(v,\tau)}
   \right|\\
   &\le
   |\Sigma|\cdot\frac{\varepsilon}{|\Sigma|}
   =\varepsilon.
\end{align*}

Since \(B^p\) has width \(w|\Sigma|\) and is approximated with error
\(\varepsilon/|\Sigma|\), the space bound from
\cref{lem:perm-derand-low-error} becomes
\[
   O\!\left(
      \log|\Sigma|
      +\log^2 n
      +\log\frac{|\Sigma|}{\varepsilon}
      +\log(w|\Sigma|)
   \right)
   =
   O\!\left(
      \log|\Sigma|
      +\log^2 n
      +\log(1/\varepsilon)
      +\log w
   \right).
\]
The running time remains
\(\poly(n,w,|\Sigma|,1/\varepsilon)\), completing the proof.
\end{proof}

\subsection{Proof of \cref{lem:recursive-sv-separated-slack}}
\label{sec:proof-recursive-sv-separated-slack}

\begin{proof}
We prove \eqref{eq:scaled-slack-induction-conclusion} by induction on the
interval level \(t\), simultaneously for all correction orders \(k\ge0\).
The claim trivially holds for \(t=0\).  Suppose \(t\ge1\) and that the claim holds at
level \(t-1\), and let \(J=(\ell,r]\) be an interval of length \(2^t\).
Write
\[
   J_{\mathrm L}=(\ell,m],
   \qquad
   J_{\mathrm R}=(m,r],
   \qquad
   m=\frac{\ell+r}{2}.
\]

For \(k=0\), since \(C_t,S_t\ge1\), the claim follows from
\eqref{eq:base-eps-delta-assumption}.
Fix \(k\ge1\), and write
\[
   \Delta_J^{(h)}:=\widehat M_J^{(h)}-M_J.
\]
Expanding \eqref{eq:wprr-matrix-recursion} and using
\(M_J=M_{J_{\mathrm L}}M_{J_{\mathrm R}}\) gives
\begin{equation}
\label{eq:slack-delta-identity-separated}
\begin{aligned}
   \Delta_J^{(k)}
   ={}&
   \sum_{p+q=k}
      \Delta_{J_{\mathrm L}}^{(p)}
      \Delta_{J_{\mathrm R}}^{(q)}
   -
   \sum_{p+q=k-1}
      \Delta_{J_{\mathrm L}}^{(p)}
      \Delta_{J_{\mathrm R}}^{(q)}\\
   &\quad+
   \Delta_{J_{\mathrm L}}^{(k)}M_{J_{\mathrm R}}
   +M_{J_{\mathrm L}}\Delta_{J_{\mathrm R}}^{(k)}.
\end{aligned}
\end{equation}
Fix \(x,y\in\mathbb R^w\).  We bound the quadratic and linear terms in
\(x^\top\Delta_J^{(k)}y\) separately.

\paragraph{Quadratic terms.}
Fix \(p+q\in\{k-1,k\}\), and set
\[
   a_h:=C_{t-1}\varepsilon^{(h)},
   \qquad
   b_h:=S_{t-1}\delta^{(h)}.
\]
The induction hypothesis for the two children gives
\[
   \widehat M_{J_{\mathrm L}}^{(p)}
   \asvapprox_{a_p,b_p}
   M_{J_{\mathrm L}},
   \qquad
   \widehat M_{J_{\mathrm R}}^{(q)}
   \asvapprox_{a_q,b_q}
   M_{J_{\mathrm R}}.
\]
Applying \Cref{lem:one-sided} to these two approximations gives
\begin{align}
\label{eq:quadratic-one-sided}
   \norm{(\Delta_{J_{\mathrm L}}^{(p)})^\top x}_2
   &\le
   \sqrt{(a_p+b_p)
      \left(
         a_p\norm{x}_{L(M_{J_{\mathrm L}})}^2
         +b_p\norm{x}_2^2
      \right)},\notag\\
   \norm{\Delta_{J_{\mathrm R}}^{(q)}y}_2
   &\le
   \sqrt{(a_q+b_q)
      \left(
         a_q\norm{y}_{R(M_{J_{\mathrm R}})}^2
         +b_q\norm{y}_2^2
      \right)}.
\end{align}

For every correction order \(h\), the ratio of the two error parameters is
\[
\begin{aligned}
   \frac{b_h}{a_h}
   &=
   \frac{S_{t-1}\delta^{(h)}}
        {C_{t-1}\varepsilon^{(h)}}
   &=
   \frac{2^{t-1}}{n^2}
   =:\rho_t
   \qquad
   \text{by \eqref{eq:eps-delta-sequences}.}
\end{aligned}
\]
Since \(2^t\le n\),
\[
   0\le\rho_t
   \le1.
\]
In particular, \(b_h=\rho_ta_h\).  Substituting this identity into
\eqref{eq:quadratic-one-sided}gives
\begin{align}
\label{eq:quadratic-one-sided-rho}
   \norm{(\Delta_{J_{\mathrm L}}^{(p)})^\top x}_2
   &\le
   a_p\sqrt{(1+\rho_t)
      \left(
         \norm{x}_{L(M_{J_{\mathrm L}})}^2
         +\rho_t\norm{x}_2^2
      \right)},\notag\\
   \norm{\Delta_{J_{\mathrm R}}^{(q)}y}_2
   &\le
   a_q\sqrt{(1+\rho_t)
      \left(
         \norm{y}_{R(M_{J_{\mathrm R}})}^2
         +\rho_t\norm{y}_2^2
      \right)}.
\end{align}

We can now bound the quadratic term as follows:
\begin{align}
\label{eq:verified-quadratic-term}
   \left|
      x^\top
      \Delta_{J_{\mathrm L}}^{(p)}
      \Delta_{J_{\mathrm R}}^{(q)}y
   \right|
   &\le
   \norm{(\Delta_{J_{\mathrm L}}^{(p)})^\top x}_2
   \norm{\Delta_{J_{\mathrm R}}^{(q)}y}_2
   && \text{by Cauchy--Schwarz}
   \notag\\
   &\le
   (1+\rho_t)a_pa_q
   \sqrt{
      \left(
         \norm{x}_{L(M_{J_{\mathrm L}})}^2
         +\rho_t\norm{x}_2^2
      \right)
      \left(
         \norm{y}_{R(M_{J_{\mathrm R}})}^2
         +\rho_t\norm{y}_2^2
      \right)}
   && \text{by \eqref{eq:quadratic-one-sided-rho}}
   \notag\\
   &\le
   \frac{1+\rho_t}{2}a_pa_q
   \Bigl(
      \norm{x}_{L(M_{J_{\mathrm L}})}^2
      +\norm{y}_{R(M_{J_{\mathrm R}})}^2
      +\rho_t\norm{x}_2^2
      +\rho_t\norm{y}_2^2
   \Bigr)
   && \text{by AM--GM}
   \notag\\
   &\le
   a_pa_q
   \left(
      \norm{x}_{L(M_{J_{\mathrm L}})}^2
      +\norm{y}_{R(M_{J_{\mathrm R}})}^2
   \right)
   +
   \rho_ta_pa_q
   \left(
      \norm{x}_2^2+\norm{y}_2^2
   \right)
   && \text{since \(0\le\rho_t\le1\).}
\end{align}\
From the definition of  $\norm{\cdot }_{R(W)}^2,\norm{\cdot }_{L(W)}^2$, we have
\begin{align}
\label{eq:product-energy-identities}
   \norm{y}_{R(M_J)}^2
   &=
   \norm{y}^2-\norm{M_{J_L}M_{J_R}y}^2
   =
   \norm{y}_{R(M_{J_{\mathrm R}})}^2
   +
   \norm{M_{J_{\mathrm R}}y}_{R(M_{J_{\mathrm L}})}^2,
   \notag\\
   \norm{x}_{L(M_J)}^2
   &=
   \norm{x}^2-\norm{M_{J_R}^T M_{J_L}^Tx}^2
   =
   \norm{x}_{L(M_{J_{\mathrm L}})}^2
   +
   \norm{M_{J_{\mathrm L}}^\top x}_{L(M_{J_{\mathrm R}})}^2.
\end{align}
In particular,
\[
   \norm{x}_{L(M_{J_{\mathrm L}})}^2
   +
   \norm{y}_{R(M_{J_{\mathrm R}})}^2
   \le
   \norm{x}_{L(M_J)}^2
   +
   \norm{y}_{R(M_J)}^2.
\]
Hence \eqref{eq:verified-quadratic-term} implies, for every
\(p+q\in\{k-1,k\}\),
\begin{align}
\label{eq:quadratic-term-parent-energy}
   \left|
      x^\top
      \Delta_{J_{\mathrm L}}^{(p)}
      \Delta_{J_{\mathrm R}}^{(q)}y
   \right|
   &\le
   a_pa_q
   \left(
      \norm{x}_{L(M_J)}^2
      +\norm{y}_{R(M_J)}^2
   \right)\notag\\
   &\quad+
   \rho_ta_pa_q
   \left(
      \norm{x}_2^2+\norm{y}_2^2
   \right).
\end{align}

It remains to sum the coefficients in
\eqref{eq:quadratic-term-parent-energy}.  Set
\[
   A_k:=
   \sum_{p+q=k}\varepsilon^{(p)}\varepsilon^{(q)}
   +
   \sum_{p+q=k-1}\varepsilon^{(p)}\varepsilon^{(q)}.
\]
By the definition of \(\varepsilon^{(h)}\) and the calculation in
\Cref{sec:scaled-convolution-proof},
\begin{equation}
\label{eq:scaled-convolution}
   A_k
   \le
   \frac{12}{100L}\varepsilon^{(k)}
   \le
   \frac{\varepsilon^{(k)}}{8L}.
\end{equation}
Since \(a_h=C_{t-1}\varepsilon^{(h)}\), we obtain
\begin{equation}
\label{eq:quadratic-coefficient-bounds}
\begin{aligned}
   \sum_{p+q\in\{k-1,k\}}a_pa_q
   &=
   C_{t-1}^2A_k
   \le
   \frac{C_{t-1}\varepsilon^{(k)}}{4L},\\
   \rho_t
   \sum_{p+q\in\{k-1,k\}}a_pa_q
   &\le
   \frac{S_{t-1}\delta^{(k)}}{4L}.
\end{aligned}
\end{equation}
Here we used \(C_{t-1}<2\), together with
\[
   \rho_t C_{t-1}\varepsilon^{(k)}
   =
   S_{t-1}\delta^{(k)}
\]
for the second inequality.

Summing \eqref{eq:quadratic-term-parent-energy} over
\(p+q=k\) and \(p+q=k-1\), and applying
\eqref{eq:quadratic-coefficient-bounds}, shows that all quadratic terms
contribute at most
\begin{equation}
\label{eq:quadratic-clean-bound}
   \frac{C_{t-1}\varepsilon^{(k)}}{2L}
   \frac{
      \norm{x}_{L(M_J)}^2+\norm{y}_{R(M_J)}^2
   }{2}
   +
   \frac{S_{t-1}\delta^{(k)}}{2L}
   \frac{
      \norm{x}_2^2+\norm{y}_2^2
   }{2}.
\end{equation}

\paragraph{Linear terms.}
For the right child, the induction hypothesis gives
\[
   \widehat M_{J_{\mathrm R}}^{(k)}
   \asvapprox_{C_{t-1}\varepsilon^{(k)},
              S_{t-1}\delta^{(k)}}
   M_{J_{\mathrm R}}.
\]
Applying this approximation with test vectors
\(M_{J_{\mathrm L}}^\top x\) and \(y\) gives
\begin{align}
\label{eq:linear-right}
   \left|
      x^\top M_{J_{\mathrm L}}
      \Delta_{J_{\mathrm R}}^{(k)}y
   \right|
   &\le
   \frac{C_{t-1}\varepsilon^{(k)}}2
   \left(
      \norm{M_{J_{\mathrm L}}^\top x}_{L(M_{J_{\mathrm R}})}^2
      +\norm{y}_{R(M_{J_{\mathrm R}})}^2
   \right)\notag\\
   &\quad+
   \frac{S_{t-1}\delta^{(k)}}2
   \left(
      \norm{M_{J_{\mathrm L}}^\top x}_2^2
      +\norm{y}_2^2
   \right).
\end{align}
Similarly, applying the induction hypothesis for the left child with test
vectors \(x\) and \(M_{J_{\mathrm R}}y\) gives
\begin{align}
\label{eq:linear-left}
   \left|
      x^\top\Delta_{J_{\mathrm L}}^{(k)}
      M_{J_{\mathrm R}}y
   \right|
   &\le
   \frac{C_{t-1}\varepsilon^{(k)}}2
   \left(
      \norm{x}_{L(M_{J_{\mathrm L}})}^2
      +\norm{M_{J_{\mathrm R}}y}_{R(M_{J_{\mathrm L}})}^2
   \right)\notag\\
   &\quad+
   \frac{S_{t-1}\delta^{(k)}}2
   \left(
      \norm{x}_2^2
      +\norm{M_{J_{\mathrm R}}y}_2^2
   \right).
\end{align}
Adding \eqref{eq:linear-right} and \eqref{eq:linear-left}, the
identities \eqref{eq:product-energy-identities} give
\begin{align*}
&\norm{M_{J_{\mathrm L}}^\top x}_{L(M_{J_{\mathrm R}})}^2
+\norm{x}_{L(M_{J_{\mathrm L}})}^2
+\norm{y}_{R(M_{J_{\mathrm R}})}^2
+\norm{M_{J_{\mathrm R}}y}_{R(M_{J_{\mathrm L}})}^2\\
&\qquad=
\norm{x}_{L(M_J)}^2+\norm{y}_{R(M_J)}^2.
\end{align*}
For the Euclidean terms, since \(M_{J_{\mathrm L}}\) and
\(M_{J_{\mathrm R}}\) are contractions,
\[
   \norm{M_{J_{\mathrm L}}^\top x}_2^2\le\norm{x}_2^2,
   \qquad
   \norm{M_{J_{\mathrm R}}y}_2^2\le\norm{y}_2^2.
\]
Therefore the two linear terms contribute at most
\begin{equation}
\label{eq:verified-linear-terms}
   C_{t-1}\varepsilon^{(k)}
   \frac{
      \norm{x}_{L(M_J)}^2+\norm{y}_{R(M_J)}^2
   }{2}
   +
   2S_{t-1}\delta^{(k)}
   \frac{
      \norm{x}_2^2+\norm{y}_2^2
   }{2}.
\end{equation}

\paragraph{Combining the bounds.}
Combining \eqref{eq:quadratic-clean-bound} and
\eqref{eq:verified-linear-terms} gives
\begin{align*}
   \left|x^\top\Delta_J^{(k)}y\right|
   &\le
   C_{t-1}\varepsilon^{(k)}
   \left(1+\frac1{2L}\right)
   \frac{
      \norm{x}_{L(M_J)}^2+\norm{y}_{R(M_J)}^2
   }{2}\\
   &\quad+
   S_{t-1}\delta^{(k)}
   \left(2+\frac1{2L}\right)
   \frac{
      \norm{x}_2^2+\norm{y}_2^2
   }{2}.
\end{align*}
By the definitions of \(C_t\) and \(S_t\),
\[
   C_{t-1}\left(1+\frac1{2L}\right)=C_t,
\]
while
\[
   S_{t-1}\left(2+\frac1{2L}\right)
   \le
   2\left(1+\frac1{2L}\right)S_{t-1}
   =S_t.
\]
Therefore
\[
   \widehat M_J^{(k)}
   \asvapprox_{C_t\varepsilon^{(k)},\,S_t\delta^{(k)}}
   M_J,
\]
which completes the induction and proves the lemma.

\end{proof}

\providecommand{\LevFB}{\Lambda_{\mathrm{FB}}}
\providecommand{\FBF}{\mathsf{F}}
\providecommand{\FBB}{\mathsf{B}}

\section{Converting the PRG of \cite{cohenDoronGoldgraberForwardBackward2026} to a derandomization in SC}
\label{sec:fb-inw-regular-permutation}

We start from the permutation case first.

\begin{theorem}[SC derandomization for permutation ROBPs based on FB--INW]
\label{thm:fb-perm}
Let $n,w\in \mathbb{N}$. Let $\Sigma$ be a finite alphabet. Let $\varepsilon>0$. There exists a deterministic algorithm that, given a length-$n$ width-$w$ permutation ROBP $B$ over alphabet $\Sigma$, computes an approximate transition matrix $\widehat{W}$ such that $\|W_{0..n}-\widehat{W}\|_1\leq \varepsilon$,where $W$ is the transition matrix of $B$. The algorithm runs in time $\poly(n,w,|\Sigma|,1/\varepsilon)$ and uses space


\[
   O\!\left(
      \log n\log\frac{w|\Sigma|}{\varepsilon}
   \right),
\]

\end{theorem}



\subsection{Forward--backward error for permutation ROBPs}
\label{sec:fb-permutation}

We analyze permutation ROBPs using the forward--backward error measure of
\cite{cohenDoronGoldgraberForwardBackward2026}, but here we have to extend it to large alphabets.  

Recall the transition-matrix notation
\(B_t(\sigma),W_t,B_{i-1..j},W_{i-1..j}\) from the preliminaries, where
\(t\)-th transitions map \(V_{t-1}\) to \(V_t\).
Fix an interval \(J=[i,j]\subseteq[n]\).

For a vector \(p\in\mathbb{R}^{V_{i-1}}\), let
\[
   p_{i-1}:=p,
   \qquad
   p_t:=W_t p_{t-1}
   \quad (i\le t\le j)
\]
be its forward propagation through the average transitions on \(J\).  Define
the forward weight of \(J\) with respect to \(p\) by
\[
 \FBF_J(p):=
 2\sum_{t=i}^{j}
 \E_{\sigma\in\Sigma}
 \bigl\|(B_t(\sigma)-W_t)p_{t-1}\bigr\|_1.
\]

Similarly, for a vector \(q\in\mathbb{R}^{V_j}\), let
\[
   q_j:=q,
   \qquad
   q_{t-1}:=W_t^\top q_t
   \quad (i\le t\le j)
\]
be its backward propagation through \(J\).  Define the backward weight by
\[
 \FBB_J(q):=
 2\sum_{t=i}^{j}
 \E_{\sigma\in\Sigma}
 \bigl\|(B_t(\sigma)-W_t)^\top q_t\bigr\|_1.
\]

For every \(J_L=[i,k]\) and \(J_R=[k+1,j]\), 
definitions immediately
give
\begin{align}
 \FBF_J(p)
 &=\FBF_{J_L}(p)+\FBF_{J_R}(W_{i-1..k}p),
 \label{eq:fb-F-decomp}\\
 \FBB_J(q)
 &=\FBB_{J_L}(W_{k..j}^\top q)+\FBB_{J_R}(q).
 \label{eq:fb-B-decomp}
\end{align}

We measure the error of an approximation on \(J\) relative to these two
weights.  For a matrix \(\Delta\), define
\begin{equation}
 \|\Delta\|_{\mathrm{fb},J}
 :=
 \inf\left\{
   \eta\in [0,\infty]:
   |q^\top\Delta p|
   \le
   \eta\,\FBF_J(p)\FBB_J(q)
   \text{ for all }p,q\in\mathbb R^w
 \right\}.
 \label{eq:fb-fb-seminorm}
\end{equation}

For \(p\in\mathbb R^w\), let
\[
   \Psi(p):=\sum_{1\le a<b\le w}|p_a-p_b|.
\]

The first property we need is that the two weights have a total budget that
does not depend on the length of \(J\).

\begin{lemma}[Forward and backward weight bounds]
\label{lem:fb-perm-basic}
\label{lem:fb-weight-bounds}
For every interval \(J=[i,j]\) and every \(p,q\in\mathbb R^w\),
\begin{align}
 \FBF_J(p)
 &\le
 4\bigl(\Psi(p)-\Psi(W_{i-1..j}p)\bigr)
 \le4\Psi(p),
 \label{eq:fb-weight-F}\\
 \FBB_J(q)
 &\le
 4\bigl(\Psi(q)-\Psi(W_{i-1..j}^\top q)\bigr)
 \le4\Psi(q).
 \label{eq:fb-weight-B}
\end{align}
In particular,
\begin{equation}
 \FBF_J(v),\FBB_J(v)
 \le4\Psi(v)
 \le4w\|v\|_1.
 \label{eq:fb-weight-budget}
\end{equation}
\end{lemma}

We give the details in \Cref{app:fb-weight-bounds-proof}.

The second property concerns individual paths. To control a fixed transition \(B_t(\sigma)\) using the forward and backward
weights, we compare its deviation from \(W_t\) with the averaged 
deviation over the alphabet.  Define the following quantities
\begin{align*}
 \Lambda_F(t)
 &:=
 \inf\Bigl\{\lambda\ge 0:
 \|(B_t(\sigma)-W_t)p\|_1
 \le
 \lambda\,
 \E_{\tau\in\Sigma}\|(B_t(\tau)-W_t)p\|_1
 \text{ for all }\sigma\in\Sigma,\ p\in\mathbb R^w
 \Bigr\},\\
 \Lambda_B(t)
 &:=
 \inf\Bigl\{\lambda\ge 0:
 \|(B_t(\sigma)-W_t)^\top q\|_1
 \le
 \lambda\,
 \E_{\tau\in\Sigma}\|(B_t(\tau)-W_t)^\top q\|_1
 \text{ for all }\sigma\in\Sigma,\ q\in\mathbb R^w
 \Bigr\}.
\end{align*}
Set
\[
 \Lambda_F:=\max\{1,\max_t\Lambda_F(t)\},
 \qquad
 \Lambda_B:=\max\{1,\max_t\Lambda_B(t)\},
 \qquad \Lambda_{FB}=\Lambda_F\Lambda_B
\]
Since the average over \(\Sigma\) is uniform, every individual term is at most
\(|\Sigma|\) times the average, and hence
\[
 \Lambda_F,\Lambda_B\le |\Sigma|.
\]

\begin{lemma}[Bounds for one-step transitions with fixed symbols]
\label{lem:fb-alphabet-leverage}
For every \(t\in[n]\),
\[
   \Lambda_F(t),\Lambda_B(t)\le|\Sigma|.
\]
Consequently,
\(\Lambda_F,\Lambda_B\le|\Sigma|\); for a binary alphabet,
\(\Lambda_F=\Lambda_B=1\).
\end{lemma}

\begin{proof}
For fixed \(t,p\), put
\(g_\sigma=\|(B_t(\sigma)-W_t)p\|_1\).  Whenever
\(\E_{\tau\in\Sigma}g_\tau>0\),
\[
   g_\sigma
   \le\sum_{\tau\in\Sigma}g_\tau
   =|\Sigma|\,\E_{\tau\in\Sigma}g_\tau.
\]
If the average is zero, then every \(g_\sigma\) is zero.  The backward bound
is identical.  When \(|\Sigma|=2\), the two matrices
\(B_t(\sigma)-W_t\) are negatives of one another, so the two deviations are
equal to their average.
\end{proof}

For an interval \(J=[i,j]\) and a string
\(\sigma_J=(\sigma_i,\ldots,\sigma_j)\in\Sigma^{j-i+1}\), define the
corresponding path matrix by
\[
   B_J(\sigma_J)
   :=B_{i-1..j}(\sigma_i,\ldots,\sigma_j).
\]
Thus \(B_J(\sigma_J)\) is the transition matrix obtained by fixing the symbol
\(\sigma_t\) at every layer \(t\in J\).

\begin{lemma}[Bounds for transitions with fixed symbols]
\label{lem:fb-pathwise-bounds}
Let \(J=[i,j]\) and
\(\sigma_J\in\Sigma^{j-i+1}\).  Then for every \(p,q\in\mathbb R^w\),
\begin{align}
 \|B_J(\sigma_J)p-W_{i-1..j}p\|_1
 &\le\frac{\Lambda_F}{2}\FBF_J(p),
 \label{eq:fb-pathwise-F}\\
 \|B_J(\sigma_J)^\top q-W_{i-1..j}^\top q\|_1
 &\le\frac{\Lambda_B}{2}\FBB_J(q).
 \label{eq:fb-pathwise-B}
\end{align}
\end{lemma}

See \Cref{app:fb-pathwise-bounds-proof} for the proof.

\subsection{Forward--backward analysis for one-step INW and samplers}
\label{subsec:fb-stability}

We next show that sampler fixing and INW merging both preserve the
forward--backward guarantees needed in the recursive construction.

\subsubsection{Sampler: from entrywise error to forward--backward error}

For a matrix \(\Delta\), write
\[
   \|\Delta\|_{\max}:=\max_{u,v\in[w]}|\Delta_{u,v}|.
\]

\begin{lemma}[Forward--backward bound from entrywise error]
\label{lem:fb-path-stability}
Fix an interval \(J=[i,j]\), and let \(\mathcal Q_J\) be a nonempty family of
path matrices \(\{B_J(\sigma_J)\}_{\sigma_J\in\Sigma^{j-i+1}}\).
Let \(M\) and \(\widehat M\) be convex combinations of matrices in
\(\mathcal Q_J\), and put \(\Delta:=\widehat M-M\).  If
\[
   \|\Delta\|_{\max}\le\alpha,
\]
then
\begin{equation}
   \|\Delta\|_{\mathrm{fb},J}
   \le \LevFB w^4\alpha.                                       \label{eq:fb-path-stability}
\end{equation}
\end{lemma}

\begin{proof}

Consider  subspaces of $ \mathbb R^w$, for which, one can always attain the same output when using an arbitrary matrix in \(\mathcal Q_J\) to operate on them. Denote
\[
   L_J
   :=
   \{r\in\mathbb R^w:
      Ur=Vr \text{ for all } U,V\in\mathcal Q_J\},
\]
and
\[
   R_J
   :=
   \{s\in\mathbb R^w:
      U^\top s=V^\top s \text{ for all } U,V\in\mathcal Q_J\}.
\]
We first construct projections onto these subspaces.

For \(U,V\in\mathcal Q_J\), let \(\pi_{U,V}\) be the permutation of
\([w]\) represented by \(U^\top V\).  Let \(G_L\) be the undirected graph on
\([w]\) containing the edge
\[
   \{z,\pi_{U,V}^{-1}(z)\}
\]
for every \(U,V\in\mathcal Q_J\) and \(z\in[w]\).  Define \(G_R\)
analogously using the permutations represented by \(UV^\top\).
For \(u\in[w]\), let \(C_L(u)\) and \(C_R(u)\) denote the connected
components containing \(u\) in \(G_L\) and \(G_R\), respectively. 

Fix arbitrary \(p,q\in\mathbb R^w\), define
\[
   (\Pi_L p)_u
   :=
   \frac{1}{|C_L(u)|}\sum_{v\in C_L(u)}p_v,
   \qquad
   (\Pi_R q)_u
   :=
   \frac{1}{|C_R(u)|}\sum_{v\in C_R(u)}q_v.
\]
Write
\[
   \bar p:=\Pi_Lp,
   \qquad
   \bar q:=\Pi_Rq.
\]

Notice that on every connected component of \(G_L\), \(\bar p\) is constant across all the entries in this component. So it is
fixed by every permutation matrix \(U^\top V\) with
\(U,V\in\mathcal Q_J\).  Hence
\[
   U\bar p=V\bar p
   \qquad
   \text{for all }U,V\in\mathcal Q_J,
\]
so \(\bar p\in L_J\).  Similarly, \(\bar q\in R_J\).

Fix any \(U_\star\in\mathcal Q_J\).  Since \(M\) and \(\widehat M\) are
convex combinations of matrices in \(\mathcal Q_J\), we have
\[
   M\bar p=U_\star\bar p=\widehat M\bar p,
   \qquad
   M^\top\bar q=U_\star^\top\bar q=\widehat M^\top\bar q.
\]
Thus
\begin{equation}
   \Delta\bar p=0,
   \qquad
   \bar q^\top\Delta=0, 
   \qquad
   p\Delta q=(p-\bar p)\Delta (q-\bar q)
   \label{eq:fb-path-kernel}
\end{equation}

It remains to bound $\|p-\bar p\|_1,\|q-\bar q\|_1$. Consider first \(p\).  For any \(U,V\in\mathcal Q_J\), since \(U\)
is a permutation matrix,
\[
   \|(I-U^\top V)p\|_1
   =
   \|Up-Vp\|_1.
\]
By the pathwise forward bound,
\begin{align*}
   \|Up-Vp\|_1
   &\le
   \|Up-W_{i-1..j}p\|_1
   +
   \|Vp-W_{i-1..j}p\|_1 \\
   &\le
   \Lambda_F\FBF_J(p).
\end{align*}
Set
\[
   \delta_F:=\Lambda_F\FBF_J(p).
\]
If \(\{z,\pi_{U,V}^{-1}(z)\}\) is an edge of \(G_L\) arising from
\(U,V\), then
\[
   |p_z-p_{\pi_{U,V}^{-1}(z)}|
   \le
   \|(I-U^\top V)p\|_1
   \le
   \delta_F.
\]

Now fix a connected component \(C\) of \(G_L\). 
Any two vertices \(u,v\in C\) can be joined by a simple path of at most
\(|C|-1\) edges.  Hence
\[
   |p_u-p_v|
   \le (|C|-1)\delta_F.
\]
Since \(\bar p\) is the average of \(p\) on \(C\),
\[
   |p_u-\bar p_u|
   \le
   \frac{1}{|C|}\sum_{v\in C}|p_u-p_v|
   \le
   (|C|-1)\delta_F.
\]
Summing over all vertices and all connected components gives
\[
   \|p-\bar p\|_1
   \le
   \sum_C |C|(|C|-1)\delta_F
   \le
   w^2\delta_F.
\]
Therefore
\begin{equation}
   \|p-\bar p\|_1
   \le
   w^2\Lambda_F\FBF_J(p).
   \label{eq:fb-orbit-F}
\end{equation}

Repeating the same argument for \(G_R\) yields
\begin{equation}
   \|q-\bar q\|_1
   \le
   w^2\Lambda_B\FBB_J(q).
   \label{eq:fb-orbit-B}
\end{equation}

Finally, \eqref{eq:fb-path-kernel} tells
\[
   q^\top\Delta p
   =
   (q-\bar q)^\top\Delta(p-\bar p).
\]
Hence
\begin{align*}
   |q^\top\Delta p|
   &\le
   \|\Delta\|_{\max}
   \|p-\bar p\|_1
   \|q-\bar q\|_1 \\
   &\le
   \Lambda_F\Lambda_B\,w^4\alpha\,
   \FBF_J(p)\FBB_J(q) \\
   &=
   \LevFB w^4\alpha\,
   \FBF_J(p)\FBB_J(q),
\end{align*}
which is the desired bound.
\end{proof}
An offline seed for the sampler in \Cref{lem:good sampler} can therefore be
found by testing the $w^2$ \(0\)-\(1\) entry functions for every current
interval.  To limit the forward--backward error from one fixing step to
\(\eta_{\mathrm{fix}}\), set the entrywise sampler accuracy to
\begin{equation}
   \alpha_{\mathrm{fix}}:=
   \frac{\eta_{\mathrm{fix}}}{\LevFB w^4}.
   \label{eq:fb-perm-sampler-alpha}
\end{equation}

\subsubsection{INW merging with forward-backward norm}

The following lemma is an extended version of the expander-product estimate of
\cite[Proposition~3.4]{cohenDoronGoldgraberForwardBackward2026}.  

\begin{lemma}[The FB error for one-step INW]
\label{lem:fb-robust-merge}
Let \(J=[i,j]\), split at \(k\) into
\[
   J_L=[i,k],
   \qquad
   J_R=[k+1,j],
\]
and write
\[
   W_L:=W_{i-1..k},
   \qquad
   W_R:=W_{k..j}.
\]
Let \(\mathcal S\) be a finite seed set, and for each
\(\xi\in\mathcal S\) let \(U_\xi^L\) and \(U_\xi^R\) be path matrices on
\(J_L\) and \(J_R\), respectively.  Define
\[
   \widetilde W_L:=\E_{\xi\in\mathcal S}U_\xi^L,
   \qquad
   \widetilde W_R:=\E_{\xi\in\mathcal S}U_\xi^R.
\]

Suppose
\[
   \|\widetilde W_L-W_L\|_{\mathrm{fb},J_L}\le\eta_L,
   \qquad
   \|\widetilde W_R-W_R\|_{\mathrm{fb},J_R}\le\eta_R.
\]
Let \(\mathcal H\) be a regular reversible expander on \(\mathcal S\) with
normalized second singular value at most \(\lambda_{\exp}\), and define
\[
   \widetilde W_J
   :=
   \E_{(\xi,\zeta)\sim\mathcal H}
   U_\zeta^R U_\xi^L,
\]
where the expectation is over a uniformly random directed edge of
\(\mathcal H\).  Then \(\widetilde W_J\) is an average of path matrices on
\(J\), and
\begin{equation}
   \|\widetilde W_J-W_{i-1..j}\|_{\mathrm{fb},J}
   \le
   \max\left\{
      \eta_L,\eta_R,
      \lambda_{\exp}\LevFB+16w^3\eta_L\eta_R
   \right\}.
   \label{eq:fb-robust-merge}
\end{equation}
\end{lemma}

\begin{proof}
Fix \(p,q\in\mathbb R^w\), and abbreviate the forward and backward weights on
the two children by
\[
   F_L:=\FBF_{J_L}(p),
   \qquad
   F_R:=\FBF_{J_R}(W_Lp),
\]
and
\[
   B_L:=\FBB_{J_L}(W_R^\top q),
   \qquad
   B_R:=\FBB_{J_R}(q).
\]
By \eqref{eq:fb-F-decomp}--\eqref{eq:fb-B-decomp},
\begin{equation}
   \FBF_J(p)=F_L+F_R,
   \qquad
   \FBB_J(q)=B_L+B_R.
   \label{eq:fb-merge-weight-decomp}
\end{equation}

We separate the error of the merge into the expander-mixing error and the
errors already present in the two children.  Put
\[
   C_J
   :=
   \widetilde W_J-\widetilde W_R\widetilde W_L,
   \qquad
   E_L:=\widetilde W_L-W_L,
   \qquad
   E_R:=\widetilde W_R-W_R.
\]
Since \(W_{i-1..j}=W_RW_L\), we have the exact decomposition
\begin{equation}
   \widetilde W_J-W_{i-1..j}
   =
   C_J+W_RE_L+E_RW_L+E_RE_L.
   \label{eq:fb-merge-error-decomp}
\end{equation}

\paragraph{Bound \(C_J\).}
We first bound the term \(C_J\).  Define
\[
   f(\xi):=U_\xi^Lp-\widetilde W_Lp,
   \qquad
   g(\zeta):=(U_\zeta^R)^\top q-\widetilde W_R^\top q.
\]
Since
\[
   \widetilde W_L=\E_{\xi}U_\xi^L,
   \qquad
   \widetilde W_R=\E_{\zeta}U_\zeta^R,
\]
we have 
\begin{align*}
   \E_{(\xi,\zeta)\sim\mathcal H}
      \langle g(\zeta),f(\xi)\rangle
   &=
   q^\top
   \left(
      \E_{(\xi,\zeta)\sim\mathcal H}U_\zeta^R U_\xi^L
      -\widetilde W_R\widetilde W_L
   \right)p \\
   &=q^\top C_Jp.
\end{align*}
The vector-valued expander-mixing inequality therefore gives
\begin{equation}
   |q^\top C_Jp|
   \le
   \lambda_{\exp}
   \left(\E_{\xi}\|f(\xi)\|_2^2\right)^{1/2}
   \left(\E_{\zeta}\|g(\zeta)\|_2^2\right)^{1/2}.
   \label{eq:fb-merge-mixing-intermediate}
\end{equation}

To bound the first factor in
\eqref{eq:fb-merge-mixing-intermediate}, for every
\(\xi\in\mathcal S\),
\begin{align*}
   \|f(\xi)\|_2
   &\le \|f(\xi)\|_1
   && \text{since \(\|x\|_2\le\|x\|_1\)}\\
   &\le
   \|(U_\xi^L-W_L)p\|_1
   +
   \|(\widetilde W_L-W_L)p\|_1
   && \text{by the triangle inequality}\\
   &\le
   \frac{\Lambda_F}{2}F_L
   +
   \E_{\xi'}
   \|(U_{\xi'}^L-W_L)p\|_1
   && \text{by \eqref{eq:fb-pathwise-F} and
      \(\widetilde W_L=\E_{\xi'}U_{\xi'}^L\)}\\
   &\le
   \Lambda_F F_L
   && \text{by \eqref{eq:fb-pathwise-F}.}
\end{align*}
where we used the pathwise bound
\eqref{eq:fb-pathwise-F} and

Hence
\[
   \left(\E_\xi\|f(\xi)\|_2^2\right)^{1/2}
   \le \Lambda_FF_L.\qquad \text{Similarly,} \qquad 
   \left(\E_\zeta\|g(\zeta)\|_2^2\right)^{1/2}
   \le \Lambda_BB_R.
\]
Substituting into \eqref{eq:fb-merge-mixing-intermediate} gives
\begin{equation}
   |q^\top C_Jp|
   \le
   \lambda_{\exp}\LevFB\,F_LB_R.
   \label{eq:fb-merge-expander}
\end{equation}

\paragraph{Bound \(W_RE_L,E_RW_L\).}
The two linear terms are controlled directly by the forward--backward
guarantees for the corresponding children.  For the left-child, 
\[
   q^\top W_RE_Lp
   =
   (W_R^\top q)^\top E_Lp,
\]
and the assumption
\(\|E_L\|_{\mathrm{fb},J_L}\le\eta_L\) gives
\begin{equation}
   |q^\top W_RE_Lp|
   \le
   \eta_L\,
   \FBF_{J_L}(p)\,
   \FBB_{J_L}(W_R^\top q)
   =
   \eta_L F_LB_L.
   \label{eq:fb-merge-linear-L}
\end{equation}
Similarly, for the right-child,
\[
   q^\top E_RW_Lp
   =
   q^\top E_R(W_Lp).
\]
Hence
\begin{equation}
   |q^\top E_RW_Lp|
   \le
   \eta_R\,
   \FBF_{J_R}(W_Lp)\,
   \FBB_{J_R}(q)
   =
   \eta_R F_RB_R.
   \label{eq:fb-merge-linear-R}
\end{equation}
\paragraph{Bound \(E_RE_L\).}
For the quadratic term, we first derive the Euclidean bounds.
\begin{align*}
   \|E_Lp\|_2
   &=
   \sup_{\|r\|_2=1}|r^\top E_Lp|
   && \text{by Euclidean duality} \\
   &\le
   \eta_L F_L
   \sup_{\|r\|_2=1}\FBB_{J_L}(r)
   && \text{since
      \(\|E_L\|_{\mathrm{fb},J_L}\le\eta_L\)} \\
   &\le
   4w
   \eta_L F_L
   \sup_{\|r\|_2=1}\|r\|_1
   && \text{by \eqref{eq:fb-weight-budget}} \\
   &\le
   4w^{3/2}\eta_LF_L
   && \text{since \(\|r\|_1\le\sqrt w\,\|r\|_2\).}
\end{align*}
Similarly,
\begin{align*}
   \|E_R^\top q\|_2
   &\le
   4w^{3/2}\eta_RB_R.
\end{align*}
Therefore, by Cauchy--Schwarz,
\begin{align}
   |q^\top E_RE_Lp|
   &=
   |\langle E_R^\top q,E_Lp\rangle| \notag\\
   &\le
   \|E_R^\top q\|_2\|E_Lp\|_2 \notag\\
   &\le
   16w^3\eta_L\eta_R\,F_LB_R.
   \label{eq:fb-merge-quadratic}
\end{align}

\paragraph{Combining the bounds}
Combining
\eqref{eq:fb-merge-expander}--\eqref{eq:fb-merge-quadratic}, we obtain
\[
\begin{aligned}
   |q^\top(\widetilde W_J-W_{i-1..j})p|
   \le{}&
   \bigl(\lambda_{\exp}\LevFB+16w^3\eta_L\eta_R\bigr)F_LB_R \\
   &\quad+\eta_LF_LB_L+\eta_RF_RB_R.
\end{aligned}
\]
Let
\[
   K
   :=
   \max\left\{
      \eta_L,\eta_R,
      \lambda_{\exp}\LevFB+16w^3\eta_L\eta_R
   \right\}.
\]
Since all four weights are nonnegative,
\[
   F_LB_R+F_LB_L+F_RB_R
   \le
   (F_L+F_R)(B_L+B_R).
\]
Together with \eqref{eq:fb-merge-weight-decomp}, this gives
\[
   |q^\top(\widetilde W_J-W_{i-1..j})p|
   \le
   K\,\FBF_J(p)\FBB_J(q).
\]
Taking the supremum over \(p,q\) in the definition of the
forward--backward seminorm proves \eqref{eq:fb-robust-merge}.
\end{proof}

A useful consequence is that the error does not accumulate across consecutive
INW merges.  Suppose every matrix entering such a sequence has
forward--backward error at most \(\eta\).  If
\begin{equation}
   \lambda_{\exp}\LevFB+16w^3\eta^2\le\eta,
   \label{eq:fb-fixed-point}
\end{equation}
then an inductive deploy of \Cref{lem:fb-robust-merge} shows that every matrix produced
by the merges also has error at most \(\eta\).

\begin{corollary}[Multilevel INW merge]
\label{cor:fb-multilevel-inw}
Let \(J_1,\ldots,J_m\) be the consecutive intervals of a partition of
an interval \(J\), where \(m=2^h\).
Let \(\mathcal S=\{0,1\}^d\).  For every \(a\in[m]\) and
\(y\in\mathcal S\), let \(U_a(y)\) be a path matrix on \(J_a\), and write
\[
   \widetilde W_a:=\E_{y\in\mathcal S}U_a(y).
\]
Suppose that
\[
   \|\widetilde W_a-W_{J_a}\|_{\mathrm{fb},J_a}\le\eta
   \qquad\text{for every }a\in[m],
\]
and that
\begin{equation}
   \lambda_{\exp}\LevFB+16w^3\eta^2\le\eta.
   \label{eq:fb-multilevel-fixed-point}
\end{equation}
Let
\[
   \INW^{(h)}:\widehat{\mathcal S}\longrightarrow\mathcal S^m,
   \qquad
   \widehat{\mathcal S}=\{0,1\}^{\widehat d},
\]
be the standard \(h\)-level INW generator obtained from regular reversible
expanders whose normalized second singular value is at most
\(\lambda_{\exp}\).  For \(z\in\widehat{\mathcal S}\), write
\[
   \INW^{(h)}(z)=(y_1,\ldots,y_m)
\]
and define
\[
   \widetilde W_J
   :=
   \E_{z\in\widehat{\mathcal S}}
   U_m(y_m)\cdots U_1(y_1).
\]
Then \(\widetilde W_J\) is an average of path matrices on \(J\) and
\begin{equation}
   \|\widetilde W_J-W_J\|_{\mathrm{fb},J}\le\eta.
   \label{eq:fb-multilevel-inw-error}
\end{equation}
Moreover, using a standard strongly explicit expander family,
\begin{equation}
   \widehat d
   =d+O\!\left(h\log(1/\lambda_{\exp})\right).
   \label{eq:fb-multilevel-inw-seed}
\end{equation}
\end{corollary}

\begin{proof}
Apply \Cref{lem:fb-robust-merge} successively along the \(h\) levels of the
INW recursion.  At every internal node, the two child means have
forward--backward error at most \(\eta\); by
\eqref{eq:fb-multilevel-fixed-point}, their merge again has error at most
\(\eta\).  Induction over the levels gives
\eqref{eq:fb-multilevel-inw-error}.  The seed-length bound is the standard INW
recurrence: each level adds
\(O(\log(1/\lambda_{\exp}))\) bits.
\end{proof}

\subsection{The derandomization: combining samplers and INW}
\label{subsec:fb-sampler-fixed-macroblocks}


\begin{proof}[Proof of \Cref{thm:fb-perm}]
For \(n=1\), compute the average transition directly; hence assume \(n\ge2\).
Pad the program with identity transitions to length
\[
   N:=2^L,
   \qquad
   L:=\lceil\log n\rceil,
\]
where \(n\le N<2n\), and continue to write \(B\) for the padded program.

\paragraph{The recursion.}
We use the same recursive structure as in \Cref{sec:INW_in_SC}.
At each stage, we first use a sampler to reduce the current alphabet and then
apply the multilevel INW generator of \Cref{cor:fb-multilevel-inw} to reduce the
current length.

Write
\begin{equation}
   \beta:=1+\log\frac{w\LevFB}{\varepsilon},
   \qquad
   H:=1+\log\frac{nw\LevFB}{\varepsilon},
   \label{eq:fb-beta-H}
\end{equation}
and let \(m\) be the power of two satisfying
\begin{equation}
   \log m
   :=
   \min\left\{
      L,
      \left\lfloor\frac{H}{\beta}\right\rfloor
   \right\}.
   \label{eq:fb-m}
\end{equation}
Set the number of stages to be
\[
   K:=\left\lceil\frac{L}{\log m}\right\rceil.
\]
For \(0\le t\le K\), set
\[
   \ell_t:=\min\{t\log m,L\},
   \qquad
   n_t:=\frac{N}{2^{\ell_t}},
\]
and for \(0\le t<K\), set
\[
   m_t:=2^{\ell_{t+1}-\ell_t}.
\]
At stage \(t\), we replace every \(m_t\) consecutive current transitions
by one transition.  Thus \(m_t=m\) except possibly in the last stage,
and \(n_K=1\).

Let \(\Sigma_0=\Sigma\), \(B^0=B\), and let
\[
   \mathcal G^0:\Sigma^N\to\Sigma^N
\]
be the identity generator.  For the $t$-th stage, we recursively define length-\(n_t\)
permutation ROBPs \(B^t\) over alphabets \(\Sigma_t\) and generators
\[
   \mathcal G^t:\Sigma_t^{n_t}\longrightarrow\Sigma^N
\]
such that
\begin{equation}
   B^t=B\circ\mathcal G^t,
   \label{eq:fb-composition-invariant}
\end{equation}
as follows:

\paragraph{Compressing alphabet with a sampler.}
Suppose \(B^t,\Sigma_t,\mathcal G^t\) have been defined.  Let
\[
   \Samp_t:
   \{0,1\}^{r_t}\times\{0,1\}^{p_t}
   \longrightarrow\Sigma_t
\]
be an
\((\alpha_{\mathrm{samp}},\gamma_{\mathrm{samp}})\)-averaging sampler.
Writing
\[
   W_k^t:=\E_{\sigma\in\Sigma_t}B_k^t(\sigma),
\]
choose an offline seed \(x_t\in\{0,1\}^{r_t}\), shared by all transitions,
such that
\begin{equation}
   \left|
      [W_k^t]_{u,v}
      -
      \E_y[B_k^t(\Samp_t(x_t,y))]_{u,v}
   \right|
   \le\alpha_{\mathrm{samp}}
   \label{eq:fb-stage-fixing}
\end{equation}
for every \(k\in[n_t]\) and \(u,v\in[w]\).  Define
\[
   \widetilde B_k^t(y)
   :=
   B_k^t(\Samp_t(x_t,y))
\]
and
\[
   \widetilde{\mathcal G}^t(y_1,\ldots,y_{n_t})
   :=
   \mathcal G^t\bigl(
      \Samp_t(x_t,y_1),\ldots,\Samp_t(x_t,y_{n_t})
   \bigr).
\]
Then
\[
   \widetilde B^t=B\circ\widetilde{\mathcal G}^t.
\]

\paragraph{Reducing length with INW generator.}
Let
\[
   \INW_t:
   \{0,1\}^{d_t}
   \longrightarrow
   \bigl(\{0,1\}^{p_t}\bigr)^{m_t}
\]
be the \((\log m_t)\)-level INW generator from
\Cref{cor:fb-multilevel-inw}, using expanders with normalized second
singular value at most \(\lambda_{\exp}\).  Thus
\begin{equation}
   d_t
   =
   p_t+
   O\!\left(
      \log m_t\cdot\log(1/\lambda_{\exp})
   \right).
   \label{eq:fb-stage-inw-seed}
\end{equation}
Since \(n_{t+1}=n_t/m_t\), set
\[
   \Sigma_{t+1}:=\{0,1\}^{d_t},
\]
and define
\begin{equation}
   B_k^{t+1}(y)
   :=
   \widetilde B^t_{(k-1)m_t\,..\,km_t}
      \bigl(\INW_t(y)\bigr).
   \label{eq:fb-next-transition}
\end{equation}
Correspondingly,
\[
   \mathcal G^{t+1}(y_1,\ldots,y_{n_{t+1}})
   :=
   \widetilde{\mathcal G}^t
   \bigl(\INW_t(y_1),\ldots,\INW_t(y_{n_{t+1}})\bigr),
\]
where the \(m_t\)-tuples are concatenated in order.  Hence
\(B^{t+1}=B\circ\mathcal G^{t+1}\).

\paragraph{Final generator.}
Repeating this procedure for \(K\) stages leaves a single transition.  Let
\[
   x:=(x_0,\ldots,x_{K-1})
\]
be the fixed offline sampler seeds and let
\[
   \mathcal G^K:\Sigma_K\to\Sigma^N.
\]
Restricting its output to the first \(n\) symbols gives the desired
generator
\begin{equation}
   \mathcal G_x(y):=\mathcal G_x^K(y)|_{[n]}.
   \label{eq:fb-final-generator}
\end{equation}

\paragraph{Error analysis.}
Set
\begin{equation}
   \eta_\star:=\frac{\varepsilon}{64w^3},
   \qquad
   \lambda_{\exp}:=\frac{\eta_\star}{16\LevFB},
   \label{eq:fb-eta-star}
\end{equation}
and
\begin{equation}
   \eta_{\mathrm{samp}}
      :=\frac{\eta_\star}{8K},
   \qquad
   \alpha_{\mathrm{samp}}
      :=\frac{\eta_{\mathrm{samp}}}{\LevFB w^4},
   \qquad
   \gamma_{\mathrm{samp}}
      :=\frac1{2Nw^2}.
   \label{eq:fb-fixing-parameters}
\end{equation}
For \(0\le t\le K\), put
\[
   z_t:=\frac{\eta_\star}{8}+t\eta_{\mathrm{samp}},
\]
so that
\begin{equation}
   \frac{\eta_\star}{8}
   \le z_t
   \le\frac{\eta_\star}{4}.
   \label{eq:fb-error-envelope-range}
\end{equation}

The \(k\)-th transition at stage \(t\) represents the interval
\[
   J_{t,k}
   :=
   \bigl[(k-1)2^{\ell_t}+1,\,k2^{\ell_t}\bigr]
\]
of the padded program.  We claim inductively that
\begin{equation}
   \|W_k^t-W_{J_{t,k}}\|_{\mathrm{fb},J_{t,k}}
   \le z_t
   \qquad(k\in[n_t]).
   \label{eq:fb-stage-error-invariant}
\end{equation}

\begin{proof}[Proof of the claim]
For \(t=0\) the two matrices are identical.
At stage \(t\), there are at most \(Nw^2\) entrywise tests in
\eqref{eq:fb-stage-fixing}.  The choice of
\(\gamma_{\mathrm{samp}}\) therefore guarantees, by a union bound, that a
suitable \(x_t\) exists.  If
\[
   \widetilde W_k^t
   :=
   \E_y\widetilde B_k^t(y),
\]
then \(W_k^t\) and \(\widetilde W_k^t\) are convex combinations of the same
path matrices.  Hence \Cref{lem:fb-path-stability} gives
\begin{equation}
   \|\widetilde W_k^t-W_k^t\|_{\mathrm{fb},J_{t,k}}
   \le\eta_{\mathrm{samp}}.
   \label{eq:fb-sampler-fb-error}
\end{equation}
Thus, assuming \eqref{eq:fb-stage-error-invariant} at stage \(t\) and using triangular inequality.
\begin{equation}
   \|\widetilde W_k^t-W_{J_{t,k}}\|_{\mathrm{fb},J_{t,k}}
   \le z_t+\eta_{samp}= z_{t+1}.
   \label{eq:fb-after-sampler-error}
\end{equation}

For every
\(z\in[\eta_\star/8,\eta_\star/4]\), the choices above satisfy
\begin{equation}
   \lambda_{\exp}\LevFB+16w^3z^2\le z.
   \label{eq:fb-parameter-check}
\end{equation}
Indeed,
\[
   \lambda_{\exp}\LevFB=\frac{\eta_\star}{16},
   \qquad
   16w^3z^2
   \le w^3\eta_\star^2
   \le\frac{\eta_\star}{128}.
\]
Each transition of \(B^{t+1}\) is obtained by applying the multilevel INW
merge to \(m_t\) consecutive transitions of \(\widetilde B^t\).  Hence
\eqref{eq:fb-after-sampler-error},
\eqref{eq:fb-parameter-check}, and
\Cref{cor:fb-multilevel-inw} imply
\[
   \|W_k^{t+1}-W_{J_{t+1,k}}\|_{\mathrm{fb},J_{t+1,k}}
   \le z_{t+1}.
\]
This proves \eqref{eq:fb-stage-error-invariant}.
\end{proof}

At the final stage,
\begin{equation}
   \left\|
      M_{0..N}[\mathcal G_x^+]-W_{0..N}
   \right\|_{\mathrm{fb},[1,N]}
   \le\frac{\eta_\star}{4}.
   \label{eq:fb-final-error-cap}
\end{equation}
Identity padding gives
\[
   M_{0..N}[\mathcal G_x^+]
   =
   M_{0..n}[\mathcal G_x],
   \qquad
   W_{0..N}=W_{0..n}.
\]
Set
\[
   \widehat W_{0..n}:=M_{0..n}[\mathcal G_x].
\]
For every start state \(v\) and accepting set \(S\),
\Cref{lem:fb-weight-bounds} gives
\[
   \FBF_{[1,N]}(e_v)\le4w,
   \qquad
   \FBB_{[1,N]}(\one_S)\le w^2.
\]
Consequently,
\[
   \left|
      \one_S^\top
      (\widehat W_{0..n}-W_{0..n})
      e_v
   \right|
   \le
   w^3\eta_\star
   <\varepsilon.
\]
This proves the accepting-set guarantee, and the entrywise guarantee follows
by taking \(S\) to be a singleton.

\paragraph{Seed and advice lengths.}
By \eqref{eq:fb-eta-star},
\[
   \log(1/\lambda_{\exp})=O(\beta),
\]
while \eqref{eq:fb-fixing-parameters} gives
\[
   \log(1/\alpha_{\mathrm{samp}}),
   \log(1/\gamma_{\mathrm{samp}})
   =O(H).
\]
Using averaging samplers from \Cref{lem:good sampler} gives
\[
   r_t=\log|\Sigma_t|+O(H),
   \qquad
   p_t=O(H).
\]
Since \(\beta\log m_t\le\beta\log m\le H\), \eqref{eq:fb-stage-inw-seed} yields
\[
   d_t
   =
   p_t+O(\beta\log m_t)
   =
   O(H).
\]
Thus
\[
   \log|\Sigma_t|=d_{t-1}=O(H)
   \qquad(t\ge1),
\]
and therefore
\[
   r_0=\log|\Sigma|+O(H),
   \qquad
   r_t=O(H)\quad(t\ge1).
\]
Moreover,
\[
   K
   =
   O\!\left(
      1+\frac{\beta\log n}{H}
   \right).
\]
It follows that
\begin{equation}
   |x|
   =
   \log|\Sigma|+O(KH)
   =
   \log|\Sigma|+O(\beta\log n+H),
   \label{eq:fb-advice}
\end{equation}
while the final online seed has length \(O(H)\).

Finally, \Cref{lem:fb-alphabet-leverage} allows
\(\LevFB=|\Sigma|^2\), so
\[
   \beta
   =
   1+\log\frac{w|\Sigma|^2}{\varepsilon},
   \qquad
   H
   =
   1+\log\frac{nw|\Sigma|^2}{\varepsilon}.
\]
So the total space for storing $|x|$ and the online seed is bounded by 
\[
O\left(\log n\log\frac{w|\Sigma|}{\varepsilon}\right)
\]
as claimed.

\paragraph{Time complexity.}
The advice strings \(x_0,\ldots,x_{K-1}\) are found sequentially by exhaustive
search, exactly as in \Cref{sec:INW_in_SC}.  At every stage the current
alphabet has polynomial size, so the search runs in polynomial time.

To compute the \(q\)-th output symbol of \(\mathcal G_x(y)\), define
\[
   q_t:=\left\lceil\frac{q}{2^{\ell_t}}\right\rceil,
   \qquad
   j_t:=q_t-(q_{t+1}-1)m_t\in[m_t].
\]
Starting from the final online seed \(s_K:=y\), trace the unique
path:
\begin{equation}
   \widehat s_t
   :=
   \bigl(\INW_t(s_{t+1})\bigr)_{j_t},
   \qquad
   s_t:=\Samp_t(x_t,\widehat s_t),
   \qquad
   t=K-1,\ldots,0.
   \label{eq:fb-strong-explicit-recursion}
\end{equation}
Then \(s_0=\mathcal G_x(y)_q\).  Only one coordinate of each INW generator is
evaluated, and the total number of expander levels over all stages is
\(\sum_{t=0}^{K-1}\log m_t=L=O(\log n)\).  Hence one output symbol can be computed in
\[
   O(H+\log n+\log|\Sigma|)
\]
workspace, plus the workspace needed to evaluate one transition of \(B\),
and in polynomial time.  Enumerating the final online seed therefore outputs
the entries of \(\widehat W_{0..n}\) one at a time within the claimed
resource bounds.
\end{proof}
\subsection{From regular ROBPs to permutation ROBPs}
\label{sec:fb-regular}

We now deduce the regular case from the permutation case.  Let \(B\) be a
length-\(n\), width-\(w\) regular ROBP over \(\Sigma\).  By
\Cref{lem:reg2perm-phase}, we can replace each state \(u\in[w]\) by
\(|\Sigma|\) copies \((u,\theta)\), \(\theta\in\Sigma\), so that every symbol
acts as a permutation on the resulting state space.  This gives a
length-\(n\), width-\(w|\Sigma|\) permutation ROBP \(B^\pi\) over the same
alphabet.

The construction preserves the original transition matrix after summing over
the copies.  Namely, for every \(u,v\in[w]\),
\begin{equation}
   [W_{0..n}]_{u,v}
   =
   \sum_{\theta\in\Sigma}
      [W^\pi_{0..n}]_{(u,\theta),(v,0)}.
   \label{eq:fb-regular-copies}
\end{equation}

\begin{proof}[Proof of \Cref{thm:perm-regular-main}]
Apply \Cref{thm:fb-perm} to \(B^\pi\), which has width
\(w|\Sigma|\), with error \(\varepsilon\) and
\(\LevFB=|\Sigma|^2\).  Let \(\widehat W^\pi_{0..n}\) be the resulting
matrix, and define
\[
   [\widehat W_{0..n}]_{u,v}
   :=
   \sum_{\theta\in\Sigma}
      [\widehat W^\pi_{0..n}]_{(u,\theta),(v,0)}
   \qquad (u,v\in[w]).
\]

By \eqref{eq:fb-regular-copies},
\[
   [W_{0..n}]_{u,v}
   =
   \sum_{\theta\in\Sigma}
      [W^\pi_{0..n}]_{(u,\theta),(v,0)}.
\]
Therefore, applying the accepting-set guarantee of
\Cref{thm:fb-perm} to the start state \((v,0)\) and the set
\(\{u\}\times\Sigma\), we obtain
\[
\begin{aligned}
   \left|
      [\widehat W_{0..n}]_{u,v}
      -
      [W_{0..n}]_{u,v}
   \right|
   &=
   \left|
      \sum_{\theta\in\Sigma}
      \left(
         [\widehat W^\pi_{0..n}]_{(u,\theta),(v,0)}
         -
         [W^\pi_{0..n}]_{(u,\theta),(v,0)}
      \right)
   \right| \\
   &\le \varepsilon.
\end{aligned}
\]
It remains to bound the resources.  Substituting the width
\(w|\Sigma|\) and \(\LevFB=|\Sigma|^2\) into the bounds of
\Cref{thm:fb-perm} gives
\[
   O\!\left(
      \log|\Sigma|
      +
      \log n\log\frac{w|\Sigma|}{\varepsilon}
      +
      \log\frac{nw|\Sigma|}{\varepsilon}
   \right)
\]
space.  By \Cref{cl:reg2perm-explicit}, each transition of \(B^\pi\) can be
evaluated using \(O(\log(w|\Sigma|))\) additional workspace and polynomially
many calls to the transition oracle for \(B\).  This is absorbed by the
displayed space bound, and the total running time is polynomial in
\(n,w,|\Sigma|,1/\varepsilon\).
\end{proof}


\section{Optimal-space derandomization at logarithmic width for bounded two-sided errors}
\label{sec:balanced-counts}

We give an evolving-set simulation of regular ROBPs that uses only
$O(w^2)$ random bits in expectation, independently of the length.
Combining it with Nisan--Zuckerman randomness reduction gives
deterministic space $O(\log(n+2)+w)$ at inverse-polynomial-in-$w$
accuracy, with polynomial running time when $w=O(\log n)$.
Extensions to smaller error and polylogarithmic alphabets appear in
\Cref{subsec:evolving-extensions}.

We first consider a binary alphabet.  We use the transition-matrix
notation from the preliminaries: $B_t(\sigma)$ maps $V_{t-1}$ to $V_t$,
\[
   W_t=\frac{B_t(0)+B_t(1)}2,
   \qquad
   W_{0..n}=W_n\cdots W_1.
\]
Distributions are column vectors, and the matrix $1$-norm is the induced
norm, namely the maximum absolute column sum.  Regularity means that
every $W_t$ is doubly stochastic; the individual maps $B_t(0)$ and
$B_t(1)$ need not be permutations.  The program is given by its
transition table on a read-only input tape.  For implicitly represented
programs, the resources needed to evaluate a transition are added to the
bounds below.

\subsection{An evolving-set simulation at logarithmic width}
\label{subsec:evolving-sets}

We evolve a random subset backwards through the program, using the
evolving-set idea of Morris and Peres~\cite{morrisPeresEvolvingSets2005}.  A length-independent bound on
the expected number of random bits, followed by the randomness reduction
of Nisan and Zuckerman~\cite{nisanRandomnessLinearSpace1996}, removes the
factor $\log w$ from the space bound at inverse-polynomial-in-$w$
accuracy.  The resulting running time is polynomial when
$w=O(\log n)$.

\begin{theorem}[Evolving-set simulation of regular ROBPs]
\label{thm:evolving-regular}
Fix a constant $C>0$.  Given the transition table of a binary regular
ROBP of length $n$ and width $w$, a start state $s$, an arbitrary
accepting set $F\subseteq[w]$, and an integer $b\ge1$ specifying
$\varepsilon=2^{-b}\ge w^{-C}$, a deterministic algorithm
approximates its acceptance probability to additive error $\varepsilon$
in space
\begin{equation}
   O_C\!\left(\log n+w\right)
   \label{eq:evolving-space}
\end{equation}
and time $2^{O_C(w)}\operatorname{poly}(n,w)$.
The same resources suffice to output a matrix
$\widehat W_{0..n}$ satisfying
\begin{equation}
   \|\widehat W_{0..n}-W_{0..n}\|_1\le\varepsilon.
   \label{eq:evolving-matrix}
\end{equation}
In particular, binary regular ROBPs of width $O(\log n)$, with arbitrary
accepting sets, admit deterministic polynomial-time, logarithmic-space
bounded-error evaluation.  At this width the matrix approximation also
has polynomial running time.
\end{theorem}

\begin{proof}[Proof of \Cref{thm:evolving-regular}]
We first give a randomized scalar estimator whose expected randomness
usage is independent of the length.  Write $\delta_t(u,0)$ and
$\delta_t(u,1)$ for the two successors of state $u$ in layer $t-1$.
Starting with $S^{(0)}=F$, process the layers backwards.  At reverse
step $j$, put $t=n-j$ and define
\begin{align*}
   A_j&=\{u:\delta_t(u,0)\in S^{(j)}
                    \text{ and }\delta_t(u,1)\in S^{(j)}\},\\
   B_j&=\{u:\delta_t(u,0)\in S^{(j)}
                    \text{ or }\delta_t(u,1)\in S^{(j)}\}.
\end{align*}
If $A_j=B_j$, set $S^{(j+1)}=A_j$ without using randomness.
Otherwise use one fresh fair bit to choose $A_j$ or $B_j$ with equal
probability.  The final output is $Z=\mathbf 1_{\{s\in S^{(n)}\}}$.
For every state, its membership probability after this update is the
average of the two successor membership bits.  Consequently,
\begin{equation}
   \mathbb E[\mathbf 1_{S^{(j+1)}}\mid S^{(j)}]
      =W_{n-j}^{\top}\mathbf 1_{S^{(j)}},
   \qquad
   \mathbb E Z
      =e_s^{\top}W_1^{\top}\cdots W_n^{\top}\mathbf 1_F
      =\mathbf 1_F^{\top}W_{0..n}e_s.
   \label{eq:evolving-unbiased}
\end{equation}
Thus $Z$ is an exact Bernoulli estimator of the acceptance probability.

Let $m_j=|S^{(j)}|$.  Regularity says that the total number of edges
entering a set $S$ is $2|S|$, counting multiplicity, and therefore
\[
   |A_j|+|B_j|=2m_j.
\]
Put $d_j=(|B_j|-|A_j|)/2$.  This is a nonnegative integer; it is at
least one exactly when a random bit is used.  Conditionally on the
current set, $m_{j+1}=m_j-d_j$ or $m_j+d_j$, each with probability
one half.  With $H(m)=m(w-m)$, we obtain
\[
   \mathbb E[H(m_{j+1})\mid S^{(j)}]=H(m_j)-d_j^2.
\]
If $N$ is the total number of random bits used, telescoping over the
finite $n$ layers gives
\begin{equation}
   \mathbb E N
      \le\sum_{j=0}^{n-1}\mathbb E d_j^2
      =|F|(w-|F|)-\mathbb E H(m_n)
      \le |F|(w-|F|)\le w^2/4.
   \label{eq:evolving-expected-bits}
\end{equation}
The same estimate bounds the expected remaining number of random bits
after any history, with the current subset and the remaining layers in
place of $F$ and the original program.

This uniform conditional bound gives an exponential tail in the number
of active updates.  Set $L=\lceil w^2/2\rceil$.  Conditional Markov's
inequality implies that after any history the probability of requiring
at least another $L$ random bits is at most $1/2$.  Use the full
reverse-time filtration, including the layer index and all coins read
so far.  At the stopping time immediately after $rL$ coins, the unused
coin suffix is independent, so the same bound applies to the remaining
computation.  Induction gives
\begin{equation}
   \Pr[N\ge rL]\le 2^{-r}\qquad(r\ge1).
   \label{eq:evolving-tail}
\end{equation}
For a desired scalar error $\eta$, choose
$r=\lceil\log(2/\eta)\rceil$ and $K=rL$.  Modify the estimator to
output zero if it would require its $(K+1)$st random bit.  The resulting
estimator $Z_K$ uses at most $K=O(w^2\log(1/\eta))$ random bits and
satisfies
\begin{equation}
   |\mathbb E Z_K-\mathbb E Z|
      \le\Pr[N>K]\le\eta/2.
   \label{eq:evolving-truncation}
\end{equation}
It stores the current and next subsets, a layer index, and an active-bit
counter, using $O(w+\log n +\log\log(1/\eta))$ space.  Each update
can be implemented by scanning the transition table and subset bit
vectors, in polynomial time. 

Nisan--Zuckerman's theorem
\cite[Theorem~1]{nisanRandomnessLinearSpace1996} states, in particular,
that a space-$S$, time-$T$ algorithm using $\operatorname{poly}(S)$
random bits can be simulated in $O(S)$ space and
$T+\operatorname{poly}(S)$ time using $O(S)$ random bits, with output
statistical error at most $\exp(-\sqrt S)$.  Here $S$ must also be at
least the logarithm of the input length.  The binary transition table
has length $N_{\mathrm{in}}=O(nw\log w)$; including the start state,
accepting set, and the binary encoding of $b$ does not change the
logarithmic bound needed here.  For $\eta\ge w^{-O(1)}$, taking
$S=O(\log n+w)$ with a sufficiently large constant ensures both
$K=\operatorname{poly}(S)$ and simulation error at most $\eta/2$.
Enumerate all seeds of this simulation, count its accepting seeds, and
divide by the number of seeds.  The seed and counter both use $O(S)$
bits.  This computes a nonnegative dyadic number in $[0,1]$ within
$\eta$ of the acceptance probability, using $O(S)$ space and
$2^{O(S)}\operatorname{poly}(n,w)$ time. 
\end{proof}

\subsection{Extensions to smaller error and larger alphabet}
\label{subsec:evolving-extensions}

\begin{remark}[Smaller error]
\label{rem:evolving-precision}
For every fixed $\tau>0$ and integer $b\ge1$, a deterministic
algorithm, given a binary regular ROBP of length $n$ and width $w$,
outputs a column-stochastic dyadic matrix $\widehat W_{0..n}$ with
$\|\widehat W_{0..n}-W_{0..n}\|_1\le2^{-b}$ in space
\begin{equation}
   O_\tau\!\left(\log n+w+b^{1+\tau}\right)
   \label{eq:evolving-precision-space}
\end{equation}
and time
$2^{O_\tau(w)}\operatorname{poly}_\tau(n,w,2^b)$.
For $w=O(\log n)$, the running time is polynomial in $n$ and
$1/\varepsilon$, where $\varepsilon=2^{-b}$.
\end{remark}

\begin{proof}[Proof sketch]
Put $a=\lceil\log(n+2)\rceil+w$ and consider scalar error
$\eta=2^{-q}$ for an integer $q\ge1$.  Truncate the estimator after
$K=O(w^2(q+1))$ random bits, with error at most $\eta/2$.
Skipping the deterministic updates between random-bit reads gives
a binary ROBP of length $K$ and width at most
\[
   V=(n+1)2^w+2.
\]
A state records the layer and current subset, with two absorbing
states for the output after halting.  Each transition is computable in $O(a)$ space
and $\operatorname{poly}(n,w)$ time.

We evaluate this short ROBP by repeatedly halving its length while
keeping its input alphabet small.  After padding the length to a
power of two, each round merges every two consecutive layers into
one.  The merged layer reads a pair of current symbols and performs
the two transitions in sequence, so this step preserves the average
transition matrix exactly.  We then apply
\Cref{lem:reducing-alphabet-size-with-sampler} to find an offline seed of a sampler,
shared by all merged transition layers, that the sampler approximates their average
transitions.  We store its fixed description and use its short seed
as the input symbol of the new ROBP.  Repeating these two steps
leaves a single layer, whose acceptance probability is computed by
enumerating its input seeds.

There are at most $\ell=\lceil\log(K+1)\rceil$ rounds.  Allocate
matrix error $\eta/(2\ell)$ to each application of the sampler lemma
and put $d=O(\log((K+2)V/\eta))=O(a+q)$.  Each new input symbol and
each fixed sampler description has $O(d)$ bits.  Thus all descriptions
occupy $O(d\ell)$ space, while each sampler search and the final
seed enumeration range over $2^{O(d)}$ strings.  Evaluating a reduced
transition makes at most $O(K)$ original transition queries and takes
polynomial time in $K,n,w,d$.  The accumulated sampling error is at
most $\eta/2$; together with truncation, the total error is at most
$\eta$.  The resulting algorithm uses space
\[
   O\!\left((a+q)\log(a+q)\right),
\]
and time $2^{O(w)}\operatorname{poly}(n,w,2^q)$.

The sampler construction works at every precision.  To obtain
\eqref{eq:evolving-precision-space}, we combine it with
Nisan--Zuckerman, which avoids the logarithmic space factor when
$q$ is small.  If $q\le a^{2/(2+\tau)}$, then
$K=\operatorname{poly}(a)$, and Nisan--Zuckerman with
$S=c_\tau a$ has error
$\exp(-S^{2/(2+\tau)})\le\eta/2$ for a sufficiently large constant
$c_\tau$.  Enumerating its $O_\tau(a)$-bit seeds uses
$O_\tau(a)$ space and $2^{O_\tau(a)}\operatorname{poly}(n,w)$ time.
For $q>a^{2/(2+\tau)}$, use the sampler construction.  Now
$a<q^{1+\tau/2}$, so its space bound is
\[
   (a+q)\log(a+q)
      =O_\tau(q^{1+\tau/2}\log(q+2))
      =O_\tau(q^{1+\tau}).
\]
Thus, selecting between the two simulations gives space
$O_\tau(a+q^{1+\tau})$ and time
$2^{O_\tau(w)}\operatorname{poly}_\tau(n,w,2^q)$.

\end{proof}

\begin{corollary}[Polylogarithmic alphabets]
\label{cor:evolving-alphabet}
Fix constants $A,C>0$.  Given a regular ROBP of length $n$, width
$w=O(\log n)$, and alphabet size $m\le(\log(n+2))^A$, a start state,
an arbitrary accepting set, and an integer $b\ge1$ with
$\varepsilon=2^{-b}\ge(\log(n+2))^{-C}$, a deterministic algorithm
approximates its acceptance probability to additive error
$\varepsilon$ in $O(\log n)$ space and $\operatorname{poly}(n)$ time.
The same bounds suffice to output a column-stochastic dyadic matrix
$\widehat W_{0..n}$ satisfying
$\|\widehat W_{0..n}-W_{0..n}\|_1\le\varepsilon$.
\end{corollary}

\begin{proof}[Proof sketch]We describe the algorithm. Initialize the set $S$ as the accept node set.
For a current set $S$, let $h(u)$ be the number of outgoing
edges from node $u\in V$ into $S$, and define threshold sets
\[
   S_j=\{u:h(u)\ge j\}\qquad(j\in[m]).
\]
Choose $J$ uniformly from $[m]$ and update $S$ to $S_J$.  If all threshold
sets coincide, perform the update deterministically. Proceed until $S$ is a subset of the first layer of nodes. Each state $u$
belongs to the new set with probability $h(u)/m$. Hence a similar induction as the binary case show that, the membership of start node in $S$
is an unbiased estimator of acceptance probability.  Regularity gives $\mathbb E[|S_J|\mid S]=|S|$.
Whenever we really need randomness to choose a threshold set, the sets $S_j$ have
nonconstant sizes. In such case,
\[
   \operatorname{Var}(|S_J|\mid S)
      =\frac1m\sum_{j=1}^m(|S_j|-|S|)^2\ge\frac2m.
\]
The potential $H(x)=x(w-x)$ drops by this conditional variance.
Since $H\le w^2/4$, the expected number of random choices
is at most $mw^2/8$.

Rejection sampling generates a uniformly sampled set using at most
$2\lceil\log m\rceil$ fair bits in expectation.  Thus the total
expected number of fair bits is $O(mw^2\log m)$.
Counting every bit, including rejected attempts, truncate after
$O(mw^2\log m/\eta)$ bits to incur error at most $\eta/2$ by
Markov's inequality.  Under the stated bounds, taking
$\eta=\varepsilon/(8w)$ gives $\operatorname{polylog}(n)$ random
bits, $O(\log n)$ workspace, and polynomial running time.
Nisan--Zuckerman reduces the randomness to $O(\log n)$ bits with
error at most $\eta/2$.  Enumerating its seeds gives the 
estimate.
\end{proof}

\section{Reachability in regular ROBPs}
\label{sec:regular-positivity}

Let $B$ be a binary regular ROBP of length $n$ and width $w$, with start state $s$ and accepting set $F\subseteq[w]$. We consider the problem of deciding whether an accepting state is reachable from $s$, or equivalently, whether $\mathbf 1_F^\top W_{0..n}e_s>0$.

\begin{theorem}[Reachability in narrow regular ROBPs]
\label{thm:regular-positivity-sc}
Fix a constant $C>0$. Given a binary regular ROBP of length $n$ and width $w$ satisfying $\log^2 w\le C\log n$, a start state $s$, and an accepting set $F$, a deterministic algorithm decides whether $\mathbf 1_F^\top W_{0..n}e_s>0$ in space
\begin{equation}
   O_C\!\left(\log n+
      \frac{\log n\log w}
           {\max\!\left\{1,\log\!\left(\frac{\log n}{\log^2 w}\right)\right\}}\right)
   \label{eq:regular-positivity-sc-space}
\end{equation}
and time $\operatorname{poly}_C(n,w)$.
\end{theorem}

\begin{corollary}[Polylogarithmic width]
\label{cor:regular-positivity-polylog}
For every fixed $K>0$, reachability in binary regular ROBPs of width $w\le(\log n)^K$ is decidable in space $O_K(\log n)$ and time $\operatorname{poly}_K(n)$.
\end{corollary}

\begin{proof}
Here $\log w=O_K(\log\log n)$, so the hypothesis of \Cref{thm:regular-positivity-sc} holds for a constant depending on $K$. Moreover,
\[
   \frac{\log w}
        {\max\!\left\{1,\log\!\left(\frac{\log n}{\log^2 w}\right)\right\}}=O_K(1).
\]
Substituting in \eqref{eq:regular-positivity-sc-space} gives the claim.
\end{proof}

For $w=O(\log n)$, storing the set of reachable states already gives a logarithmic-space algorithm. \Cref{cor:regular-positivity-polylog} extends this width range to every fixed power of $\log n$. Braverman, Rao, Raz, and Yehudayoff~\cite{bravermanPseudorandomGeneratorsRegular2014} showed that the strings of Hamming weight less than $w$ form a hitting set for width-$w$ regular ROBPs.

\begin{remark}[NL-completeness]
\label{rem:regular-reachability-nl}
For every fixed integer $c\ge1$, reachability in binary regular ROBPs is $\mathrm{NL}$-complete even when $w\le n^{1/c}$. An SC algorithm for these widths would therefore imply $\mathrm{NL}\subseteq\mathrm{SC}$.

For hardness, let $G$ be a directed graph with $N$ vertices and $M\ge1$ edges, with designated vertices $s,t$. Scan its edges $N$ times in a fixed order. For each occurrence of an edge $u\to v$, introduce a fresh auxiliary state $a$ and append two layers. Their $0$-transitions are identities, and their $1$-transitions are the transpositions $(u\,a)$ and $(a\,v)$, respectively. Each auxiliary state is fixed outside its designated pair of layers. Both transitions in every layer are permutations, so the ROBP is regular.

The state $a$ is unreachable before its designated pair of layers and cannot be left afterward. Hence this pair preserves all reachable original vertices and makes $v$ reachable if $u$ is reachable; no other original vertex becomes reachable. After $N$ scans, the reachable original vertices are exactly those reachable from $s$ in $G$, since every reachable vertex has a simple path of length at most $N-1$. Taking $s$ as the start state and $t$ as the sole accepting state gives a logspace reduction with length $n=2NM$ and width $w=N+NM=O(n)$.

For any fixed integer $c\ge1$, append identity layers until the length is $n':=\max\{n,w^c\}$. This preserves acceptance and gives $w\le(n')^{1/c}$. Since $c$ is fixed, the padding is computable in logspace with polynomial output size.
\end{remark}

\subsection{Relabeling the states}
\label{subsec:positivity-lazy-coordinates}

We identify a permutation with its permutation matrix, so that $Pe_u=e_{P(u)}$. Although the two labeled transitions of a binary regular layer need not be permutations, its edges can be partitioned into two perfect matchings.

\begin{lemma}[State relabeling]
\label{lem:positivity-lazy-coordinates}
For a binary regular ROBP, there are permutations $P_i,Q_i$ satisfying $W_i=(P_i+Q_i)/2$. Define
\begin{equation}
   H_0=I,
   \qquad H_i=P_iH_{i-1},
   \qquad R_i=H_i^{-1}Q_iH_{i-1}.
   \label{eq:positivity-coordinate-permutations}
\end{equation}
Then
\begin{equation}
   C_i:=H_i^{-1}W_iH_{i-1}=\frac{I+R_i}{2},
   \qquad W_{0..n}=H_nC_n\cdots C_1.
   \label{eq:positivity-lazy-coordinates}
\end{equation}
Images and preimages under $P_i,Q_i,H_i,R_i$ can be computed in $O(\log(nw))$ space and $\operatorname{poly}(n,w)$ time. The same bounds hold for membership in $F':=H_n^{-1}(F)$.
\end{lemma}

\begin{proof}
Form the bipartite multigraph of layer $i$, with the source states on the left and destination states on the right. Every vertex has degree two, counting multiplicity, so each connected component is an even cycle. Alternating two colors around each component partitions its edges into two perfect matchings. Their permutation matrices are $P_i$ and $Q_i$, and each original edge contributes once to their sum.

To query these matchings consistently, identify an edge by its source and original binary label. Starting from the queried edge, traverse its cycle to find the least edge identifier. Restart from this edge in the left-to-right direction and assign alternating colors. A traversal step takes the other edge at the current endpoint. At a left vertex this edge is specified directly; at a right vertex it is found by scanning the transition table. A constant number of edge identifiers, a direction bit, and a parity bit suffice. The procedure also covers a component consisting of two parallel edges. It determines the color of any edge in logarithmic space and polynomial time, and hence answers image queries for both matchings. Preimages are found by scanning candidate sources.

The first identity in \eqref{eq:positivity-lazy-coordinates} follows from $H_i=P_iH_{i-1}$, and the product identity telescopes. To apply $H_i$ to one state, scan $P_1,\ldots,P_i$ in order, keeping the current state. To apply $H_i^{-1}$, scan their inverses in reverse order. Combining these routines gives the queries for $R_i$ and its inverse. Finally, $v\in F'$ if and only if $H_n(v)\in F$. These computations use logarithmic space and polynomial time.
\end{proof}

The matrices $C_i$ are obtained by relabeling the states in each layer. The matching colors may differ from the original input labels. After relabeling, a path can keep the same state label at every layer, and the acceptance probability is $\mathbf 1_{F'}^\top C_n\cdots C_1e_s$.

\begin{lemma}[Paths with few state changes]
\label{lem:positivity-short-witness}
Let $C_1,\ldots,C_n$ be nonnegative $w\times w$ matrices with positive diagonal, and let $0\le t\le t'\le n$. A path from $u$ at time $t$ to $v$ at time $t'$ is a sequence $x_t,\ldots,x_{t'}\in[w]$ such that $x_t=u$, $x_{t'}=v$, and $(C_i)_{x_i,x_{i-1}}>0$ for every $t<i\le t'$. A state change occurs at step $i$ if $x_i\ne x_{i-1}$.

If such a path exists, then there is one with at most $\min\{t'-t,w-1\}$ state changes.
\end{lemma}

\begin{proof}
Choose a path $x_t,\ldots,x_{t'}$ with the fewest state changes. Suppose that it leaves some state and later returns to it. Then there are times $a<b$ with $x_a=x_b$ such that the segment from time $a$ to time $b$ contains a state change. Replace this segment by staying at $x_a$ throughout the same interval. Every replacement step is valid because $(C_i)_{x_a,x_a}>0$ for $a<i\le b$. The states at times $a$ and $b$ are unchanged, as are all transitions outside this interval. Thus the replacement gives a path with the same endpoints at the same times and fewer state changes, a contradiction.

Consequently, each state change enters a state that has not previously been visited. Including the initial state, the path therefore visits one distinct state for each state change, so there are at most $w-1$ changes. There are also at most $t'-t$ changes, since each of the $t'-t$ steps contributes at most one.
\end{proof}

\subsection{A recursive reachability algorithm}
\label{subsec:positivity-earliest-arrival}

Let $C_1,\ldots,C_n$ be nonnegative $w\times w$ matrices with positive diagonal, and let $\mathcal T:=\{0,\ldots,n\}\cup\{\infty\}$. For $u,v\in[w]$, define $E_{v,u}:\mathcal T\to\mathcal T$ by $E_{v,u}(\infty)=\infty$ and, for finite $t$,
\begin{equation}
   E_{v,u}(t):=
   \begin{cases}
      t,&v=u,\\
      \min\{i:t<i\le n,\ (C_i)_{v,u}>0\},&v\ne u,
   \end{cases}
   \label{eq:positivity-one-change}
\end{equation}
where the minimum of the empty set is $\infty$. Thus $E_{v,u}(t)$ is the earliest arrival time at $v$ when starting from $u$ at time $t$ and using at most one state change. Each function $E_{v,u}$ is non-decreasing because starting later only removes
available transitions.

For matrices $X,Y$ whose entries are nondecreasing functions $\mathcal T\to\mathcal T$ preserving $\infty$, define
\begin{equation}
   (XY)_{v,u}(t):=\min_{z\in[w]}X_{v,z}\bigl(Y_{z,u}(t)\bigr).
   \label{eq:positivity-function-product}
\end{equation}
This product describes first reaching an intermediate state $z$ according to $Y$ and then continuing to $v$ according to $X$, choosing $z$ to minimize the arrival time. The multiplication is associative.

For every $m\ge1$, $(E^m)_{v,u}(t)$ is the earliest arrival time at $v$ from $u$ at time $t$ using at most $m$ state changes. This follows by induction: multiplication by $E$ appends at most one state change, with $E_{u,u}(t)=t$ allowing no change.

Put $h:=\min\{n,w-1\}$. For the program obtained in \Cref{lem:positivity-lazy-coordinates}, \Cref{lem:positivity-short-witness} gives, for every $m\ge h$,
\begin{equation}
   \mathbf 1_F^\top W_{0..n}e_s>0
   \quad\Longleftrightarrow\quad
   \exists v\in[w]:\ H_n(v)\in F
   \text{ and }(E^m)_{v,s}(0)\le n.
   \label{eq:positivity-arrival-criterion}
\end{equation}
Only one function value is needed at a time. The following lemma evaluates it without storing the function table.

\begin{lemma}[Recursive evaluation]
\label{lem:positivity-arity-recursion}
Let $a\ge2$ be an integer and let $d$ be the least nonnegative integer satisfying $a^d\ge h$. Suppose testing whether $(C_i)_{v,u}>0$ takes space $S_0$ and time $T_0\ge1$. One value $(E^{a^d})_{v,u}(t)$ can be computed in space
\begin{equation}
   O\bigl((d+1)(\log n+a\log w)\bigr)+S_0
   \label{eq:positivity-arity-space}
\end{equation}
and time at most
\begin{equation}
   (aw^{a-1})^d\operatorname{poly}(n,w,a)T_0.
   \label{eq:positivity-arity-time}
\end{equation}
\end{lemma}

\begin{proof}
Define $\operatorname{Eval}(j,u,v,t):=(E^{a^j})_{v,u}(t)$. If $t=\infty$, return $\infty$. At level $j=0$, compute \eqref{eq:positivity-one-change} by a sequentially scanning the layers. At level $j>0$, enumerate all tuples $(z_1,\ldots,z_{a-1})\in[w]^{a-1}$, with $z_0=u$ and $z_a=v$. For each tuple initialize $\tau=t$ and sequentially perform
\begin{equation}
   \tau\leftarrow\operatorname{Eval}(j-1,z_{r-1},z_r,\tau),
   \qquad r=1,\ldots,a.
   \label{eq:positivity-recursive-update}
\end{equation}
Return the minimum final value of $\tau$ over all tuples.

For correctness, put $G=E^{a^{j-1}}$. Associativity and distributivity in \eqref{eq:positivity-function-product} give
\[
   (G^a)_{v,u}(t)
   =\min_{z_1,\ldots,z_{a-1}}
      G_{z_a,z_{a-1}}\!\left(\cdots
      G_{z_2,z_1}\bigl(G_{z_1,z_0}(t)\bigr)\cdots\right).
\]
The recursive computation in \eqref{eq:positivity-recursive-update} evaluate exactly this composition for the current tuple. Taking the minimum proves the claim.

At each recursion level, we store the endpoints, the tuple of intermediate states, the loop indices, and three time values: the starting time, the current time $\tau$, and the minimum arrival time found so far. Together, these require $O(\log n+a\log w)$ bits. The recursive calls for each tuple are evaluated sequentially, with each returned value replacing $\tau$, so their return values need not be stored simultaneously. Since the recursion has $d+1$ levels, its total workspace is $O((d+1)(\log n+a\log w))$. The additional $S_0$ bits needed to test $(C_i)_{v,u}>0$ are used only at a leaf and reused across all such tests.

Each non-leaf call enumerates $w^{a-1}$ tuples and makes at most $a$ recursive calls per tuple. The recursion tree therefore has branching factor at most $aw^{a-1}$ and depth $d$, giving $O((aw^{a-1})^d)$ nodes in total. Each leaf takes at most $\operatorname{poly}(n,w,a)T_0$ time, while tuple enumeration and loop updates require only $\operatorname{poly}(n,w,a)$ work per child. Since $T_0\ge1$, the total running time satisfies \eqref{eq:positivity-arity-time}.
\end{proof}

For binary regular ROBPs, \Cref{lem:positivity-lazy-coordinates} supplies $S_0=O(\log(nw))$ and $T_0=\operatorname{poly}(n,w)$. Enumerating the final state and checking $H_n(v)\in F$ in \eqref{eq:positivity-arrival-criterion} changes only the polynomial factor in time. Thus \eqref{eq:positivity-arity-space} and~\eqref{eq:positivity-arity-time} also bound the reachability algorithm.

\subsection{Setting the parameters}
\label{subsec:positivity-parameter-bounds}

We now choose $a$ so that the algorithm runs in polynomial time.

\begin{proof}[Proof of \Cref{thm:regular-positivity-sc}]
Set
\[
   a:=\max\!\left\{2,
      \left\lfloor\frac{\lceil\log n\rceil}
                            {\lceil\log w\rceil^2}\right\rfloor\right\}.
\]
Since $h\le w-1$, the depth in \Cref{lem:positivity-arity-recursion} satisfies
\[
   d=O\!\left(1+
      \frac{\log w}
           {\max\!\left\{1,\log\!\left(\frac{\log n}{\log^2 w}\right)\right\}}\right).
\]
The choice of $a$ also gives
\begin{align*}
   a\log w
      &=O\!\left(\log w+\frac{\log n}{\log w}\right),\\
   a\log^2 w
      &=O\bigl(\log^2 w+\log n\bigr)
       =O_C(\log n).
\end{align*}
Substituting these bounds in \eqref{eq:positivity-arity-space} gives \eqref{eq:regular-positivity-sc-space}. For the time bound, $d=O(\log w)$ and $\log a\le a\log w$, so
\[
   \log\bigl((aw^{a-1})^d\bigr)
   \le d\bigl(\log a+a\log w\bigr)
   =O\bigl(a\log^2 w\bigr)=O_C(\log n).
\]
Also $a=O(\log n)$, and therefore the remaining factor in \eqref{eq:positivity-arity-time} is polynomial in $n,w$. The complete algorithm runs in $\operatorname{poly}_C(n,w)$ time.
\end{proof}

For example, when $\log w=\Theta(\log^\alpha n)$ for a fixed $0<\alpha<1/2$, the space bound is $O_\alpha(\log^{1+\alpha} n/\log\log n)$, with polynomial time. At $\log w=\Theta(\sqrt{\log n})$ it becomes $O(\log^{3/2} n)$, still with polynomial time.

If polynomial time is not required, a larger value of $a$ yields the following bound without a restriction on $n,w$.

\begin{corollary}[A space bound for arbitrary width]
\label{cor:regular-positivity-all-parameters}
Reachability in binary regular ROBPs is decidable in space
\begin{equation}
   O\!\left(\log n+\log w+
      \frac{\log n\log w}
           {\max\!\left\{1,\log\!\left(\frac{\log n}{\log w}\right)\right\}}\right)
   \label{eq:regular-positivity-all-space}
\end{equation}
and time
\begin{equation}
   \operatorname{poly}(n,w)\,
   2^{O\!\left(\frac{\log n\log w}
      {\max\{1,\log(\log n/\log w)\}}\right)}.
   \label{eq:regular-positivity-all-time}
\end{equation}
\end{corollary}

\begin{proof}
When $n\ge w$, take
\[
   a:=\max\!\left\{2,
      \left\lfloor\frac{\lceil\log n\rceil}
                            {\lceil\log w\rceil}\right\rfloor\right\}.
\]
Then $a\log w=O(\log n)$ and
\[
   d=O\!\left(1+
      \frac{\log w}
           {\max\!\left\{1,\log\!\left(\frac{\log n}{\log w}\right)\right\}}\right).
\]
Substituting in \Cref{lem:positivity-arity-recursion} gives the space bound and a time exponent
\[
   O\!\left(\log n+
      \frac{\log n\log w}
           {\max\!\left\{1,\log\!\left(\frac{\log n}{\log w}\right)\right\}}\right).
\]
The $O(\log n)$ term is absorbed by $\operatorname{poly}(n,w)$. When $n<w$, use $a=2$. Now $h\le n$, so $d=O(\log n)$ and $\log n\le\log w$. The resulting space is $O(\log n\log w)$ and time is $\operatorname{poly}(n,w)2^{O(\log n\log w)}$, which agree with \eqref{eq:regular-positivity-all-space} and~\eqref{eq:regular-positivity-all-time} because $\max\{1,\log(\log n/\log w)\}=1$.
\end{proof}

\section{Reachability and approximation for regular matrix powers}
\label{app:regular-powering}

Let $A$ be the common layer average of a binary regular ROBP, so $W_i\equiv A$. For a start state $s\in[w]$, an accepting set $F\subseteq[w]$, and an exponent $n$ in binary, write
\[
   p_n=\mathbf1_F^\top A^ne_s,
   \qquad
   R_n=\{v\in[w]:(A^n)_{v,s}>0\}.
\]

Let $G$ be the directed multigraph with $2A_{v,u}$ edges from $u$ to $v$. Every vertex has indegree and outdegree two. In particular, every weak component is strongly connected.

\subsection{Reachability}
\label{subsec:regular-powering-computation}

\begin{theorem}[Reachability for regular matrix powers]
\label{thm:regular-powering-positivity}
\label{lem:regular-powering-support}
Given $A,s,F$, and a $b$-bit integer $n\ge w^2-1$, whether $p_n>0$ can be decided deterministically in space $O(\log(w+b))$ and time $\operatorname{poly}(w,b)$.
\end{theorem}

\begin{proof}
We first test whether $(A^n)_{v,s}>0$ for a given $v\in[w]$. For each $1\le m\le w$, form a directed graph $G_m$ on $[w]\times\mathbb Z_m$ by replacing each edge $u\to v$ of $G$ with the edges
\[
   (u,z)\longrightarrow(v,z+1\bmod m),
   \qquad z\in\mathbb Z_m.
\]
The second coordinate records the walk length modulo $m$. Let $U_m$ be the underlying undirected graph of $G_m$. Since every vertex of $G_m$ has equal indegree and outdegree, each weakly connected component is strongly connected. We show that
\begin{equation}
   (A^n)_{v,s}>0
   \quad\Longleftrightarrow\quad
   \text{$(s,0)$ and $(v,n\bmod m)$ are connected in $U_m$ for every $1\le m\le w$.}
   \label{eq:regular-powering-support-test}
\end{equation}

Any length-$n$ walk from $s$ to $v$ lifts to a walk from $(s,0)$ to $(v,n\bmod m)$ in every $G_m$, proving the forward implication. Conversely, suppose that all the connections on the right-hand side hold. Choose a simple directed cycle through $s$ in $G$, and let $c\le w$ be its length. For $m=c$, the two vertices $(s,0)$ and $(v,n\bmod c)$ belong to the same strongly connected component of $G_c$. They are therefore joined by a directed simple path whose length $\ell$ satisfies
\[
   \ell\le wc-1\le w^2-1\le n.
\]
Projecting this path onto $G$ gives an $s$-to-$v$ walk of length $\ell$, with $\ell\equiv n\pmod c$. Traversing the cycle $(n-\ell)/c$ times before following this walk gives a walk of length exactly $n$, proving the reverse implication.

To implement the test, observe that each $U_m$ has at most $w^2$ vertices and its adjacency can be computed in $O(\log w)$ space. Reingold's algorithm~\cite{reingoldUndirectedConnectivity2008} thus decides connectivity in $O(\log w)$ space and time polynomial in $w$. The residue $n\bmod m$ can be computed by scanning the $b$ input bits, using $O(\log w+\log b)$ space and time polynomial in $w$ and $b$. Testing all moduli sequentially and reusing the workspace gives the claimed space and time bounds for each target. Enumerating $v\in F$ then decides whether $p_n>0$ within the same bounds.
\end{proof}

\begin{remark}[NL-completeness]
\label{rem:regular-powering-nl}
Reachability for binary regular ROBPs with identical layers is $\mathrm{NL}$-complete for $n=O(w)$ under logspace many-one reductions.

The problem belongs to $\mathrm{NL}$ by nondeterministically following an $n$-step walk.
For hardness, take the permutation ROBP constructed in \Cref{rem:regular-reachability-nl}, with length $\ell$, width $u$, start state $s_0$, sole accepting state $t_0$, and transition permutations $\pi_{i,\sigma}$. On $\{0,\ldots,\ell\}\times[u]$, define
\[
   \Pi_\sigma(z,v)=
   \begin{cases}
      (z+1,\pi_{z+1,\sigma}(v)), & 0\le z<\ell,\\
      (0,v), & z=\ell.
   \end{cases}
\]
Each $\Pi_\sigma$ is a permutation. Let $A$ be the average of their permutation matrices. Starting at $(0,s_0)$, the first $\ell$ steps simulate the source ROBP, so $(A^\ell)_{(\ell,t_0),(0,s_0)}>0$ exactly when that ROBP accepts some input. The resulting instance has exponent $n=\ell$ and width $w=(\ell+1)u>n$. Its transition table is computable in logspace and has polynomial size.
\end{remark}

\subsection{Approximation}
\label{subsec:regular-powering-convergence}

For larger exponents, we approximate the acceptance probability $p_n$ by the fraction of reachable states that are accepting,
\[
   q_n=\frac{|F\cap R_n|}{|R_n|}.
\]

\begin{proposition}[Approximation of regular matrix powers]
\label{thm:regular-powering-approximation}
Given $A,s,F$, and a $b$-bit integer $n\ge(w-1)w^2$, a deterministic algorithm computes the numerator and denominator of $q_n$ in space $O(\log(w+b))$ and time $\operatorname{poly}(w,b)$. For every integer $k\ge1$ such that $n\ge k(w-1)w^2$,
\begin{equation}
   |p_n-q_n|\le(5/48)^k.
   \label{eq:regular-powering-approximation}
\end{equation}
\end{proposition}

For $k=1$, the error is less than $1/6$, so comparing $q_n$ with $1/2$ distinguishes $p_n\le1/3$ from $p_n\ge2/3$.

Consider the component containing $s$. Let $d$ be its period, the greatest common divisor of its directed cycle lengths, and let $C_0,\ldots,C_{d-1}$ be its cyclic classes, with $s\in C_0$. Every edge goes from $C_i$ to $C_{i+1}$, where indices are read modulo $d$. Counting the edges between consecutive classes shows that all classes have the same size $h$. The walk is supported on $C_t$ at time $t$, and $A$ maps the uniform distribution on $C_t$ to the uniform distribution on $C_{t+1}$.

We first bound how long a reachable set can keep the same size.

\begin{lemma}[Growth of the reachable set]
\label{lem:regular-powering-growth}
Let $S$ be a nonempty proper subset of a cyclic class. More than $|S|$ vertices are reachable from $S$ by walks of length exactly $w-1$.
\end{lemma}

\begin{proof}
Let $T_i$ be the set of vertices reachable from $S$ in exactly $i$ steps, so $T_0=S$. The $2|T_i|$ edges leaving $T_i$ all enter $T_{i+1}$, which has $2|T_{i+1}|$ incoming edges. Hence $|T_{i+1}|\ge|T_i|$. If equality holds, every incoming edge of $T_{i+1}$ comes from $T_i$.

Suppose $|T_{w-1}|=|S|$. All these inequalities are then equalities, so every edge $u\to v$ satisfies
\begin{equation}
   u\in T_i\quad\Longleftrightarrow\quad v\in T_{i+1}
   \qquad(0\le i<w-1).
   \label{eq:regular-powering-membership}
\end{equation}
Start with the cyclic classes and split the class containing $S$ into $S$ and its complement within that class. For each successive set $T_i$, $1\le i\le w-1$, replace every current group $Q$ by the nonempty sets among $Q\cap T_i$ and $Q\setminus T_i$. The groups remain disjoint and nonempty, so there are at most $w$ of them. Initially there are $d+1$ groups, and every step that splits a group increases their number. Therefore some $T_{t+1}$ causes no split.

At this step, vertices in the same group belong to the same cyclic class and agree on membership in $T_0,\ldots,T_{t+1}$. If two edges end in this group, \eqref{eq:regular-powering-membership} shows that their starting vertices agree on membership in $T_0,\ldots,T_t$. These starting vertices also lie in the same cyclic class, so they belong to the same group. Thus all edges entering a given group come from a single group.

Form a directed graph with one vertex for each group and an edge whenever an original edge connects the corresponding groups, keeping only one edge per ordered pair. This graph is strongly connected, and each vertex has exactly one incoming neighbor. If there are $r$ groups, there are therefore $r$ edges. Strong connectivity gives each vertex at least one outgoing neighbor, so each has exactly one. The groups consequently form a single directed cycle. Every closed walk in the original graph has length divisible by $r$, so $r\mid d$. But $r\ge d+1$, a contradiction.
\end{proof}

Starting from $\{s\}$, \Cref{lem:regular-powering-growth} shows that the number of reachable vertices increases at least once every $w-1$ steps until it reaches $h$. Once a whole cyclic class is reachable, the whole next class is reachable at the next step. Thus
\begin{equation}
   R_t=C_t
   \qquad\text{for all }t\ge(w-1)(h-1).
   \label{eq:regular-powering-full-class}
\end{equation}
We now show that the distribution also becomes close to uniform on this class.

\begin{lemma}[Mixing on the reachable set]
\label{lem:regular-powering-mixing}
For every integer $k\ge1$ and $n\ge k(w-1)w^2$,
\begin{equation}
   \left\|A^ne_s-\operatorname{Unif}(R_n)\right\|_{\mathrm{TV}}
   \le(5/48)^k.
   \label{eq:regular-powering-mixing}
\end{equation}
\end{lemma}

\begin{proof}
If $h=1$, the walk is deterministic and the claim holds exactly. Assume $h\ge2$. We use the evolving-set process from \Cref{subsec:evolving-sets}, here moving forwards through the graph. Start with $S_0=\{s\}$. At each step, with probability $1/2$, let $S_{t+1}$ be the set of vertices whose two incoming edges both come from $S_t$; otherwise, let it be the set of vertices with at least one incoming edge from $S_t$. The conditional probability that $v$ belongs to $S_{t+1}$ is $(A\mathbf1_{S_t})_v$. Therefore
\[
   S_t\subseteq C_t,\qquad
   \mathbb E\mathbf1_{S_t}=A^te_s,\qquad
   \mathbb E|S_t|=1.
\]

Put $\Phi(S)=|S|(h-|S|)$. The preceding identities and the triangle inequality give
\begin{align}
   \left\|A^te_s-\operatorname{Unif}(C_t)\right\|_{\mathrm{TV}}
   &\le \mathbb E\!\left[
      \frac12\left\|\mathbf1_{S_t}
      -\frac{|S_t|}{h}\mathbf1_{C_t}\right\|_1
   \right]\notag\\
   &=\frac{\mathbb E\Phi(S_t)}h.
   \label{eq:regular-powering-potential-error}
\end{align}
It suffices to bound this expected value.

Fix $S_t$ of size $j$. Its two possible successors are nested, and their sizes sum to $2j$, the number of edges leaving $S_t$. Write their sizes as $j-a$ and $j+a$, with $a$ a nonnegative integer. Then
\[
   \mathbb E[\Phi(S_{t+1})\mid S_t]
   =j(h-j)-a^2
   =\Phi(S_t)-a^2.
\]
Thus the expected value never increases, and it drops by at least one whenever the two choices differ.

Suppose $S_t$ is neither empty nor the whole class. Until the first step at which the two choices differ, the updates are deterministic and the current set is exactly the set reachable from $S_t$ after that many steps. Its size remains $|S_t|$. By \Cref{lem:regular-powering-growth}, such updates cannot continue for $w-1$ steps. The first differing choice therefore occurs within this interval, at a time determined by $S_t$. It decreases the expected value of $\Phi$ by at least one, and later updates cannot increase its expectation. Since $\Phi(S_t)\le h^2/4$,
\[
   \mathbb E[\Phi(S_{t+w-1})\mid S_t]
   \le\Phi(S_t)-1
   \le\left(1-\frac4{h^2}\right)\Phi(S_t).
\]
For the empty set and the whole class, $\Phi$ remains zero, so the final inequality holds for every possible $S_t$.

Iterating over blocks of $w-1$ steps, using $\Phi(S_0)=h-1$ and that the remaining steps cannot increase its expectation, gives
\[
   \frac{\mathbb E\Phi(S_n)}h
   \le \exp\!\left(-\frac{4\lfloor n/(w-1)\rfloor}{h^2}\right)
   \le e^{-4k}.
\]
The last inequality uses $n\ge k(w-1)w^2$ and $h\le w$. In this range, \eqref{eq:regular-powering-full-class} gives $R_n=C_n$. Now \eqref{eq:regular-powering-potential-error} and $e^{-4}<5/48$ prove the lemma.
\end{proof}

\begin{proof}[Proof of \Cref{thm:regular-powering-approximation}]
Enumerate $v\in[w]$ and apply \Cref{lem:regular-powering-support} with accepting set $\{v\}$ to count $|R_n|$ and $|F\cap R_n|$. The assumption $n\ge(w-1)w^2$ ensures that the theorem applies. Both counters use $O(\log w)$ bits, giving the numerator and denominator of $q_n$ in the stated space and time. Finally, $q_n=\mathbf1_F^\top\operatorname{Unif}(R_n)$, so \Cref{def:total-variation,lem:regular-powering-mixing} give $|p_n-q_n|\le(5/48)^k$.
\end{proof}




\section{SC derandomization for models beyong BPL}
\label{sec:inw-sc}

\subsection{The INW communication model}
\label{subsec:inw-communication-model}

We use the communication model of Impagliazzo, Nisan, and
Wigderson~\cite{impagliazzoPseudorandomnessNetworkAlgorithms1994}.
We re-describe it as the following.

Let \(H=(V,E)\) be an undirected graph with \(N=|V|\) processors.
Processor \(u \in V\) has an initial input \(x_u\) and a private random
string \(\rho_u\in\{0,1\}^r\). The protocol specifies the initial state
and local transition rules for every processor.
The computation runs in the following way. For each round, for each processor $u$, it first does its own computation. Then for each of its neighbor in $H$, $u$ sends a message.
These messages depend only on $u$'s input, private randomness, and messages received in previous rounds.
After every processor receives all the messages that sent to it, the computation goes to the next round.
The whole computation terminates within a pre-defined number of rounds.
Notice that the whole computation can be determined by fixing \(x=(x_u)_{u\in V}\) and
\(\rho=(\rho_u)_{u\in V}\) .
We denote \(\Pi(x;\rho)\) as the whole computation history i.e. it includes all the configurations, messages of the whole computation for all processors.
When \(\rho_u, u\in V\) are independent uniform distributions, we denote this joint distribution as \(U_r^V\).
\begin{definition}[Partition tree and width]
\label{def:inw-partition-tree}
A \emph{partition tree} \(\mathcal T\) for \(H\) is a rooted binary tree
whose leaves are in bijection with \(V\), and in  \(\mathcal T\) every internal
node has two children. For a node \(\nu\), let \(V_\nu\) be the set of
processors corresponding to its descendant leaves. For an internal
node with children \(\nu_{\mathrm L},\nu_{\mathrm R}\), write
\[
   A_\nu:=V_{\nu_{\mathrm L}},
   \qquad
   B_\nu:=V_{\nu_{\mathrm R}},
   \qquad
   C_\nu:=V\setminus V_\nu.
\]
Let \(\cut(\nu)\) consist of the edges with endpoints in different sets
among \(A_\nu,B_\nu,C_\nu\). The width \(\omega(\nu)\) is the minimum
size of a vertex cover that can cover all edges of \(\cut(\nu)\), and the width of the tree is
\[
   \omega(\mathcal T):=\max_{\nu\text{ internal}}\omega(\nu),
\]
with value zero for a one-leaf tree. The tree is \emph{balanced} if its
depth is \(O(\log(N+1))\). The \emph{INW width} $\omega$ of \(H\) is the minimum
width of a balanced partition tree for \(H\).
\end{definition}
We use width bounds \(\omega\ge1\) throughout.

\paragraph{Communication traces.}
For a fixed deterministic input \(x\), the \emph{communication trace}
\(\operatorname{tr}_u(x;\rho)\) of processor \(u\) records all messages
it sends and receives during the execution, also including the rounds they belong to,
their incident edges, their directions, lengths. 
\begin{definition}[Effective communication]
   A protocol has
\emph{effective communication of at most \(c\)} bits, if for every \(x\) and
\(u\in V\),
\[
   \left|
      \left\{
         \operatorname{tr}_u(x;\rho):
         \rho\in(\{0,1\}^r)^V
      \right\}
   \right|
   \le 2^c.
\] 
\end{definition}
Here \(\rho\) ranges over all choices of the processors' private random strings. 



The communication across a set of edges $E'\subseteq E$ is the complete
record of all messages transmitted along edges in $E'$
during the whole computation.
\begin{lemma}[Communication bound across a cut]
\label{lem:inw-cut-traces}
Let \(\Pi\) have effective communication of at most \(c\) bits, and let \(\nu\)
be an internal node of a partition tree. For every input,
the communication across \(\cut(\nu)\) has at most
\(2^{\omega(\nu)c}\) possibilities.
\end{lemma}

\begin{proof}
Let \(X_\nu\) be a minimum vertex cover of \(\cut(\nu)\). Every 
message crossing  \(\cut(\nu)\) is recorded by a processor in \(X_\nu\), so the whole communication is
determined by the communication traces of these processors. The number of possible communication records is therefore at most
\[
   \prod_{u\in X_\nu}2^c=2^{\omega(\nu)c}.
\]
\end{proof}
\paragraph{Measurements and the two-party generator.}
A \emph{\(k\)-measurement} \(M=(M_1,\ldots,M_k)\) is a \(k\)-bit
function of an execution. Each coordinate is computed at the end by
a fixed processor which is designated before the execution,  from its deterministic input, private
randomness, and communication history. For decision problems, we
take \(k=1\).

We use the statistical distance $\SD$ from \Cref{def:total-variation}.
We use the following form of the INW two-party generator
\cite[Theorem~1 and Lemma~1]{impagliazzoPseudorandomnessNetworkAlgorithms1994}.

\begin{theorem}[INW two-party generator]
\label{lem:inw-two-party-generator}
For all integers \(d\ge1\) and \(q,k\ge0\), and every
\(0<\eta<1/2\), there is an explicit generator
\[
   g=(g_{\mathrm L},g_{\mathrm R}):
   \{0,1\}^{\widehat d}
   \longrightarrow \{0,1\}^d\times\{0,1\}^d,
   \qquad
   \widehat d=d+O\!\left(q+k+\log(1/\eta)\right),
\]
with the following property.

Let \(\Pi\) be a two-party protocol in the communication model above,
with private random strings
\(\rho_{\mathrm L},\rho_{\mathrm R}\in\{0,1\}^d\).
Suppose that \(\Pi\) has effective communication at most \(q\).
Then, for every deterministic input \(x\) and every
\(k\)-measurement \(M\),
\[
   \SD\!\left(
      M(\Pi(x;(U_d,U_d))),
      M(\Pi(x;g(U_{\widehat d})))
   \right)
   \le\eta,
\]
where the two copies of \(U_d\) are independent.
The generator \(g\) is computable in time  
\(\poly(\widehat d)\) and workspace \(O(\widehat d)\).
\end{theorem}

\subsection{Tests for a sub-tree}
\label{subsec:inw-deterministic-evaluation}

Fix a protocol \(\Pi\) with effective communication at most \(c\), a
partition tree \(\mathcal T\) of width at most \(\omega\), a
deterministic input \(x\), a non-negative integer $k$ and a \(k\)-measurement \(M\). For simplicity of notations, we omit
\(x\) from the notation for traces and tests.

\paragraph{Boundary tests.}
For a node \(\nu\) of the partition tree \(\mathcal T\), we check if
a pseudo-random distribution inside \(V_\nu\) preserves the acceptance probability, by testing it against every possible communication that could happen between  \(V_\nu\) and its compliment.

Define the boundary of \(V_\nu\) to be
\[
   \partial V_\nu
   :=\bigl\{\{u,v\}\in E:u\in V_\nu,\ v\notin V_\nu\bigr\}.
\]
A \emph{boundary trace} specifies the messages on these edges in
both directions, including the rounds they belong to,
their incident edges, their directions, lengths, and contents.
Let \(J_\nu\subseteq[k]\) be the coordinates of $M$, which are measured by processors in \(V_\nu\). 

\begin{definition}[Boundary test]
\label{def:inw-interface-view}
For a proposed boundary trace \(\tau\) and a string
\(z\in\{0,1\}^{J_\nu}\), define
\[
   \chi_{\nu,\tau,z}:
   (\{0,1\}^r)^{V_\nu}\longrightarrow\{0,1\}
\]
as follows. Given local random strings \(\rho \in (\{0,1\}^r)^{V_\nu}\), simulate the
processors in \(V_\nu\) from their prescribed initial states,
using \(\rho\), and using \(\tau\) to supply the incoming boundary messages.
Notice that messages between processors in \(V_\nu\) are generated
by the simulation itself.

$\chi_{\nu,\tau,z}$ accepts if if the outgoing boundary communication is the same from that specified by \(\tau\), and the measured coordinates in \(J_\nu\) equal \(z\). Otherwise it rejects.

\end{definition}

For an internal node $\nu$,
\(\partial V_\nu\subseteq\cut(\nu)\), so
\Cref{lem:inw-cut-traces} bounds the number of possible boundary traces
by \(2^{\omega c}\).
At a leaf, the number of boundary traces is controlled by  \(2^{c}\).

For each node \(\nu\), Let \(\mathcal I_\nu\) be the set of its boundary traces.
Let the test family at node \(\nu\) be the set
\[
   \left\{
      \chi_{\nu,\tau,z}:
      \tau\in\mathcal I_\nu,\;
      z\in\{0,1\}^{J_\nu}
   \right\}.
\]
It has size at most
\[
   |\mathcal I_\nu|2^{|J_\nu|}
   \le B:=2^{\omega c+k}.
\]

At the root \(\varrho\), there are no boundary
messages. Take
\(\mathcal I_\varrho=\{\varnothing\}\) and \(J_\varrho=[k]\), 
\[
   \chi_{\varrho,\varnothing,y}(\rho)
   =\mathbf 1_{\{M(\Pi(x;\rho))=y\}}.
\]
Thus for each $y$, the corresponding root test simply checks whether the final measurement outcome is equal to $y$.

Next we show that any distribution preserving the expectations of
the tests corresponding to  \(V_\nu\)  suffices to replace the randomness in \(V_\nu\),
provided that it is independent of the randomness outside of \(V_\nu\).
\begin{lemma}[Substitution of randomness in a subtree]
\label{lem:inw-subtree-replacement}
Fix a node \(\nu\). Let \(D,D'\) be distributions over the support of the random
strings in \(V_\nu\), and let \(R\) be an independent distribution for those nodes
outside of \(V_\nu\).
If
\[
   \left|
      \E_D[\chi_{\nu,\tau,z}]
      -\E_{D'}[\chi_{\nu,\tau,z}]
   \right|
   \le\alpha
\]
for all \(\tau\in\mathcal I_\nu\) and
\(z\in\{0,1\}^{J_\nu}\), then
\[
   \SD\!\left(
      M(\Pi(x;D,R)),M(\Pi(x;D',R))
   \right)
   \le\frac12|\mathcal I_\nu|2^{|J_\nu|}\alpha
   \le B\alpha.
\]
\end{lemma}

\begin{proof}
Write \(J:=J_\nu\) and \(J^c:=[k]\setminus J\). Set
\[
   p_D(\tau,z):=\E_D[\chi_{\nu,\tau,z}],
   \qquad
   p_{D'}(\tau,z):=\E_{D'}[\chi_{\nu,\tau,z}].
\]
Let \(b_R(\tau,v)\) be the probability that the simulation
outside \(V_\nu\) produces measurement
outcome \(v\in\{0,1\}^{J^c}\), given the boundary trace \(\tau\).

For each fixed boundary trace $\tau\in\mathcal I_\nu$, simulate the processors
in $V\setminus V_\nu$ using random strings drawn from $R$,
with messages from $V_\nu$ given by $\tau$.
The boundary-consistency test succeeds if the outgoing boundary
records are exactly those prescribed by $\tau$.
For $v\in\{0,1\}^{J^c}$, let $b_R(\tau,v)$ be the probability,
over $R$, that this test succeeds and the
measurement outcome restricted to these bits indexed by $J^c$ is $v$.

Fix the random strings inside and outside $V_\nu$. We show that exactly one  pair of trace and measurement outcome let the random strings pass both checks:
let
$\tau^\ast,z^\ast$ be the boundary trace and measurement outcome of the actual execution.
So both boundary tests accept $\tau^\ast,z^\ast$; now for uniqueness,
suppose that both tests accept another candidate pair $\tau,z$.
Induction over the protocol steps shows that both simulations
agree with the actual execution $\tau^\ast$, and must produce the same measurement outcome $z^\ast$.

For each fixed \(\tau\), the two tests depend only on their
respective random strings, so by independence, for every
\(y\in\{0,1\}^k\),
\[
   \Pr[M(\Pi(x;D,R))=y]
   =\sum_{\tau\in\mathcal I_\nu}
      p_D(\tau,y|_J)b_R(\tau,y|_{J^c}),
\]
and the same identity holds with \(D'\) in place of \(D\).
For each fixed \(\tau\), the events defining \(b_R(\tau,v)\)
are disjoint over \(v\), so \(\sum_v b_R(\tau,v)\le1\).
Using the hypothesis
\(\lvert p_D(\tau,z)-p_{D'}(\tau,z)\rvert\le\alpha\), we obtain
\[
\begin{aligned}
   &\SD\!\left(
      M(\Pi(x;D,R)),M(\Pi(x;D',R))
   \right)\\
   &\quad\le
   \frac12\sum_{\tau,z}
      |p_D(\tau,z)-p_{D'}(\tau,z)|
      \sum_v b_R(\tau,v)\\
   &\quad\le
   \frac12|\mathcal I_\nu|2^{|J|}\alpha
   \le B\alpha.
\end{aligned}
\]
\end{proof}

\paragraph{\(\DTISP(T,S)\)-evaluable boundary tests}

We write $\DTISP(T,S)$ for the class of languages decidable
by a deterministic machine in time $T$ and workspace $S$.

For each node $\nu$, index all tests $\chi_{\nu,\tau,z}$, with
$\tau\in\mathcal I_\nu$ and $z\in\{0,1\}^{J_\nu}$, as
\[
   (\chi_{\nu,i})_{i\in[L_\nu]},
   \qquad \mbox{ where }
   L_\nu:=|\mathcal I_\nu|\,2^{|J_\nu|}.
\]

\begin{definition}[\(\DTISP(T,S)\)-evaluable boundary tests]
\label{def:inw-boundary-test-evaluable}
Let $H$ be an \(N\)-processor graph.
We say that \((\Pi, \mathcal T) \) has
\emph{\(\DTISP(T(N),S(N))\)-evaluable boundary tests}, if the following operations can be performed  in deterministic time \(T(N)\) and workspace \(S(N)\) :
\begin{enumerate}


\item Compute a description  of $\mathcal T$. The description is consisted of the root of $\mathcal T$, the parent and children for every node, corresponding processors for each leaf, \(L_\nu\) for each $\nu$.

\item  Given an input
\(x\), a node \(\nu\),  an  index \(i\in[L_\nu]\),  a string $\rho\in (\{0,1\})^{V_\nu}$,
compute \(\chi_{\nu,i}(\rho)\), where $\chi_{\nu, i}$ is defined according to $\Pi$ by \cref{def:inw-interface-view}.
\end{enumerate}
\end{definition}

\subsection{An SC simulation}
\label{subsec:inw-resource-lifting}

\begin{theorem} 
\label{thm:inw-resource-lifting}
Let $H$ be an \(N\)-processor graph.
Let \(\Pi\) be a protocol on \(H\) whose effective
communication is at most \(c\), and suppose that every processor uses at most
\(r\) private random bits. Let
\(\mathcal T\) be a balanced partition tree of depth \(h\) and width at most
\(\omega\).  Suppose that \((\Pi,\mathcal T)\) has
\(\DTISP(T,S)\)-evaluable boundary tests. For every \(0<\varepsilon<1/2\), 
there is a deterministic
algorithm that, given \(x\) and \(z\in\{0,1\}^k\), computes
\(\widetilde p_x(z)\), where \(\widetilde p_x\) is a probability
distribution satisfying
\[
   \SD\!\left(\widetilde p_x,M(\Pi(x;U_r^V))\right)
   \le\varepsilon.
\]

The algorithm runs in time
\[
   2^{O(r+\omega c+k)}\poly(N,T,h,1/\varepsilon)
\]
and workspace
\[
   O\!\left(S+r+
      \bigl(\omega c+k+\log(N/\varepsilon)\bigr)(h+1)\right).
\]

\end{theorem}

\begin{proof}
Fix \(x\). Write
\(\kappa:=\omega c+k+\log(N/\varepsilon)\), and set
\[
   B:=2^{\omega c+k},
   \qquad
   \eta:=\frac{\varepsilon}{2N},
   \qquad
   \alpha:=\frac{\varepsilon}{4NB},
   \qquad
   \gamma:=\frac1{4NB}.
\]
There are at most \(B\) boundary tests at each node, and
\(\log(1/\eta),\log(1/\alpha),\log(1/\gamma)=O(\kappa)\).
For each required output length \(m\),
\Cref{lem:good sampler} gives an \((\alpha,\gamma)\)-averaging
sampler
\[
   \Samp_m:\{0,1\}^{a_m}\times\{0,1\}^{d}
       \longrightarrow\{0,1\}^{m},
   \qquad
   a_m=m+O(\kappa),\quad d=O(\kappa).
\]
We pad the inner seeds to a common length \(d\ge1\).
Let
\[
   g=(g_{\mathrm L},g_{\mathrm R}):
      \{0,1\}^{\widehat d}
      \longrightarrow\{0,1\}^{d}\times\{0,1\}^{d}
\]
be the generator of \Cref{lem:inw-two-party-generator} for
communication bound \(\omega c\), measurement size \(k\), and error
\(\eta\). Its seed length is \(\widehat d=O(\kappa)\).

\paragraph{Constructing the subtree generators.}
For each node \(\nu\), we construct a generator
\[
   G_\nu:\{0,1\}^{d}
       \longrightarrow(\{0,1\}^{r})^{V_\nu}.
\]
We process all internal nodes at the same depth together, starting
at depth \(h\) and proceeding upward to the root at depth \(0\).

At the leaf layer, choose an outer seed \(a_{\mathrm{leaf}}\) such that for every boundary test \(\chi\) at every leaf
\[
   \left|
      \E_y \chi(\Samp_r(a_{\mathrm{leaf}},y))
      -\E_{\rho\sim U_r}\chi(\rho)
   \right|\le\alpha.
\]
For any one test, the sampler condition fails for at most a
\(\gamma\)-fraction of outer seeds. There are at most \(NB\)
leaf tests, and \(NB\gamma=1/4\), so such a seed exists.
Take the first one in lexicographic order and define
\[
   G_u(y):=\Samp_r(a_{\mathrm{leaf}},y)
   \qquad\text{for every leaf }u.
\]

Assume we have fixed generators for all nodes of depth $j+1$.
Let \(\nu\) be an internal node of depth $j$, define
\[
   \widehat G_\nu(s):=
      G_{\nu_{\mathrm L}}(g_{\mathrm L}(s))
      \sqcup G_{\nu_{\mathrm R}}(g_{\mathrm R}(s)),
\]
where \(\sqcup\) concatenates the random strings in the two disjoint sets of 
 processors. Choose one outer seed \(a_j\) satisfying
\[
   \left|
      \E_y \chi\!\left(
         \widehat G_\nu(\Samp_{\widehat d}(a_j,y))
      \right)
      -\E_s \chi(\widehat G_\nu(s))
   \right|\le\alpha
\]
for every node \(\nu\) at that depth and every boundary test
\(\chi\) at \(\nu\). The same union bound gives such an \(a_j\),
since there are
at most \(NB\) tests at this depth.
Choose the first valid outer seed and set
\[
   G_\nu(y):=
      \widehat G_\nu(\Samp_{\widehat d}(a_j,y)).
\]
Thus the construction stores one outer seed for the leaves and
at most \(h\) outer seeds for the internal depths. 

Let \(\varrho\) denote the root of \(\mathcal T\).
Once the sampler seeds are fixed, we compute
\[
   \widetilde p_x(z)
   :=\Pr_{y\sim U_d}\!\left[M(\Pi(x;G_\varrho(y)))=z\right]
\]
for each \(z\in\{0,1\}^k\) by enumerating all \(y\in\{0,1\}^d\).

\paragraph{Error analysis.}
With all sampler seeds fixed as above, we compare the original
and generated measurement distributions by a hybrid argument.
We first replace the randomness at the leaves, one leaf at a time,
and then merge sibling subtrees, one pair at a time, until reaching
the root. We bound the statistical distance introduced by each
replacement and sum these bounds to obtain the total error.

At each leaf \(u\), we replace its uniform random string by
\(G_u(U_d)\).
The choice of \(a_{\mathrm{leaf}}\) ensures that this changes
each boundary-test expectation at \(u\) by at most \(\alpha\).
By \Cref{lem:inw-subtree-replacement}, each replacement changes
the final measurement distribution by at most \(B\alpha\).
Thus all leaf replacements contribute at most \(NB\alpha\).

Consider an internal node \(\nu\) of depth \(j\).
In the preceding hybrid, the two child generators use independent
uniform \(d\)-bit seeds, also independent of the randomness outside
\(V_\nu\). 

We first replace these seeds by \(g(U_{\widehat d})\).
To bound the resulting error, fix the randomness outside \(V_\nu\)
and view the computation as a two-party protocol.
One party simulates \(A_\nu\), and the other simulates
\(B_\nu\cup C_\nu\), using the fixed outside randomness.
Their private random strings are the two child seeds.
Their communication consists of the messages on edges leaving
\(A_\nu\), all of which belong to \(\cut(\nu)\).
By \Cref{lem:inw-cut-traces}, this protocol has at most
\(2^{\omega c}\) possible communication traces.
Hence \Cref{lem:inw-two-party-generator} shows that replacing
the two child seeds by \(g(U_{\widehat d})\) changes the
measurement distribution by at most \(\eta\) in statistical distance.
This bound holds for every fixed choice of the outside randomness,
so averaging over it gives the same bound.

The subtree \(V_\nu\) is now driven by
\(\widehat G_\nu(U_{\widehat d})\).
For the second step, replace \(U_{\widehat d}\) by
\(\Samp_{\widehat d}(a_j,U_d)\), so that its random strings
are generated by \(G_\nu(U_d)\).
The choice of \(a_j\) changes each boundary-test expectation
at \(\nu\) by at most \(\alpha\).
By \Cref{lem:inw-subtree-replacement}, this changes the
measurement distribution by at most \(B\alpha\).
The parent subtree is now driven by one uniform \(d\)-bit seed,
independently of the other current subtrees, restoring the
invariant for the next merge.
Each merge therefore contributes at most \(\eta+B\alpha\).

There are \(N\) leaf replacements and \(N-1\) merges.
Finally we generate all random strings by
\(G_\varrho(U_d)\). By the triangle inequality,
\[
\begin{aligned}
   &\SD\!\left(
      M(\Pi(x;U_r^V)),
      M(\Pi(x;G_\varrho(U_d)))
   \right)\\
   &\qquad\le NB\alpha+(N-1)(\eta+B\alpha)
      \le\varepsilon.
\end{aligned}
\]
Since \(\widetilde p_x\) is the distribution of
\(M(\Pi(x;G_\varrho(U_d)))\), this proves the stated guarantee.

\paragraph{Space and time analysis.}
We first consider boundary tests.
Once the sampler seeds are fixed,  to compute a requested bit of \(G_\nu(y),\widehat G_\nu\),
follow the corresponding path from \(\nu\) to a leaf.  At an internal
node of depth \(j\), with current seed \(y'\), replace \(y'\) by the
appropriate component of

$$
   g\!\left(\Samp_{\widehat d}(a_j,y')\right);
$$

at the leaf, apply \(\Samp_r(a_{\mathrm{leaf}},y')\). Apart from tree navigation, this
takes

$$
   \poly(r+\kappa,h)
   \quad\text{time and}\quad
   O(r+\kappa+\log(h+1))
   \quad\text{workspace}.
$$

Together with the \(\DTISP(T,S)\) tree and boundary-test procedures,
a boundary test can therefore be evaluated on such an implicitly
generated random input in time

$$
   \poly(N,T,r,\kappa,h) 
   \quad\text{time and}\quad
   O(S+r+\kappa+\log(h+1))
   \quad\text{workspace}.
$$

We find \(a_{\mathrm{leaf}},a_{h-1},\ldots,a_0\) by exhaustive
search.  For each candidate, enumerate the nodes and then
boundary tests and compute the required expectations by enumerating
their uniform seeds, evaluating the generator on demand as above.
There are \(2^{r+O(\kappa)}\) candidates at the leaf level and
\(2^{O(\kappa)}\) at each depth, at most
\(NB\le2^{O(\kappa)}\) tests in each search, and at most
\(2^{O(r+\kappa)}\) uniform seeds per expectation.  Hence all sampler
seeds are found in

$$
   2^{O(r+\kappa)}\poly(N,T,\kappa,h)
$$

time.

Throughout the search, we retain the leaf sampler seed and the sampler
seeds fixed at the internal depths.  These use

$$
   r+O(\kappa(h+1))
$$

bits altogether.  All remaining workspace is reused across boundary-test
and generator evaluations, so the total workspace is

$$
   O\!\left(S+r+\kappa(h+1)\right).
$$

After all sampler seeds have been fixed, we compute
\(\widetilde p_x(z)\) by enumerating the \(d\)-bit root seed and
evaluating \(G_\varrho\) on demand.  This requires
\(2^{O(\kappa)}\poly(N,T,r,\kappa,h)\) time and the same workspace
bound.  Together with
\(\kappa=\omega c+k+\log(N/\varepsilon)\), the bounds above give the
time and space claimed in the theorem.

\end{proof}

\subsection{SC derandomization of read-\(p\) machines}
\label{subsec:inw-read-p}

We apply \Cref{thm:inw-resource-lifting} to read-\(p\) machines,
attaining the following result.

\begin{corollary}[SC derandomization of read-\(p\) machines]
\label{cor:read-p-inw-sc}
Let \(M\) be a read-\(p\) probabilistic machine running in time
\(T\) and using workspace \(S\) on \(n\)-bit inputs.
It uses a two-way read-only random tape, with at most \(p\)
visits to each cell in every computation.
Consecutive steps staying at the same cell is counted as a single visit.

For every input and every \(0<\varepsilon<1/2\), its acceptance
probability can be deterministically approximated to additive
error \(\varepsilon\) in time
\[
   2^{O(p(S+\log(n+T)))}
   \poly(n,T,1/\varepsilon)
\]
and workspace
\[
   O\!\left(
      \bigl(p(S+\log(n+T))+\log(1/\varepsilon)\bigr)
      \log(n+T)
   \right).
\]
In particular, languages decided with bounded error by such
machines with \(S=O(\log n)\), \(T=\poly(n)\), and \(p=O(1)\)
belong to \(\SC^2\).
\end{corollary}

\paragraph{The model.}
Fix an input \(x\in\{0,1\}^n\). Denote the random tape of
\(M\) as
\[
   \rho=\rho_1\cdots\rho_N\in\{0,1\}^N.
\]
The computation is deterministic once \(\rho\) is fixed, and the
randomized computation chooses \(\rho\) uniformly at random.
The random-tape head starts at cell \(1\) and, at each step,
moves to an adjacent cell or stays in place. At cell \(i\),
the machine reads the fixed bit \(\rho_i\).

For every \(\rho\): the machine
halts within \(T\) steps and visits each random-tape cell at most
\(p\) times. The initial position is at
cell \(1\). 
Staying at a cell for multiple steps is regarded as a single visit to that cell.
Set
\[
   s_0:=S+\left\lceil\log(n+T+N+2)\right\rceil.
\]
A configuration of \(M\) (excluding its fixed input), also together with the current timestamp, has an
\(O(s_0)\)-bit encoding.  

\paragraph{The path protocol.}
Let \(H\) be a simple path whose nodes are processors \(1,\ldots,N\).
Processor \(i\) has deterministic input \((x,i)\) and private random
bit \(\rho_i\).
We define the following protocol which can simulate the computation.

The protocol maintains a single token (a string) consisting a configuration of
\(M\) and also the index $i\in [N]$ indicating which machine is running during the conducting of the protocol. Initially, processor \(1\) holds the token which contains the initial
configuration. In simulation round \(t\), the running processor \(i\)
holding the token reads its local bits and simulates step
\(t\) of \(M\). If the random-tape head stays at \(i\), the token
remains there; if the head moves to \(i\pm1\), the token carrying
the updated configuration is sent to that neighbor. If \(M\)
halts, the token remains at its current
processor and we further do a final sweep as the following.
Each processor checks if it has the token, if yes, send the token to the processor before it. Here the first processor will keep the token if it already has it. This takes $2N$ rounds of computation and communication.
Finally the protocol conduct a $1$-measurement at the 1st processor.

Each visit to a cell causes at most one incoming and one outgoing
token transfer. Hence each processor participates in \(O(p)\)
simulation transfers and \(O(1)\) transfers during the final sweep.
Since transfer times are adaptive, we use effective communication.
A processor's communication trace is determined by \(O(p)\)
nonempty transfers, each specified by its round, direction, and an
\(O(s_0)\)-bit token; all other messages are the fixed empty
message. Since \(\log(T+N)=O(s_0)\), the trace has an
\(O(ps_0)\)-bit encoding. Thus
\[
   c=O(ps_0),
   \qquad
   r=k=1.
\]

\paragraph{The partition tree.}
Use the canonical balanced interval tree of the path. The root
represents \([1,N]\), and a nonsingleton interval \([a,b]\) has
children
\[
   [a,m]
   \qquad\text{and}\qquad
   [m+1,b],
   \qquad
   m=\left\lfloor\frac{a+b}{2}\right\rfloor.
\]
The leaves are the singleton processors, and the tree has depth
\(O(\log(N+1))\).

For an internal node \(\nu\) representing \([a,b]\), the three-way
cut consists of the edge between its two children and at most two
path edges leaving \([a,b]\). Hence the tree has INW width at
most \(3\).

Identify each node by its binary root-to-node path. Following the
corresponding bisections recovers its interval in \(O(\log N)\)
workspace. Within the same bound, we can find its parent and
children, identify the processor at a leaf, and test membership in
\(V_\nu\). The interval endpoints also determine its boundary
edges and
\[
   J_\nu=
   \begin{cases}
      \{1\},&1\in V_\nu,\\
      \varnothing,&\text{otherwise}.
   \end{cases}
\]
Enumerating the valid node identifiers therefore produces the
required tree description in polynomial time and \(O(\log N)\)
workspace.

\paragraph{Evaluating the boundary tests.}
Given a node \(\nu\), a candidate boundary trace
\(\tau\in\mathcal I_\nu\), and \(z\in\{0,1\}^{J_\nu}\), simulate
the path protocol inside \(V_\nu\), using read-only access to its
local random bits.

If \(1\in V_\nu\), begin with the initial token; otherwise, begin
with no token. While the token lies in \(V_\nu\), simulate its
successive steps and check each outgoing transfer against the
corresponding record of \(\tau\). When the token is outside
\(V_\nu\), the next incoming record of \(\tau\) supplies it to the
appropriate processor. The test accepts exactly when the simulated
boundary communication agrees with \(\tau\) and, if
\(1\in V_\nu\), the final answer equals \(z\). This is the boundary
test defined in \Cref{def:inw-interface-view}.

The evaluator stores one configuration, the current round, the
interval endpoints, and a position in the trace. It therefore uses
\(O(s_0)\) workspace and time polynomial in \(n,T,N,s_0\).
Index the tests in a fixed lexicographic order on \((\tau,z)\),
so that \(L_\nu=|\mathcal I_\nu|2^{|J_\nu|}\).
Given \(i\in[L_\nu]\), the corresponding pair \((\tau,z)\) can be
identified within the same resource bounds. Hence the boundary
tests are
\[
   \DTISP\bigl(\poly(n,T,N,s_0),O(s_0)\bigr)\text{-evaluable}.
\]

\begin{proof}[Proof of \Cref{cor:read-p-inw-sc}]
We may assume \(N\le T+1\), since no cell beyond \(T+1\) can be
reached within \(T\) steps. Hence
\[
   s_0=O(S+\log(n+T)).
\]
Apply \Cref{thm:inw-resource-lifting} with
\[
   r=k=1,\qquad
   c=O(ps_0),\qquad
   \omega=O(1),\qquad
   h=O(\log N).
\]
Together with the boundary-test evaluator above, this gives the
stated bounds.

If \(S=O(\log n)\), \(T=\poly(n)\), and \(p=O(1)\), then
\(s_0=O(\log n)\). Taking \(\varepsilon\) to be a sufficiently
small constant gives polynomial running time and
\(O(\log^2 n)\) workspace, proving the final assertion.
\end{proof}

\subsection{SC derandomization of machines with oblivious stacks}
\label{subsec:inw-oblivious-stack}

\begin{definition}
For a probabilistic AuxPDM $M$, we say that \(M\) has a \emph{(randomness) oblivious stack} if for every
input \(x\), there is a stack-operation schedule
\[
   \sigma_x:[T]\to
   \{\mathsf{push},\mathsf{pop},\mathsf{idle}\}
\]
such that the operation at step \(t\) is always \(\sigma_x(t)\),
independently of the random choices.
\end{definition}

\begin{corollary}[SC derandomization of machines with oblivious stacks]
\label{cor:oblivious-stack-inw-sc}
Let \(M\) be a polynomial-time probabilistic auxiliary
pushdown machine using \(O(\log n)\) workspace on
\(n\)-bit inputs.
Suppose the stack of $M$ is randomness oblivious.

For every input and every \(0<\varepsilon<1/2\), its acceptance
probability can be deterministically approximated to additive
error \(\varepsilon\) in time \(\poly(n,1/\varepsilon)\)
and workspace
\[
   O\!\left(\log^2 n+\log n\log(1/\varepsilon)\right).
\]
In particular, languages decided with bounded error by such
machines belong to \(\SC^2\).
\end{corollary}

\paragraph{The model.}
We consider a standard form in which stack symbols are read only
when they are popped out. To obtain this form, cache the current stack top on
the work tape, initially storing a distinguished bottom marker.
A push moves the cached symbol onto the stack and stores the new
symbol in the cache. A pop uses the cached symbol as the popped
value and restores the cache by popping a symbol from the stack.
Inspections or updates of the current top affect only the cache.
This uses \(O(\log n)\) additional work bits and preserves the
push/pop/idle schedule.

We continue to denote the resulting machine by \(M\). It has a
read-only input, \(O(\log n)\) work bits, and an initially empty
stack. We allow stack symbols with \(O(\log n)\)-bit encodings.
On an input \(x\in\{0,1\}^n\), the machine runs for
\(T=\poly(n)\) steps and uses independent random blocks
\[
   \rho_1,\ldots,\rho_T\in\{0,1\}^r,
   \qquad r=O(1),
\]
one block at each step. A local configuration consists of the
finite control, work tape, and head positions, excluding the stack
contents.

At each step \(t\), the machine performs one of
\(\mathsf{push}\), \(\mathsf{pop}\), or \(\mathsf{idle}\).
A push writes one symbol, a pop reads and removes the top symbol,
and an idle step leaves the stack unchanged; no other operation
reads the stack. The next local configuration and any symbol
written are determined by \(x\), the current local configuration,
and \(\rho_t\), together with the popped symbol at a pop step.
The stack never underflows, and the output is a designated bit of
the local configuration after step \(T\).

The \(O(\log n)\)-bit stack symbols can be represented by
equal-length blocks over a fixed alphabet. Expanding each operation
into a fixed-length block of elementary steps, padding as necessary,
preserves polynomial running time, \(O(\log n)\) workspace, and
the schedule's independence of the random choices.

For a fixed input \(x\), the schedule \(\sigma_x\) uniquely matches
each pop step with the push whose symbol it removes. These
push--pop pairs are noncrossing in time order. Pushes whose symbols
remain on the stack after step \(T\) are unmatched.

\paragraph{The communication protocol.}
Fix \(x\) and its stack-operation schedule \(\sigma_x\).
We define a protocol \(\Pi_M\) on a graph \(H_x\) with processors
\(v_0,\ldots,v_T\). The graph has a time edge
\(\{v_{t-1},v_t\}\) for every \(t\in[T]\) and a matching edge
\(\{v_s,v_t\}\) for every matched push--pop pair \(s<t\).

The time edges carry the local configuration from one computation
step to the next, while each matching edge carries the symbol
written at a push to the pop that reads it. Since each step belongs
to at most one matching pair, \(H_x\) has maximum degree at most
three. The matching edges are noncrossing in the order
\(v_0,\ldots,v_T\).

Processor \(v_t\), for \(t\in[T]\), receives the private random
block \(\rho_t\). Processor \(v_0\) sends the initial local
configuration to \(v_1\). For each \(t\in[T]\), processor \(v_t\)
simulates step \(t\) using the configuration received from
\(v_{t-1}\), its random block \(\rho_t\), and, at a pop step,
the symbol received from the matching push. If \(t<T\), it sends
the resulting local configuration to \(v_{t+1}\). At a matched
push, it also sends the pushed symbol to the matching pop.
After simulating step \(T\), processor \(v_T\) outputs the answer.

Each processor exchanges only \(O(1)\) messages of
\(O(\log n)\) bits, at rounds determined by \(x\). Hence the
protocol has effective communication \(c=O(\log n)\), and the
output of \(v_T\) is a \(1\)-measurement.

\begin{lemma}[Efficiently evaluable partition trees for oblivious-stack computations]
\label{lem:oblivious-stack-interface-sc}
For every input \(x\), the graph \(H_x\) admits a partition tree
\(\mathcal T_x\) of depth \(O(\log T)\) and INW width at most
\(6\) such that \((\Pi_M,\mathcal T_x)\) has
\(\DTISP(\poly(n),O(\log^2 n))\)-evaluable boundary tests.
\end{lemma}

\begin{proof}
Fix \(x\). We first assume that the stack-operation schedule
\(\sigma_x\) is available on a read-only tape. Given a push or pop
step, we can find its matching partner, if one exists, by scanning
the schedule with a stack-height counter, using \(O(\log T)\)
workspace.

\emph{Partition tree.}
Identify processor \(v_t\) with its time index \(t\).
Call an interval \(I\subseteq[0,T]\) \emph{closed} if no matching
edge has exactly one endpoint in \(I\). We construct the tree so
that every nonleaf processor set has the form \(I\) or
\(I\setminus J\), where \(I\) and \(J\subsetneq I\) are closed
intervals. Such a set is a union of at most two intervals and does
not split any matching pair.

For processor sets of size at most three, use singleton leaves,
keeping the endpoints of any matching pair together until their
final separation. The recursive constructions below handle larger
sets.

We first describe how to split a closed interval \(K=[a,b]\).
If \(a\) and \(b\) form a matching pair, split \(K\) into
\[
   \{a,b\}
   \qquad\text{and}\qquad
   [a+1,b-1].
\]
Otherwise, let \(p\) be the matching partner of \(a\), or let
\(p=a\) if \(a\) is unmatched, and split \(K\) into
\[
   [a,p]
   \qquad\text{and}\qquad
   [p+1,b].
\]
Since the matching edges are noncrossing, each part is either a
closed interval or a matching pair.

Now let \(I\) be a closed interval of size \(m\ge4\).
Starting with \(K=I\), repeatedly replace \(K\) by the unique
part of its split having more than \(m/2\) vertices, as long as
such a part exists. Whenever such a part exists, it is a closed
interval. At termination, let \(P\) and \(Q\) be the two parts of
the split of \(K\). Then
\[
   |I\setminus K|<m/2,
   \qquad
   |P|,|Q|\le m/2.
\]
Make \(I\setminus K\) and \(K\) the two children of \(I\), and
make \(P\) and \(Q\) the two children of \(K\). If
\(I\setminus K=\varnothing\), simply split \(I=K\) into \(P\)
and \(Q\). Within two levels, every resulting subproblem thus has
size at most \(m/2\).

It remains to handle \(W=I\setminus J\), where \(I\) and
\(J\subsetneq I\) are closed intervals. Write \(J=[c,d]\) and
let \(m=|W|\ge4\). We first construct a nested sequence of closed
intervals
\[
   I=K_0\supsetneq K_1\supsetneq\cdots\supsetneq K_s=J
\]
such that every difference \(K_i\setminus K_{i+1}\) is either a
closed interval or a matching pair.

To construct the sequence, suppose that the current interval is
\(K=[a,b]\supsetneq J\). Let \(p\) be the matching partner of
\(a\), or let \(p=a\) if \(a\) is unmatched. If \(a=c\), take
\(K'=J\). Otherwise, \(a<c\), and the closedness of \(J\) implies
that \(p\notin J\). Set
\[
   K'=
   \begin{cases}
      [p+1,b],   & p<c,\\
      [a,p],     & d<p<b,\\
      [a+1,b-1], & p=b.
   \end{cases}
\]
In each case, \(K'\) is a proper closed subinterval of \(K\)
containing \(J\), while \(K\setminus K'\) is either a closed
interval or the matching pair \(\{a,b\}\). This gives the desired
sequence.

Choose \(i\) so that
\[
   |K_i\setminus J|>m/2,
   \qquad
   |K_{i+1}\setminus J|\le m/2.
\]
Let
\[
   A:=I\setminus K_i,\qquad
   D:=K_i\setminus K_{i+1},\qquad
   C:=K_{i+1}\setminus J.
\]
Then
\[
   W=A\cup D\cup C,
   \qquad
   |A|<m/2,\qquad |C|\le m/2.
\]
Make \(A\) and \(K_i\setminus J=D\sqcup C\) the two children of
\(W\), and make \(D\) and \(C\) the two children of
\(K_i\setminus J\), omitting empty sets and suppressing unary
nodes. If \(|D|\le m/2\), every resulting subproblem already has
size at most \(m/2\). Otherwise, \(D\) cannot be a matching pair,
so it is a closed interval. Applying the closed-interval
construction to \(D\) reduces it to subproblems of size at most
\(m/2\) within two further levels.

Consequently, along every root-to-leaf path, the size of the
current processor set decreases by a factor of two within at most
four levels. The resulting partition tree therefore has depth
\(O(\log T)\).

Every nonleaf processor set constructed above is a union of at
most two time intervals and has no matching edge leaving it.
Its boundary therefore contains only time edges, so
\[
   |\partial V_\nu|\le4.
\]
The same bound holds for a leaf because every processor has degree
at most three. If an internal node has children with processor sets
\(A\) and \(B\), then
\[
   |\cut(\nu)|
   =
   \frac{
      |\partial A|+|\partial B|+|\partial(A\cup B)|
   }{2}
   \le6.
\]
Thus the partition tree has INW width at most \(6\).

\emph{Evaluating a boundary test.}
Fix a node \(\nu\), a candidate boundary trace \(\tau\), and
\(z\in\{0,1\}^{J_\nu}\). We evaluate \(\chi_{\nu,\tau,z}\) by
simulating the execution inside \(V_\nu\) in chronological order
on \(v_0,\ldots,v_T\). Use \(\tau\) to supply incoming boundary
messages and check outgoing boundary messages against \(\tau\).

For a time edge whose endpoints both lie in \(V_\nu\), the
receiving processor uses the configuration produced by the
preceding simulated processor. For a time edge crossing the
boundary of \(V_\nu\), supply an incoming configuration from
\(\tau\) or check an outgoing configuration against \(\tau\),
as appropriate.

It remains to consider the matching edges. For every nonleaf node
\(\nu\), the set \(V_\nu\) does not split any matching pair.
Hence every matching edge incident to \(V_\nu\) has both endpoints
in \(V_\nu\). During the simulation, place the symbol written at
each matched push on an auxiliary stack; at the matching pop,
pop and use that symbol. Noncrossing ensures that these operations
are properly nested, so each pop retrieves the required symbol.
If \(V_\nu\) consists of two time intervals, retain the auxiliary
stack while skipping the intervening processors outside \(V_\nu\).

At a leaf, an incident matching edge may cross the boundary of
\(V_\nu\). Supply its incoming message from \(\tau\) or check its
outgoing message against \(\tau\), as appropriate.

If \(v_T\in V_\nu\), also check that the final output equals
\(z\). The simulation accepts exactly when its boundary
communication agrees with \(\tau\) and its local measurement
outcome equals \(z\). Its acceptance indicator is therefore
\(\chi_{\nu,\tau,z}\), as defined in
\Cref{def:inw-interface-view}.

The simulation runs in polynomial time and uses \(O(\log n)\)
workspace in addition to its auxiliary stack. Hence the boundary
test is decidable by a deterministic logspace auxiliary pushdown
automaton in polynomial time, and therefore lies in
\(\mathrm{LogDCFL}\) \cite{Sudborough1978}. By
\(\mathrm{LogDCFL}\subseteq \mathsf{SC}^{2}\)
\cite{Cook1979}, it can be evaluated in polynomial time and
\(O(\log^{2} n)\) workspace when \(\sigma_x\) is given explicitly.

\emph{Computing the stack schedule.}
For \(t\in[T]\), run \(M\) through step \(t\) with every random
block fixed to \(0^r\), and inspect the stack operation at that
step. By obliviousness, this operation is exactly \(\sigma_x(t)\).
The simulation runs in polynomial time and uses \(O(\log n)\)
workspace in addition to its stack. By the same auxiliary-pushdown
simulation theorem, \(\sigma_x(t)\) can therefore be computed in
polynomial time and \(O(\log^2 n)\) workspace.

It follows that the tree description, test indexing, and
boundary-test evaluation are all computable in
\(\DTISP(\poly(n),O(\log^2 n))\), as required.
\end{proof}

\begin{proof}[Proof of \Cref{cor:oblivious-stack-inw-sc}]
Apply \Cref{thm:inw-resource-lifting} with \(N=T+1\),
\(r=O(1)\), \(k=1\), \(\omega=O(1)\), \(c=O(\log n)\), and
\(h=O(\log n)\), using
\Cref{lem:oblivious-stack-interface-sc}. This gives the stated
approximation bounds. Taking \(\varepsilon\) to be a sufficiently
small constant proves the final assertion.
\end{proof}

\bibliographystyle{alpha}
\bibliography{ref}

\appendix
\section{From regular ROBPs to permutation ROBPs}
\label{sec:reg-to-perm-phase}

We convert a regular ROBP into a permutation ROBP over the same alphabet
by replacing each state with $|\Sigma|$ copies.  Summing over the copies
recovers the original average transition matrix.  We also give an
alphabet-reduction variant and a binary encoding of the construction.
Throughout this appendix, matrix rows index target states and columns
index source states, as in the preliminaries.

\subsection{Conversion over the same alphabet}
\label{subsec:reg2perm-phase-setup}
\label{subsec:reg2perm-phase-reduction}

Write $m:=|\Sigma|$ and identify $\Sigma$ with
$\mathbb Z_m=\{0,\ldots,m-1\}$.

\begin{lemma}[Regular-to-permutation conversion]
\label{lem:reg2perm-phase}
Let $B$ be a length-$n$, width-$w$ regular ROBP over a finite alphabet
$\Sigma$.  There is a length-$n$, width-$w|\Sigma|$ permutation ROBP
$B^\pi$ over $\Sigma$, with states $(v,\tau)\in[w]\times\Sigma$, such that
for every interval $[i,j]\subseteq[n]$, every $u,v\in[w]$, and every
$\tau\in\Sigma$,
\begin{equation}
   \sum_{\tau'\in\Sigma}
      [W^\pi_{i..j}]_{(u,\tau'),(v,\tau)}
   =[W_{i..j}]_{u,v}.
   \label{eq:reg2perm-interval-copies}
\end{equation}
In particular, choosing start state $(v_{\mathrm{start}},0)$ and accept
set $V_{\mathrm{accept}}\times\Sigma$ gives $\E B^\pi=\E B$.
The transition function of $B^\pi$ is computable in logspace relative to
the transition oracle of $B$.
\end{lemma}

\begin{proof}
Let $\delta_i(v,z)$ be the target of state $v$ under symbol $z$ in layer $i$.
For each target $u$, order its incoming edges
\[
   \Edge^i_u:=\{(v,z)\in[w]\times\Sigma:\delta_i(v,z)=u\}
\]
lexicographically.  Regularity gives $|\Edge^i_u|=m$, so the rank map
\[
   \rk^i_u:\Edge^i_u\longrightarrow\{0,\ldots,m-1\}
\]
is a bijection.  Define the transition of $B^\pi$ on symbol $x$ by
\begin{equation}
   \Pi^i_x(v,\tau)
   :=\bigl(\delta_i(v,z),\rk^i_{\delta_i(v,z)}(v,z)\bigr),
   \qquad z:=\tau+x\pmod m.
   \label{eq:reg2perm-transition}
\end{equation}

\begin{claim}
\label{cl:reg2perm-permutation}
For every $i\in[n]$ and $x\in\Sigma$, the map $\Pi^i_x$ is a permutation.
\end{claim}
\begin{proof}
The map $(v,\tau)\mapsto(v,\tau+x\bmod m)$ is a bijection.  The map
\[
   R^i(v,z):=\bigl(\delta_i(v,z),\rk^i_{\delta_i(v,z)}(v,z)\bigr)
\]
is also a bijection: each pair $(u,\tau')$ identifies the unique edge of
rank $\tau'$ entering $u$.  Their composition is $\Pi^i_x$.
\end{proof}

\begin{claim}
\label{cl:reg2perm-expectation}
The identity \eqref{eq:reg2perm-interval-copies} holds.
\end{claim}
\begin{proof}
Start at an arbitrary copy $(v,\tau)$ in layer $i-1$, and let
$X_i,\ldots,X_j$ be independent uniform symbols.  If the state before
step $t$ is $(V_{t-1},\Theta_{t-1})$, then
\[
   Z_t:=\Theta_{t-1}+X_t\pmod m,
   \qquad V_t=\delta_t(V_{t-1},Z_t).
\]
Conditioned on $X_i,\ldots,X_{t-1}$, the value $\Theta_{t-1}$ is fixed
and $Z_t$ is uniform on $\Sigma$.  Thus $Z_i,\ldots,Z_j$ are independent
uniform symbols, and the first coordinate follows the original program
from $v$.  The probability that this coordinate is $u$ at layer $j$ is
both sides of \eqref{eq:reg2perm-interval-copies}.
\end{proof}

\begin{claim}
\label{cl:reg2perm-explicit}
For a fixed layer $i$, one transition of $B^\pi$ can be evaluated in
$O(\log(w|\Sigma|))$ additional workspace using $O(w|\Sigma|)$ calls to
$\delta_i$.
\end{claim}
\begin{proof}
Compute $z=\tau+x\bmod m$ and $u=\delta_i(v,z)$.  To compute the rank
of $(v,z)$, enumerate the pairs $(v',z')$ preceding it in lexicographic
order and count those satisfying $\delta_i(v',z')=u$.  The enumeration
and counter use $O(\log(wm))$ bits.  Including the layer index requires
$O(\log n+\log(wm))$ bits in total, in addition to the workspace of the
original transition oracle.
\end{proof}

The three claims prove the lemma.
\end{proof}

\subsection{Binary encoding}
\label{subsec:reg2perm-binary}

For a power-of-two alphabet, the conversion can be implemented by reading
one bit of each input symbol at a time.

\begin{lemma}[Binary permutation conversion]
\label{lem:reg2perm-binary}
Let $B$ be a length-$n$, width-$w$ regular ROBP over an alphabet
$\Sigma$ of size $2^t$, where $t\ge1$.  There is a binary permutation
ROBP $B^{\mathrm{bin}}$ of length $nt$ and width $w2^t$ such that
$\E B^{\mathrm{bin}}=\E B$.  Its transition function is computable in
logspace relative to the transition oracle of $B$.
\end{lemma}

\begin{proof}
Identify $\Sigma$ with $\mathbb Z_{2^t}$, and write each symbol as
$x=\sum_{b=1}^t x^{(b)}2^{b-1}$.  Replace layer $i$ of the construction
\eqref{eq:reg2perm-transition} by $t$ binary layers on the same state set
$[w]\times\Sigma$.  On bit $y\in\{0,1\}$, the $b$-th transition is
\[
   U^{i,b}_y(v,\tau):=
   \begin{cases}
      (v,\tau+y2^{b-1}\bmod 2^t),& b<t,\\[2pt]
      R^i(v,\tau+y2^{t-1}\bmod 2^t),& b=t,
   \end{cases}
\]
where $R^i$ is the bijection defined in the proof of
\Cref{lem:reg2perm-phase}.  Each map is a permutation.  Composing the
$t$ transitions on the bits of $x$ adds $x$ to the copy index and then
applies $R^i$, giving exactly $\Pi^i_x$.

Use start state $(v_{\mathrm{start}},0)$ and accept set
$V_{\mathrm{accept}}\times\Sigma$.  Independent uniform input bits
give independent uniform symbols in successive blocks, so
$\E B^{\mathrm{bin}}=\E B^\pi=\E B$.  A binary transition requires
only modular addition and, at the last bit of a block, evaluation of
$R^i$.  The space bound follows from \Cref{cl:reg2perm-explicit}.
\end{proof}

\section{A convolution bound}
\label{sec:scaled-convolution-proof}

We prove \eqref{eq:scaled-convolution}.  For \(k\ge0\), recall that
\[
   A_k
   =
   \sum_{\substack{p,q\ge0\\p+q=k}}
      \varepsilon^{(p)}\varepsilon^{(q)}
   +
   \sum_{\substack{p,q\ge0\\p+q=k-1}}
      \varepsilon^{(p)}\varepsilon^{(q)},
   \qquad
   \varepsilon^{(h)}
   =
   \frac{\gamma^{h+1}}{100L(h+1)^2}.
\]
Here \(0<\gamma\le1\) and \(L\ge1\), and the second sum is empty
when \(k=0\).
Dividing by \(\varepsilon^{(k)}\) gives
\[
\frac{A_k}{\varepsilon^{(k)}}
=
\frac{1}{100L}
\left(
   \gamma(k+1)^2
   \sum_{p=0}^{k}
      \frac{1}{(p+1)^2(k-p+1)^2}
   +
   (k+1)^2
   \sum_{p=0}^{k-1}
      \frac{1}{(p+1)^2(k-p)^2}
\right).
\]

For every integer \(N\ge2\), partial fractions give
\[
\begin{aligned}
   N^2\sum_{a=1}^{N-1}\frac{1}{a^2(N-a)^2}
   &=
   2\sum_{a=1}^{N-1}\frac1{a^2}
   +
   \frac4N\sum_{a=1}^{N-1}\frac1a\\
   &\le 6,
\end{aligned}
\]
where we used
\[
   \sum_{a=1}^{N-1}\frac1{a^2}\le2,
   \qquad
   \sum_{a=1}^{N-1}\frac1a
      \le1+\frac{N-2}{2}=\frac N2.
\]
Apply this estimate with \(N=k+2\) to the first sum and, for
\(k\ge1\), with \(N=k+1\) to the second.  Since
\((k+1)^2\le(k+2)^2\), we obtain
\[
   \frac{A_k}{\varepsilon^{(k)}}
   \le
   \frac{6\gamma+6}{100L}
   \le
   \frac{12}{100L}
   \le
   \frac1{8L},
\]
as required.

\section{Deferred proofs for the forward--backward analysis}
\label{app:fb-technical-proofs}

We prove the weight and pathwise bounds used in
\Cref{sec:fb-inw-regular-permutation}.  Throughout, intervals
\(J=[i,j]\) index transitions from \(V_{i-1}\) to \(V_j\), as in the
preliminaries.

\subsection{Dissipation for doubly stochastic matrices}
\label{app:fb-dissipation-proof}

Recall that \(\Psi(q)=\sum_{a<b}|q_a-q_b|\).  The weight bounds follow
from the following inequality.

\begin{lemma}[Doubly stochastic dissipation]
\label{lem:fb-dissipation}
Let \(D\) be a random \(w\times w\) matrix with entries in \(\{0,1\}\)
and exactly one \(1\) in each row.  If \(A:=\E D\) is doubly stochastic,
then, for every \(q\in\mathbb R^w\),
\begin{equation}
   \E\|(D-A)q\|_1
   \le 2\bigl(\Psi(q)-\Psi(Aq)\bigr).
   \label{eq:fb-dissipation}
\end{equation}
If \(D\) is always a permutation matrix, the same inequality holds with
\(D^\top,A^\top\) in place of \(D,A\).
\end{lemma}

\begin{proof}
We first show that, for any permutation matrices \(P,Q\),
\begin{equation}
   \|Pq-Qq\|_1
   \le 2\left(
      \Psi(q)-\Psi\!\left(\frac{Pq+Qq}{2}\right)
   \right).                                                       \label{eq:fb-midpoint-claim}
\end{equation}
Relabeling coordinates, assume \(q_1\le\cdots\le q_w\) and write
\[
   q=q_1\one+
     \sum_{s=1}^{w-1}(q_{s+1}-q_s)e^{(s)},
   \qquad
   e^{(s)}=\one_{\{s+1,\ldots,w\}}.
\]
Each difference \(q_b-q_a\), for \(a<b\), is the sum of the intervening
gaps, so
\[
   \Psi(q)=
   \sum_{s=1}^{w-1}(q_{s+1}-q_s)\Psi(e^{(s)}).
\]
For one of these indicator vectors \(e\), write \(Pe=\one_S\),
\(Qe=\one_T\), and \(k=|S|=|T|\).  Put
\(r=|S\setminus T|=|T\setminus S|\).  The midpoint
\((Pe+Qe)/2\) has \(k-r\) coordinates equal to \(1\), \(2r\) equal to
\(1/2\), and \(w-k-r\) equal to \(0\).  Hence
\[
   \|Pe-Qe\|_1=2r,
   \qquad
   \Psi(e)-\Psi\!\left(\frac{Pe+Qe}{2}\right)=r^2.
\]
Since \(r\) is a nonnegative integer, \(r\le r^2\), proving
\eqref{eq:fb-midpoint-claim} for \(e\).  By the triangle inequality,
\begin{align*}
   \|Pq-Qq\|_1
   &\le\sum_{s=1}^{w-1}(q_{s+1}-q_s)
      \|Pe^{(s)}-Qe^{(s)}\|_1\\
   &\le 2\Psi(q)
      -2\sum_{s=1}^{w-1}(q_{s+1}-q_s)
         \Psi\!\left(\frac{Pe^{(s)}+Qe^{(s)}}2\right)\\
   &\le 2\left(\Psi(q)-\Psi\!\left(\frac{Pq+Qq}2\right)\right).
\end{align*}
The last step uses that \(\Psi\) is subadditive, positively homogeneous,
and unchanged by adding a constant vector.

Let $P$ be a random permutation matrix with $\E P=A$, and let $P'$ be an
independent copy.  Such a distribution exists by the Birkhoff--von Neumann
theorem.  Jensen's inequality, \eqref{eq:fb-midpoint-claim}, and convexity of
$\Psi$ imply
\begin{align}
   \E\|(P-A)q\|_1
   &\le \E\|Pq-P'q\|_1 \notag\\
   &\le 2\left(
      \Psi(q)-\E\Psi\!\left(\frac{Pq+P'q}{2}\right)
   \right) \notag\\
   &\le 2\bigl(\Psi(q)-\Psi(Aq)\bigr).                \label{eq:fb-perm-barycenter}
\end{align}

For both \(D\) and \(P\), the unique \(1\) in row \(i\) lies in
column \(j\) with probability \(A_{ij}\).  Therefore
\[
   \E\|(D-A)q\|_1
   =\sum_{i,j}A_{ij}|q_j-(Aq)_i|
   =\E\|(P-A)q\|_1.
\]
Together with \eqref{eq:fb-perm-barycenter}, this proves
\eqref{eq:fb-dissipation}.  If \(D\) is always a permutation matrix,
then \(D^\top\) also satisfies the hypotheses, giving the transposed
inequality.
\end{proof}

\subsection{Forward and backward weight bounds}
\label{app:fb-weight-bounds-proof}

\begin{proof}[Proof of \Cref{lem:fb-weight-bounds}]
Fix \(J=[i,j]\).  For the forward propagation
\(p_{i-1}=p\), \(p_t=W_t p_{t-1}\), apply
\Cref{lem:fb-dissipation} to \(B_t(\sigma)\), with
\(\sigma\) uniform in \(\Sigma\).  Since its mean is \(W_t\),
\[
   \E_{\sigma\in\Sigma}\|(B_t(\sigma)-W_t)p_{t-1}\|_1
   \le
   2\bigl(\Psi(p_{t-1})-\Psi(p_t)\bigr).
\]
The definition of \(\FBF_J\) and telescoping give
\begin{align*}
   \FBF_J(p)
   &\le4\sum_{t=i}^{j}
      \bigl(\Psi(p_{t-1})-\Psi(p_t)\bigr)\\
   &=4\bigl(\Psi(p)-\Psi(W_{i..j}p)\bigr).
\end{align*}

For the backward propagation \(q_j=q\), \(q_{t-1}=W_t^\top q_t\),
apply the transposed inequality to obtain
\[
   \E_{\sigma\in\Sigma}\|(B_t(\sigma)-W_t)^\top q_t\|_1
   \le
   2\bigl(\Psi(q_t)-\Psi(q_{t-1})\bigr).
\]
Thus
\begin{align*}
   \FBB_J(q)
   &\le4\sum_{t=i}^{j}
      \bigl(\Psi(q_t)-\Psi(q_{t-1})\bigr)\\
   &=4\bigl(\Psi(q)-\Psi(W_{i..j}^\top q)\bigr).
\end{align*}
The nonnegativity of $\Psi$ proves the second inequalities in
\eqref{eq:fb-weight-F} and \eqref{eq:fb-weight-B}.  Finally,
\[
   \Psi(v)
   =\sum_{x<y}|v_x-v_y|
   \le\sum_{x<y}(|v_x|+|v_y|)
   =(w-1)\|v\|_1
   \le w\|v\|_1,
\]
which proves \eqref{eq:fb-weight-budget}.
\end{proof}

\subsection{Bounds for transitions with fixed symbols}
\label{app:fb-pathwise-bounds-proof}

\begin{proof}[Proof of \Cref{lem:fb-pathwise-bounds}]
Fix \(J=[i,j]\) and \(\sigma_J\in\Sigma^{j-i+1}\), and put
\[
   U:=B_J(\sigma_J)=B_j(\sigma_j)\cdots B_i(\sigma_i).
\]
For the forward bound, let \(r_{i-1}=p\) and
\(r_t=B_t(\sigma_t)r_{t-1}\) for \(i\le t\le j\).  Compare this
path with the average propagation \(p_{i-1}=p\),
\(p_t=W_t p_{t-1}\).  Since permutation matrices preserve the
\(\ell_1\) norm, the definition of \(\Lambda_F\) gives
\begin{align*}
   \|r_t-p_t\|_1
   &\le
     \|B_t(\sigma_t)(r_{t-1}-p_{t-1})\|_1
     +\|(B_t(\sigma_t)-W_t)p_{t-1}\|_1\\
   &\le\|r_{t-1}-p_{t-1}\|_1
     +\Lambda_F\E_{\sigma\in\Sigma}
        \|(B_t(\sigma)-W_t)p_{t-1}\|_1.
\end{align*}
Summing over \(t=i,\ldots,j\), and using
\(r_j=Up\) and \(p_j=W_{i..j}p\), yields
\[
   \|Up-W_{i..j}p\|_1
   \le \Lambda_F\sum_{t=i}^{j}
      \E_{\sigma\in\Sigma}\|(B_t(\sigma)-W_t)p_{t-1}\|_1
   =\frac{\Lambda_F}{2}\FBF_J(p).
\]

For the backward bound, let \(s_j=q\) and
\(s_{t-1}=B_t(\sigma_t)^\top s_t\), and compare it with
\(q_j=q\), \(q_{t-1}=W_t^\top q_t\).  Similarly,
\begin{align*}
   \|s_{t-1}-q_{t-1}\|_1
   &\le
     \|B_t(\sigma_t)^\top(s_t-q_t)\|_1
     +\|(B_t(\sigma_t)-W_t)^\top q_t\|_1\\
   &\le \|s_t-q_t\|_1
     +\Lambda_B\E_{\sigma\in\Sigma}
        \|(B_t(\sigma)-W_t)^\top q_t\|_1.
\end{align*}
Summing from \(t=j\) down to \(i\), with
\(s_{i-1}=U^\top q\) and \(q_{i-1}=W_{i..j}^\top q\), gives
\[
   \|U^\top q-W_{i..j}^\top q\|_1
   \le \Lambda_B\sum_{t=i}^{j}
      \E_{\sigma\in\Sigma}\|(B_t(\sigma)-W_t)^\top q_t\|_1
   =\frac{\Lambda_B}{2}\FBB_J(q),
\]
as required.
\end{proof}

\section{Strongly-few reachability decidable by deterministic logspace AuxPDM in polynomial time}
\label{app:stcon-sf}

We use the definition of $\mathrm{STCON}_{\mathrm{sf}}$ from
\cite{ApersEdenhofer2025}. For a directed graph $G=(V,E)$, let
$N_G(u,v)$ denote the number of finite directed walks from $u$ to $v$,
including the length-zero walk when $u=v$. Let $N_G(u,v)=\infty$ when such number is infinite. Then
\[
  \mathrm{STCON}_{\mathrm{sf}}
  =\left\{\langle G,s,t,1^k\rangle:
    N_G(u,v)\le k\text{ for every }u,v\in V,
    \quad N_G(s,t)\ge 1\right\}.
\]
The graph is given explicitly, and the bound $k$ is given in unary.

\begin{proposition}
\label{prop:stcon-sf-auxpdm}
The language $\mathrm{STCON}_{\mathrm{sf}}$ is decidable by a
deterministic AuxPDM using $O(\log N)$ work space and $N^{O(1)}$ time,
where $N$ is the input length. Consequently,
\[
  \mathrm{STCON}_{\mathrm{sf}}
  \in \mathsf{LogDCFL}\subseteq\mathsf{SC}^{2}.
\]
\end{proposition}

\begin{proof}
Write $V=[n]$. Assume $k\ge 1$.

\paragraph{The walk tree.}
For each $u\in V$, consider the rooted tree $\mathcal T_u$ whose nodes
are all finite directed walks starting at $u$. Its root is the
length-zero walk, and the children of a walk are its
one-step progressions. Order these
children by their final vertices. Different walks ending at the same
vertex of $G$ are distinct nodes of $\mathcal T_u$.
It satisfies the bound
\begin{equation}
\label{eq:stcon-sf-tree-size}
  |\mathcal T_u|=\sum_{v\in V}N_G(u,v)\le nk
\end{equation}
on every acceptable input graph.

\paragraph{A bounded depth-first traversal.}
Fix an ordered pair $(u,v)$. We traverse $\mathcal T_u$ in depth-first
order, maintaining two counters on the work tape. The counter $a$
counts all tree nodes entered, and $b$ counts those whose walks end
at $v$. Both counters start at zero and are incremented only on the
first entry into a tree node, including the root. Reject the input
immediately if
\[
  a>nk \qquad\text{or}\qquad b>k.
\]
Otherwise, continue the traversal; if it finishes, return $b$. Note that a root that can reach a
directed cycle has an infinite walk tree, and its traversal must
reach one of the cutoffs rather than run indefinitely.

Run this procedure for every ordered pair $(u,v)\in V^2$, reusing the
work tape and the auxiliary stack.  Accept if every traversal
ends without early termination and a path from $s$ to $t$ is found during the process, i.e. the procedure run with pair $(s,t)$ returns $b\geq1$.
\paragraph{Correctness.}
Each node counted by $a$ represents a distinct walk starting at $u$.
If $a>nk$, the pigeonhole principle implies that more than $k$ of
these walks end at some common vertex. Hence the bound on the tree-size is
violated, regardless of the particular target $v$ being tested.
Likewise, $b>k$ directly witnesses $N_G(u,v)>k$. Both cutoff rules
therefore reject only inputs outside the language.

If $N_G(u,v)\le k\text{ for every }u,v\in V$, \eqref{eq:stcon-sf-tree-size} guarantees
that every traversal finishes without reaching either cutoff. In such case, the algorithm accepts
exactly when $N_G(s,t)\ge 1$.
Conversely, if the algorithm accepts, every traversal has completed
and returned the exact count $N_G(u,v)\le k$, and the count for
$(s,t)$ is positive.

\paragraph{Time and work space.}
Each traversal enters at most $nk+1$ tree nodes before finishing or
rejecting. Each node requires at most $n$ candidate-neighbor tests,
and each tree edge is traversed at most once in each direction.
A candidate-neighbor test is performed by scanning the explicit
input graph, using logarithmic work space and polynomial time.
Since there are $n^2$ traversals and $k$ is encoded in unary, the
total running time is polynomial in $N$ on every input.

A depth-first search frame contains only the current endpoint and
the index of the next candidate neighbor, requiring $O(\log n)$
bits. The active frame is kept on the work tape; suspended frames
are pushed onto the auxiliary stack and recovered on return.
These frames can be encoded over a fixed finite stack alphabet,
with each frame push or pop implemented by $O(\log n)$ single-symbol
operations. The counters, loop indices, and active frame occupy
$O(\log(nk))=O(\log N)$ work space. The stack has polynomial height
and is not charged to this work-space bound.

This proves the claimed deterministic AuxPDM bound.
Sudborough's characterization of polynomial-time deterministic
logspace auxiliary pushdown computation \cite{Sudborough1978} and
Cook's simulation \cite{Cook1979} give
$\mathsf{LogDCFL}\subseteq\mathsf{SC}^{2}$ and establish the final
containment.
\end{proof}

\end{document}